\documentclass[dvipsnames,11pt]{article}
\usepackage{amsmath,amssymb,amsthm,color,bm,mathrsfs,extarrows,tikz,graphicx,mathtools,enumitem,fancybox,makecell}
\usepackage[square,sort,comma,numbers]{natbib}
\usepackage{subcaption}
\usepackage[ruled]{algorithm2e}
\usepackage[margin=1.in]{geometry}
\usepackage{xcolor,bbm}
\usepackage[pagebackref]{hyperref}
\usepackage{comment}
\allowdisplaybreaks

\setlist[itemize]{itemsep=0pt}
\setlist[enumerate]{itemsep=0pt}
\hypersetup{colorlinks=true,urlcolor=blue,linkcolor=blue,citecolor=[rgb]{.42,.56,.14},}
\usetikzlibrary{decorations.pathreplacing,decorations.markings,decorations.pathmorphing,decorations.shapes,arrows.meta,positioning}

\SetKwProg{Procedure}{Procedure}{:}{end}
\SetKwFunction{AtomicCSPSampling}{AtomicCSPSampling}
\SetKwFunction{RejectionSampling}{RejectionSampling}
\SetKwFunction{SafeSampling}{SafeSampling}
\SetKwFunction{FinalSampling}{FinalSampling}
\SetKwFunction{BoundingChain}{BoundingChain}
\SetKwFunction{SystematicScan}{SystematicScan}
\SetKwFunction{Component}{Component}
\SetKwFunction{FailedConstraints}{FailedConstraints}

\usepackage[capitalise,nameinlink]{cleveref}
\Crefname{lemma}{Lemma}{Lemmas}
\Crefname{fact}{Fact}{Facts}
\Crefname{theorem}{Theorem}{Theorems}
\Crefname{corollary}{Corollary}{Corollaries}
\Crefname{claim}{Claim}{Claims}
\Crefname{example}{Example}{Examples}
\Crefname{problem}{Problem}{Problems}
\Crefname{definition}{Definition}{Definitions}
\Crefname{notation}{Notation}{Notations}
\Crefname{assumption}{Assumption}{Assumptions}
\Crefname{subsection}{Subsection}{Subsections}
\Crefname{section}{Section}{Sections}
\Crefformat{equation}{(#2#1#3)}

\newtheorem{theorem}{Theorem}[section]
\newtheorem*{theorem*}{Theorem}

\newtheorem{proposition}[theorem]{Proposition}
\newtheorem*{proposition*}{Proposition}
\newtheorem{lemma}[theorem]{Lemma}
\newtheorem*{lemma*}{Lemma}
\newtheorem{corollary}[theorem]{Corollary}
\newtheorem*{corollary*}{Corollary}
\newtheorem*{conjecture*}{Conjecture}
\newtheorem{fact}[theorem]{Fact}
\newtheorem*{fact*}{Fact}

\newtheorem*{exercise*}{Exercise}

\newtheorem*{hypothesis*}{Hypothesis}

\theoremstyle{definition}
\newtheorem{definition}[theorem]{Definition}

\newtheorem{exercise-easy}[theorem]{Exercise}
\newtheorem{exercise-med}[theorem]{Exercise}
\newtheorem{exercise-hard}[theorem]{Exercise$^\star$}

\newtheorem*{claim*}{Claim}

\newtheorem{remark}[theorem]{Remark}
\newtheorem*{remark*}{Remark}

\newtheorem*{observation*}{Observation}

\DeclareMathOperator*{\sd}{\mathrm d}
\DeclareMathOperator*{\E}{\mathbb E}

\renewcommand{\Pr}{\operatorname*{\mathbf{Pr}}}

\newcommand{\eps}{\varepsilon}
\newcommand{\abs}[1]{\left| #1 \right|}

\newcommand{\pbra}[1]{\left( #1 \right)}
\newcommand{\sbra}[1]{\left[ #1 \right]}
\newcommand{\cbra}[1]{\left\{ #1 \right\}}
\renewcommand{\mid}{\,\middle\vert\,}
\newcommand{\bin}{\{0,1\}}

\newcommand{\binTF}{\{\mathsf{True},\mathsf{False}\}}
\newcommand{\True}{\mathsf{True}}
\newcommand{\False}{\mathsf{False}}
\newcommand{\Qmark}{\star}

\newcommand{\dist}{\mathsf{dist}}
\newcommand{\poly}{\mathsf{poly}}

\newcommand{\vbl}{\mathsf{vbl}}

\newcommand{\indicator}{\mathsf{1}}

\newcommand{\Final}{\mathsf{Final}}

\newcommand{\sigmain}{{\sigma_\mathsf{in}}}
\newcommand{\sigmaout}{\sigma_\mathsf{out}}
\newcommand{\Lin}{\mathsf{Lin}}
\newcommand{\Token}{\mathsf{Token}}
\newcommand{\ext}{\mathsf{ext}}
\newcommand{\UpdTime}{\mathsf{UpdTime}}

\newcommand{\leafs}{\mathsf{leafs}}

\newcommand{\tpath}{\mathsf{path}}
\newcommand{\Trans}{\mathsf{Trans}}
\newcommand{\childs}{\mathsf{childs}}

\newcommand{\KL}{\mathsf{KL}}
\newcommand{\Glauber}{P_\textsf{Glauber}}
\newcommand{\BChains}{P_\textsf{BChains}}

\newcommand{\Naturale}{\mathbf{e}}

\newcommand{\pibm}{{\bm{\pi}}}

\newcommand{\Rbb}{\mathbb{R}}
\newcommand{\Zbb}{\mathbb{Z}}

\newcommand{\Ccal}{\mathcal{C}}
\newcommand{\Dcal}{\mathcal{D}}
\newcommand{\Ecal}{\mathcal{E}}

\newcommand{\Mcal}{\mathcal{M}}

\newcommand{\Pcal}{\mathcal{P}}
\newcommand{\Scal}{\mathcal{S}}
\newcommand{\Tcal}{\mathcal{T}}
\newcommand{\Ucal}{\mathcal{U}}

\newcommand{\Xcal}{\mathcal{X}}

\newcommand{\Fsf}{\mathsf{F}}
\newcommand{\Psf}{\mathsf{P}}

\renewcommand{\tilde}{\widetilde}

\renewcommand{\hat}{\widehat}

\title{Perfect Sampling for (Atomic) Lov\'asz Local Lemma}
\author{
Kun He\thanks{Institute of Computing Technology, Chinese Academy of Sciences. Email: \texttt{hekun.threebody@foxmail.com}}
\and
Xiaoming Sun\thanks{Institute of Computing Technology, Chinese Academy of Sciences. Email: \texttt{sunxiaoming@ict.ac.cn}}
\and
Kewen Wu\thanks{Department of EECS, University of California at Berkeley. Email: \texttt{shlw\_kevin@hotmail.com}}
}
\date{}

\begin{document}
\maketitle

\begin{abstract}
We give a Markov chain based perfect sampler for uniform sampling solutions of constraint satisfaction problems (CSP).
Under some mild Lov\'asz local lemma conditions where each constraint of the CSP has a small number of forbidden local configurations, our algorithm is accurate and efficient: it outputs a \emph{perfect} uniform random solution and its expected running time is \emph{quasilinear} in the number of variables.
Prior to our work, perfect samplers are only shown to exist for CSPs under much more restrictive conditions (Guo, Jerrum, and Liu, JACM'19).

Our algorithm has two components: 
\begin{itemize}
\item A simple perfect sampling algorithm using \emph{bounding chains} (Huber, STOC'98; Haggstrom and Nelander, Scandinavian Journal of Statistics'99). 

This sampler is efficient if each variable domain is small. 
\item A simple but powerful \emph{state tensorization} trick to reduce large domains to smaller ones. 

This trick is a generalization of \emph{state compression} (Feng, He, and Yin, STOC'21).
\end{itemize}
The crux of our analysis is a simple \emph{information percolation} argument which allows us to achieve bounds even beyond current best approximate samplers (Jain, Pham, and Vuong, ArXiv'21).

Previous related works either use intricate algorithms or need sophisticated analysis or even both. Thus we view the simplicity of \emph{both} our algorithm and analysis as a strength of our work.
\end{abstract}

\section{Introduction}\label{sec:introduction}
The \emph{constraint satisfaction problem} (CSP) is one of the most important topics in computer science (both theoretically and practically). A CSP is a collection of constraints defined on a set of variables, and a solution to the CSP is an assignment of variables that satisfies all the constraints.
For any given CSP, it is natural to ask the following questions:
\begin{itemize}
\item \textsc{Decision.} Can we decide efficiently if the CSP \emph{has} a solution?
\item \textsc{Search.} If the CSP is satisfiable, can we \emph{find} a solution efficiently?
\item \textsc{Sampling.} If we can efficiently find a solution, can we efficiently \emph{sample} a uniform random solution from the whole solution space?
\end{itemize}
These questions, each deepening one above, progressively enhance our understanding on the computational complexity of CSPs.
One can easily imagine the hardness of fully resolving these broad questions.
Thus, not surprisingly, despite enormous results centered around them, we only have partial answers. 
Here we mention those related to our work.

\paragraph*{The Decision Problem.} 
A fundamental criterion for the existence of solutions is given by the famous Lov{\'a}sz local lemma (LLL)~\cite{EL75}.
Interpreting the space of all possible assignments as a probability space and the violation of each constraint as a bad event, the local lemma provides a sufficient condition for the existence of an assignment to avoid all the bad events.
This sufficient condition, commonly referred to as the \emph{local lemma regime}, is characterized in terms of the violation probability of each constraint and the dependency relation among the constraints.

\paragraph*{The Search Problem.} 
The \emph{algorithmic LLL} (also called \emph{constructive LLL}) provides efficient algorithms to find a solution in the local lemma regime.
Plenty of works have been devoted to this topic \cite{beck1991algorithmic,alon1991parallel,molloy1998further,CS00, moser2009constructive,moser2010constructive,Kolipaka2011MoserAT,haeupler2011new,HS17,HS19}.
The Moser-Tardos algorithm \cite{moser2010constructive} is a milestone along this line: it finds a solution efficiently up to a sharp condition known as the Shearer’s bound~\cite{shearer85,Kolipaka2011MoserAT}.

\paragraph*{The Sampling Problem.} 
The \emph{sampling LLL} asks for efficient algorithms to sample a uniform random solution from all solutions in the local lemma regime. 
It serves as a standard toolkit for the probabilistic inference problem in graphical models~\cite{Moi19}, and has many applications in the theory of computing, such as all-terminal network reliability~\cite{GJL19,GJ19a,guo2018tight}. With a better understanding of the decision and search problem, much attention has been devoted to the sampling LLL in recent years~\cite{GJL19,Moi19,guo2019counting,galanis2019counting,FGYZ20,feng2020sampling,Vishesh20towards,Vishesh21sampling}. Since this is also our focus, we elaborate it in the next subsection.

\subsection{Sampling Lov{\'a}sz Local Lemma}\label{sec:sampling_lll}

To state the long list of results on sampling LLL, we need the following notations.
Given a CSP, let $n$ be the number of variables, 
$k$ be the maximum number of variables in each constraint,
$Q$ be the maximum number of values that each variable can take,
$\Delta$ be the maximum constraint degree, 
$p$ be the maximum violation probability of a constraint, 
and $N$ be the maximum number of forbidden local configurations for each constraint.
A constraint is called \emph{atomic} if it only has one forbidden local configuration.
A CSP is called atomic if all of its constraints are atomic, i.e., $N=1$.

For example in Boolean $k$-CNF formula, each constraint will depend on exactly $k$ Boolean variables and thus $Q=2,p=2^{-k},N=1$.
For hypergraph coloring, each vertex is allowed to choose a color from $\cbra{1,2,\ldots,Q}$ and each edge contains exactly $k$ vertices. The constraints require no edge is monochromatic. Therefore $N=Q$, $p=Q^{1-k}$, and $\Delta$ equals the maximum edge degree in the hypergraph.
For general CSPs, each constraint may depend on different number of variables and each variable may have different domain size.

The sampling LLL turns out to be computationally more challenging than the algorithmic LLL.
For example, for $k$-CNF the Moser-Tardos algorithm can efficiently find a solution if $\Delta\lesssim2^k$ where $\gtrsim$ informally hides lower order terms.
However, it is intractable to approximately sample a uniform solution if $\Delta\gtrsim2^{k/2}$, even when the formula is monotone~\cite{BGGGS19}.

On the algorithmic side, most efforts are on the \emph{approximate} sampling, where the output distribution is close to uniform under \emph{total variation distance}.
The breakthrough of Moitra~\cite{Moi19} shows $k$-CNF solutions can be sampled in time $n^{\poly(k\Delta)}$ if $\Delta\lesssim2^{k/60}$, where they novelly use the algorithmic LLL to \emph{mark/unmark} variables and then convert the problem into solving linear programs of size $n^{\poly(k\Delta)}$. We remark that this algorithm is deterministic if we only need a multiplicative approximation of the number of solutions, which is another topic closely related with approximate sampling \cite{JVV86}.
Moitra's method has been successfully applied to hypergraph colorings~\cite{guo2019counting} and random CNF formulas~\cite{galanis2019counting}\footnote{\cite{galanis2019counting} only provides an approximate counting algorithm for random CNF formulas. But with a close inspection, their algorithm can be turned to do approximate sampling. This follows from standard reductions and noticing fixing \emph{bad variables} (defined in \cite{galanis2019counting}) does not influence their (deterministic) algorithm.}. 

Recently, a much faster algorithm for sampling solutions of $k$-CNF is given in~\cite{FGYZ20}, which implements a Markov chain on the assignments of the marked variables chosen via Moitra's method.
The resulting sampling algorithm has a near linear running time $\widetilde{O}\pbra{n^{1.001}}$ together with an improved regime $\Delta\lesssim2^{k/20}$, where $\widetilde{O}$ hides $\poly(N,k,\Delta,Q,\log(n))$.
We also remark that this algorithm is inherently randomized even if we move to approximate counting.  

This \emph{nonadapive} mark/unmark approach seems to only work for the Boolean variables, where each variable has two possible values.
To extend the approach to general CSPs, Feng, He, and Yin~\cite{feng2020sampling} introduced \emph{states compression}, which considerably expands the applicability of the method used in~\cite{FGYZ20}. 
Their sampling algorithm runs in time $\widetilde{O}\pbra{n^{1.001}}$ if $p^{1/350}\cdot\Delta\lesssim 1/N$. 
This algorithm is limited to the special cases of CSPs where each constraint is violated by a small number of local configurations (i.e., $N$ is small).

Recently, Jain, Pham, and Vuong~\cite{Vishesh20towards}, shaving the dependency on $N$, provides a sampling algorithm with running time $n^{\poly(\Delta,k,\log(Q))}$ when $p^{1/7}\cdot\Delta\lesssim 1$. They revisit Moitra's mark/unmark framework and use it in an \emph{adaptive} way. This is the first polynomial time algorithm (assuming $\Delta,k,Q=O(1)$) for general CSPs in local lemma regime.
By a highly sophisticated information percolation argument, they~\cite{Vishesh21sampling} also prove that the sampling algorithm in~\cite{feng2020sampling} runs in time $\widetilde{O}\pbra{n^{1.001}}$ if $p^{0.142}\cdot\Delta\lesssim 1/N$.

\begin{table}[ht]
\centering
\begin{tabular}{| l | l | l | l | l |}
		\hline
		Method & $k$-CNF & Hypergraph Coloring & General CSPs & Time \\ 
		\hline
		\cite{Moi19} & $\Delta\lesssim2^{k/60}$ &  &  & $n^{\poly(k\Delta)}$\\
		\cite{guo2019counting} &  & $\Delta\lesssim Q^{k/16}$ &  & $n^{\poly(k\Delta\log(Q))}$\\
		\cite{FGYZ20} & $\Delta\lesssim2^{k/20}$ &  &  & $\widetilde{O}\pbra{n^{1.001}}$ \\
		\cite{feng2020sampling} & $\Delta\lesssim2^{k/13}$ & $\Delta\lesssim Q^{k/9}$ & $p^{1/350}\cdot\Delta\lesssim1/N$ & $\widetilde{O}\pbra{n^{1.001}}$\\
		\cite{Vishesh20towards} & $\Delta\lesssim2^{k/7}$ &  $\Delta\lesssim Q^{k/7}$ & $p^{1/7}\cdot\Delta\lesssim1$ &$n^{\poly(k\Delta\log(Q))}$\\
		\cite{Vishesh21sampling} & $\Delta\lesssim2^{0.175k}$ &  $\Delta\lesssim Q^{k/3}$ & $p^{0.142}\cdot\Delta\lesssim1/N$ & $\widetilde{O}\pbra{n^{1.001}}$\\
		\hline
\end{tabular}
\caption{Approximate sampling algorithms in the local lemma regime.}\label{tab:regime}
\end{table}

\Cref{tab:regime} summarizes the efficient regimes of these algorithms.
We emphasize that all these sampling results, via standard reductions~\cite{jerrum1986random,vstefankovivc2009adaptive}, also imply efficient algorithms for (random) approximate counting, which estimates the number of solutions within some multiplicative error.
In addition, for algorithms using Moitra's linear programming approach~\cite{Moi19,guo2019counting,galanis2019counting,Vishesh20towards}, their approximate counting counterparts are deterministic.
For the approaches using Markov chains~\cite{FGYZ20,feng2020sampling,Vishesh21sampling}, the running time of their approximate counting counterparts is $\widetilde{O}(m\cdot T)$, where $T$ is the running time of the corresponding approximate sampling algorithm and $m$ is the number of constraints.

Though much progress has been made for the \emph{approximate} sampling, much less are known for the \emph{perfect} sampling.
As far as we know, the only result on the perfect sampling in the local lemma regime is due to Guo, Jerrum, and Liu~\cite{GJL19}, which provides a perfect sampler for the \emph{extremal} CSPs where any two constraints sharing common variables cannot be violated simultaneously by the same assignment. 
Though there are known reductions from approximate sampling/counting to perfect sampling~\cite{jerrum1986random}, it is unclear how to adopt them here considering the local lemma conditions.

Meanwhile, perfect sampling is an important topic in theoretical computer science. 
Plenty of works have been devoted to the study of perfect samplers \cite{jerrum1986random,haggstrom1999exact,huber1998exact,huber2004perfect,Bhandari020Improve,Vishesh2020Perfectly,fill1997interruptible,fill2000extension,abraham2012fully,FVY19}. 
Apart from its mathematical interest, one advantage of perfect sampler over approximate sampler is that the quality of the output is never in question. In contrast, some solution may never be outputted by an approximate sampler. Consider the following simple example: Let $\Dcal_1$ and $\Dcal_2$ be two distributions where $\Dcal_1$ is uniform over $\cbra{1,2,\ldots,n}$ and $\Dcal_2$ is uniform over $\cbra{\sqrt n+1,\sqrt n+2,\ldots,n}$. 
Then the total variation distance between $\Dcal_1$ and $\Dcal_2$ is only $1/\sqrt n=o(1)$. Thus $\Dcal_2$ is considered a good approximation of $\Dcal_1$ while the last $\sqrt n$ items are never sampled. 
This is indeed the case for \cite{FGYZ20,feng2020sampling,Vishesh21sampling}, and is undesirable if the CSP is used for addressing social problems and the missing solutions are contributed by the minority or the underrepresented. 
Besides the potential drawback in social fairness, this also leads to the following reasonable worry: is it possible that some solution is inherently harder to find than others. Fortunately our work shows this is not the case.

Perfect sampling is also advantageous for practical purposes and heuristic algorithms.
To perform a sampling task, Markov chain is arguably the most common approach. By the convergence theorem for Markov chains \cite{levin2017markov}, it is usually easy to show the chain mixes to the desired distribution almost surely.
However to provide a good bound for the mixing time is in general a difficult task.
On the other hand, if the mixing time is unknown or poorly analyzed, it is not sure when to stop the chain so the output distribution is close enough to the desired one.
However given the Markov chain, there are known techniques, like \emph{coupling from the past} \cite{PW96}, to convert it into a perfect sampler, which always gives desired distribution when it stops even if we may not know any bounds on its expected running time.

\subsection{Our Results}\label{sec:our_results}

In this paper, we provide perfect samplers for solutions of atomic CSPs in the local lemma regime. Though in previous paragraphs we only focus on sampling perfect uniform solution, our algorithm in fact works for general underlying distribution.

Let $\Phi$ be an atomic CSP with variable set $V$ and $|V|=n$ where each $v\in V$ is endowed with a distribution $\Dcal_v$ supported on finite domain $\Omega_v$.
\begin{itemize}
\item Let $p$ be the maximum violation probability of a constraint under distribution $\prod_{v\in V}\Dcal_v$.
\item Let $\Delta$ be the maximum constraint degree.
\item Let $Q$ be the maximum size of a variable domain $\Omega_v$.
\item Let $k$ be the maximum number of variables that a constraint depends on.
\end{itemize}
Let $\mu$ be the distribution of solutions of $\Phi$ under $\prod_{v\in V}\Dcal_v$, i.e.,
$$
\mu(\sigma)=\Pr_{\sigma'\sim\prod_{v\in V}\Dcal_v}\sbra{\sigma'=\sigma\mid\sigma'\text{ is a solution of }\Phi}
\quad\text{for each }\sigma\in\prod_{v\in V}\Omega_v.
$$

The original Lov{\'a}sz local lemma~\cite{EL75} states if $p\cdot\Delta\lesssim1$ then $\Phi$ has a solution, i.e., $\mu$ is well-defined. Then the algorithmic LLL \cite{moser2010constructive} shows one can efficiently find a solution under the same condition.
Our main theorem shows one can efficiently sample a solution distributed as $\mu$ under a similar condition.

\begin{theorem}[\Cref{thm:generalcsp_arbitrary}, Informal]\label{thm:main-informal}
If $p^\gamma\cdot\Delta\lesssim1/c$ where
$$
\gamma=\frac{3+\ln(c+1)-\sqrt{\ln^2(c+1)+6\ln(c+1)}}9
\quad\text{and}\quad
c=\max\cbra{2,\max_{v\in V}\max_{q,q'\in\Omega_v}\frac{\Dcal_v(q)}{\Dcal_v(q')}},
$$
then our algorithm runs in expected time $\poly(k,Q,\Delta)\cdot n\log(n)$ and outputs a random solution distributed perfectly as $\mu$.
\end{theorem}

We remark that for the uniform case (i.e., $\Dcal_v$ is the uniform distribution), we have $c=2$ and $\gamma>0.145$, which already beats $1/7$ from \cite{Vishesh20towards} and $0.142$ from \cite{Vishesh21sampling}.

\Cref{thm:main-informal} is proved by a black-box reduction, using our state tensorization trick, to the perfect sampling algorithm on small variable domains. It is possible to get improved bounds for specific underlying distributions by a finer analysis. 
We take the uniform distribution as a starting example and improve $0.145$ to $0.175$ for general atomic CSPs.

\begin{theorem}[\Cref{cor:generalcsp_uniform}, Informal]\label{thm:uniform-informal}
If $p^{0.175}\cdot\Delta\lesssim1$ and each $\Dcal_v$ is the uniform distribution, then our algorithm runs in expected time $\poly(k,Q,\Delta)\cdot n\log(n)$ and outputs a perfect uniform random solution.
\end{theorem}
We remark that this $0.175$ also matches previous best bound for approximately uniform sampling solutions of $k$-CNF formula \cite{Vishesh21sampling}. In fact, in our analysis binary domains are the worst case for general atomic CSPs: the bound on $k$-CNF formula is the bottleneck for the bound on general atomic CSPs.

Indeed, \Cref{thm:uniform-informal} can be further improved if variable domains are large. We use hypergraph coloring as an illustrating example, the bound of which matches the current best bound of approximate samplers \cite{Vishesh21sampling}.
\begin{theorem}[\Cref{thm:hypergraph_coloring}, Informal]\label{thm:coloring-informal}
Let $H$ be a hypergraph on $n$ vertices where each edge contains exactly $k$ vertices. Let $\Delta$ be the edge degree of $H$.
Assume each vertex can choose a color from $\cbra{1,2,\ldots,Q}$, and a coloring of vertices is \emph{proper} if it does not produce monochromatic edge.

If $\Delta\lesssim Q^{k/3}$, then our algorithm runs in expected time $\poly(k,Q,\Delta)\cdot n\log(n)$ and outputs a perfect uniform random proper coloring of $H$.
\end{theorem}

Finally we briefly discuss the connection between our result and other topics.
\begin{itemize}
\item \textbf{Non-atomic CSPs.}
For a non-atomic CSP, let $N$ be the maximum number of forbidden local configurations. We can convert it into an atomic CSP by decomposing the non-atomic constraints to atomic ones as in~\cite{feng2020sampling}, which only increases the constraint degree $\Delta$ to at most $\Delta\cdot N$.
Then our perfect sampler is still efficient if $N$ is small.
\item \textbf{Approximate Sampling.}
Our perfect sampler is a Las Vegas algorithm with quasilinear expected running time $T$. It is well-known that terminating the algorithm after $T/\eps$ steps gives an $\eps$-approximate sampler under total variation distance. In particular, the local lemma condition of our approximate sampler is the same as our perfect sampler, which breaks the current best record for atomic CSPs by \cite{Vishesh20towards,Vishesh21sampling}.

We remark that one can obtain a better bound by analyzing the moments of the running time of our perfect sampler. We do not make the effort here since this is not our focus and this may require stronger local lemma conditions.
\item \textbf{Approximate Counting.}
One way to reduce counting to sampling is to start from a CSP with no constraint, then add clauses one by one and use the self-reducibility \cite{JVV86}. 
Another strategy is to use the simulated annealing approach developed in \cite{BezakovaSVV08,StefankovicVV09,FGYZ20}. Both reductions produce efficient randomized approximate counting algorithms. We refer interested readers to their paper for details.

Similarly as in the approximate sampling case, the local lemma condition of our approximate counting algorithm is the same as our perfect sampler, which breaks the current best record for atomic CSPs by \cite{Vishesh20towards,Vishesh21sampling}.
\end{itemize}

\subsection{Proof Overview}\label{sec:proof_overview}

To illustrate the idea, we first focus on sampling a uniform solution of $k$-CNF formula:
\begin{itemize}
\item There are $n$ Boolean variables. Each variable is endowed with the uniform distribution over $\bin$, and appears in at most $d$ constraints.
\item Each constraint is a clause depending on exactly $k$ variables and has exactly one forbidden local assignment.
\end{itemize}
For example, $(x_1\lor x_2\lor\neg x_3)\land(x_1\lor x_5\lor x_7)\land(x_2\lor\neg x_4\lor\neg x_6)$ is a $3$-CNF formula where $n=6,m=3$ and $k=3,d=2$.

Similar as previous works to deal with the connectivity issue of Glauber dynamics \cite{wigderson2019mathematics}, the first step of our algorithm is to mark variables so every clause has a certain amount of marked and unmarked variables.
Let $V$ be the set of variables and $\Mcal\subseteq V$ be the set of marked variables.
Then by the local lemma (\Cref{thm:HSS_LLL}), for any $\sigma\in\bin^\Mcal$ and any $v\in\Mcal$ the following two distributions are close under total variation distance:
\begin{itemize}
\item An unbiased coin in $\bin$.
\item The distribution of $\sigma'(v)$ where $\sigma'\in\bin^V$ is a uniform random solution conditioning on $\sigma'(\Mcal\setminus\cbra{v})=\sigma(\Mcal\setminus\cbra{v})$.
\end{itemize}
We call this \emph{local uniformity}.

To sample a solution \emph{approximately}, previous works~\cite{FGYZ20,feng2020sampling,Vishesh21sampling} simulate an idealized Glauber dynamics $\Glauber$ (\Cref{alg:systematicscan}) on the assignments of the marked variables as follows:
\begin{itemize}
\item[(A)] Initialize $\sigma(v)\sim\bin$ uniformly and independently for each $v\in\Mcal$.
\item[(B)] Going forward from time $0$ to $T\to+\infty$, let $v_{t}$ be the variable selected at time $t\ge0$.
Iteratively find all clauses that are not yet satisfied by $\sigma(\Mcal\setminus\cbra{v_t})$ and are connected to $v_{t}$ (\Cref{alg:component}). Let $\Phi'$ be this sub-$k$-CNF.

Then update $\sigma(v_{t})\gets\sigma'(v_t)$, where $\sigma'\in\bin^V$ is a uniform random solution of $\Phi'$ conditioning on $\sigma'(\Mcal\setminus\cbra{v_t})=\sigma(\Mcal\setminus\cbra{v_t})$. Algorithmically, $\sigma'$ is provided via rejection sampling (\Cref{alg:rejectionsampling}) on $\Phi'$.
\item[(C)] After Step (A) (B), extend $\sigma$ to unmarked variables $V\setminus\Mcal$ by sampling a uniform random solution conditioning on $\sigma(\Mcal)$.
\end{itemize}

To sample a solution perfectly, we simulate \emph{bounding chains} $\BChains$ (\Cref{alg:boundingchain}) of $\Glauber$ as follows:
The algorithm guarantees at any point, each variable $v\in\Mcal$ is assigned with a value in $\cbra{0,1,\Qmark}$ where $\Qmark$ represents uncertainty.
\begin{itemize}
\item[(1)] Initialize $\sigma(v)=\Qmark$ for each $v\in\Mcal$.
\item[(2)] Going forward from time $-T$ to $-1$, let $v_{t}$ be the variable selected at time $-T\le t<0$.
Iteratively find all clauses that are not yet satisfied by $\sigma(\Mcal\setminus\cbra{v_t})$ and are connected with $v_{t}$ (\Cref{alg:component}). Let $\Phi'$ be this sub-CSP.
\begin{itemize}
\item If all marked variables connected to $v_t$ in $\Phi'$ have value $0$ or $1$, we say $v_{t}$ is \emph{coupled}.
Then we update $\sigma(v_t)$ by rejection sampling on $\Phi'$ and $\sigma(v_t)$ is always updated to $0$ or $1$.
\item Otherwise $\sigma(v_{t})$ is updated based on the local uniformity, which may be assigned to $\Qmark$ with small probability (\textsf{SafeSampling} subroutine in \Cref{alg:boundingchain}).
\end{itemize}
\item[(3)] After Step (1) (2), 
\begin{itemize}
\item if some marked variable has value $\Qmark$, then we double $T$ and re-run Step (1) (2);\footnote{We remark that the randomness is reused. That is, the randomness used for time $t<0$ is the same one regardless of the starting time $-T$.}
\item otherwise we stop and extend $\sigma$ to unmarked variables $V\setminus \Mcal$ by sampling a uniform random solution conditioning on $\sigma(\Mcal)$. 
\end{itemize}
\end{itemize}

To simplify the analysis of the algorithm, we use \emph{systematic scan} for $\Glauber$ and $\BChains$ rather than random scan~\cite{he2016scan}. 
Specifically, at time $t\in\Zbb$ the algorithm always updates the variable with index $(t \mod m)$ (\Cref{alg:systematicscan}) where $m=|\Mcal|$.

Let $\mu^\Mcal$ be the distribution of a uniform random solution projected on the marked variables $\Mcal$.
Our goal is to prove the following claims for $\BChains$:
\begin{itemize}
\item \textsc{Correctness.} When we stop in Step (3), $\sigma(\Mcal)$ has distribution $\mu^\Mcal$ (\Cref{sec:the_distribution_after_boundingchain_subroutines}).
\item \textsc{Efficiency.} In expectation, each update in Step (2) is efficient. (\Cref{sec:moment_bounds_on_the_running_time})
\item \textsc{Coalescence.} In expectation, we stop with $T=O(n\log(n))$ (\Cref{sec:concentration_bounds_for_the_coalescence}).
\end{itemize}

\paragraph*{Proof of Correctness.} 
Firstly we show $\Glauber$ converges to $\mu_\Mcal$ in Step (B) when $T\to+\infty$. Though it is a time inhomogeneous Markov chain, we are able to embed it into a time homogeneous Markov chain $P'$ by viewing $|\Mcal|$ consecutive updates as one step. Then it is easy to check $P'$ is aperiodic and irreducible with unique stationary distribution $\mu^\Mcal$.
After that, we unpack $P'$ to show $\Glauber$ also converges to $\mu^\Mcal$ (\Cref{lem:systematicscan_final_distribution}).

Next, we use the idea of \emph{coupling from the past} \cite{PW96} and \emph{bounding chains} \cite{huber1998exact,haggstrom1999exact}. 
\begin{itemize}
\item \textsc{Coupling from the Past.} 
Observe that for any positive integer $L$ if we run $\Glauber$ from $-L\cdot m$ to $-1$, it has the \emph{same} distribution as we run it from $0$ to $L\cdot m-1$. Thus by the argument above, Step (B) also has distribution $\mu^\Mcal$ if we run $\Glauber$ from time $-\infty$ to $-1$.
\item \textsc{Bounding Chains.}
For each $t\in\Zbb$, if $\sigma(\Mcal)$ in $\BChains$ has no $\Qmark$, then the update process is exactly $\Glauber$. This means $\BChains$ is a \emph{coupling} of $\Glauber$ (\Cref{prop:coupling}). 
Note that we use $\Qmark$ to denote uncertainty which includes all possible assignments that we need to couple. Thus when $\BChains$ stops at time $T$ with $\hat\sigma\in\bin^\Mcal$ at Step (3), any assignment, going through the updates from time $-T$ to $-1$, converges to $\hat\sigma$.
\end{itemize}
Combining the two observations above, we know $\hat\sigma$ is distributed exactly as $\mu^\Mcal$.

\paragraph*{Proof of Efficiency.}
We first remark that only the rejection sampling is time consuming, and its running time is a geometric distribution with expectation controlled by the local uniformity (\Cref{prop:rejectionsampling}).
Thus to bound its expectation, it suffices to bound the size of $\Phi'$ (\Cref{prop:component_size}). 
This uses the same $2$-tree argument as in previous works and we briefly explain here.

Since $k$-CNF has bounded degree, if $\Phi'$ is large then we can find a large independent set $S$ of clauses in $\Phi'$. Note that clauses in $\Phi'$ are connected, thus we can further assume $S$ is a $2$-tree --- $S$ will be connected if we link any two clauses in $S$ that are at distance $2$. Intuitively, a $2$-tree is an independent set that is not very spread out.
Then it suffices to union bound the probability that some large $2$-tree growing out of $v_t$ survives after previous updates.

One potential pitfall is the total running time of $\BChains$ depends on both $T$ and the running time of each update, where they can be arbitrarily correlated. Thus we need to calculate the second moment of the update time (\Cref{sec:moment_bounds_on_the_running_time}) and apply Cauchy-Schwarz inequality to break the correlation (\Cref{sec:putting_everything_together}).

\paragraph*{Proof of Coalescence.} 
To upper bound the round $T$, we employ the \emph{information percolation} argument similarly used in \cite{lubetzky2016information,HSZ19,Vishesh21sampling}. 
For the sake of analysis, we assume Step (2) in $\BChains$ also includes unmarked variables though it does nothing for the update, i.e., $v_t$ is the variable with index $(t\mod n)$ where $n=|V|$.

The crucial observation is the following.
If $\sigma(v_{t_0})$ is updated to $\Qmark$ at time $t_0$, then at that point there must be some variable $u\neq v_{t_0}$ with value $\Qmark$ and connected to $v_{t_0}$.
Let $t_1$ be the last update time of $u$ before $t_0$, and thus $u = v_{t_1}$.
Then we can find another variable $u'\neq v_{t_1}$ with value $\Qmark$ and connected to $v_{t_1}$ at time $t_1$.
Continuing this process until we reach the initialization phase, we will find a list of time $0>t_0>t_1>\cdots>t_\ell\ge-T$ such that for each time $t_i$,
\begin{itemize}
\item $\sigma(v_{t_{i}})$ is updated to $\Qmark$,
\item $v_{t_i}$ is connected to $v_{t_{i+1}}$ and $v_{t_{i+1}}$ has value $\Qmark$.
\end{itemize}

To express constraints through time, we define the extended constraint $(e,C)$ (\Cref{def:extended_constraints}), where $C$ is a clause and $e=\cbra{t'_1,\ldots,t'_k}\subseteq\cbra{-T,\ldots,-1}$ is a time sequence such that
\begin{itemize}
\item $v_{t'_1},\ldots,v_{t'_k}$ are the variables $C$ depends on,
\item $t'_1,\ldots,t'_k$ are succinct rounds of update for each $v_{t'_1},\ldots,v_{t'_k}$.
\end{itemize}
Since each $v_{t_i}$ and $v_{t_{i+1}}$ are connected at time $t_i$ as we discussed above, it means we are able to find extended constraints $(e^i_1,C^i_1),\ldots,(e^i_{s_i},C^i_{s_i})$ such that $t_i\in e^i_1$ and $t_{i+1}\in e^i_{s_i}$ and $e^i_1,\ldots,e^i_{s_i}$ are connected over $\cbra{-T,\ldots,-1}$. 
Thus all the edges $\{e^i_j\}_{i,j}$ form a connected sub-hypergraph $H$ on vertex set $\cbra{-T,\ldots,-1}$.

Since each extended constraint $(e,C)$ represents succinct rounds of update for variables in $C$, $e$ has range less than $n=|V|$, i.e., $\max_{t_1,t_2\in e}|t_1-t_2|<n$. Thus the longest path $\Pcal$ in $H$ has length $|\Pcal|=\Omega(T/n)$. On the other hand, Step (2) of $\BChains$ only finds clauses that are \emph{not} satisfied by marked variables, which means each fixed extended constraint $(e,C)$ appears in $H$ with extremely low (roughly $2^{-k}$) probability.

Putting everything together, if we do not stop at round $T$, we should find a path $\Pcal$ of extended constraints of length $|\Pcal|=\Omega(T/n)$. 
Meanwhile, each fixed extended constraint is found with probability at most roughly $2^{-k}$. Thus any fixed $\Pcal$ exists with probability at most $2^{-k\cdot|\Pcal|/2}$, since extended constraints in odd positions of $\Pcal$ form an independent set.
Moreover, it is easy to see each extended constraint overlaps with $O(k^2d)$ many other extended constraints, which provides an upper bound $O(k^2d)^{|\Pcal|}$ for the number of possible $\Pcal$. 
By union bound, the probability that we do not stop at round $T$ is roughly 
$$
n\cdot\pbra{\frac{k^4d^2}{2^k}}^{\Omega\pbra{T/n}}\ll n\cdot 2^{-T/n},
$$
where we assume $k^4d^2\ll2^k$ and the additional $n$ comes from choosing $t_0\in\cbra{-1,\ldots,-n}$, i.e., the last update resulting in $\Qmark$.

We remark that to deal with general atomic CSPs, we need to be more careful with the union bound (\Cref{sec:concentration_bounds_for_the_coalescence}). This is because in general atomic CSPs constraints may depend on different number of variables.
Nevertheless, our analysis is much simpler than the one in~\cite{Vishesh21sampling}.
One main reason is that we use systematic scan instead of random scan, which makes the updates of each variable well behave through time.
Moreover, our main data structure, extended constraint (\Cref{def:extended_constraints}), is also much simpler than the discrepancy checks used in their argument.

\paragraph*{State Tensorization.}
The marking process is only efficient when the variable domains are small; otherwise it cannot guarantee useful local uniformity.
Similarly for $\Qmark$, it compresses too much information when the domain is large.
Therefore to deal with large domains, we introduce a simple \emph{state tensorization} trick to perform the reduction.

For intuition, let's consider the following concrete atomic CSP $\Phi$:
\begin{itemize}
\item The variables are $u,v$ where $u$ is endowed with distribution $\Dcal_u$ by $\Dcal_u(a)=\Dcal_u(b)=\Dcal_u(c)=1/3$; and $v$ is endowed with $\Dcal_v$ by $\Dcal_v(A)=\Dcal_v(B)=1/4,\Dcal_v(C)=1/3,\Dcal_v(D)=1/6$.
\item The constraints are $C_1,C_2$ where $C_1=\False$ iff $u=a$; and $C_2=\False$ iff $u=c,v=B$.
\end{itemize}
Then we describe one possible state tensorization as follows (See \Cref{fig:state_tensorization_example}):
\begin{itemize}
\item Define variables $u_1,u_2$ where $u_1$ is endowed with $\Dcal_{u_1}$ by $\Dcal_{u_1}(0)=2/3,\Dcal_{u_1}(1)=1/3$; and $u_2$ is endowed with $\Dcal_{u_2}(0)=\Dcal_{u_2}(1)=1/2$.

We interpret $u=a$ if $(u_1,u_2)=(0,0)$; $u=b$ if $(u_1,u_2)=(0,1)$; and $u=c$ if $u_1=1$.
\item Define variables $v_1,v_2,v_3$ where $v_1$ is endowed with $\Dcal_{v_1}$ by $\Dcal_{v_1}(0)=\Dcal_{v_1}(1)=1/2$; and $v_2$ is endowed with $\Dcal_{v_2}(0)=\Dcal_{v_2}(1)=1/2$; and $v_3$ is endowed with $\Dcal_{v_3}(0)=2/3,\Dcal_{v_3}(1)=1/3$.

We interpret $v=A$ if $(v_1,v_2)=(0,0)$; $u=C$ if $(v_1,v_3)=(1,0)$, etc.
\end{itemize}

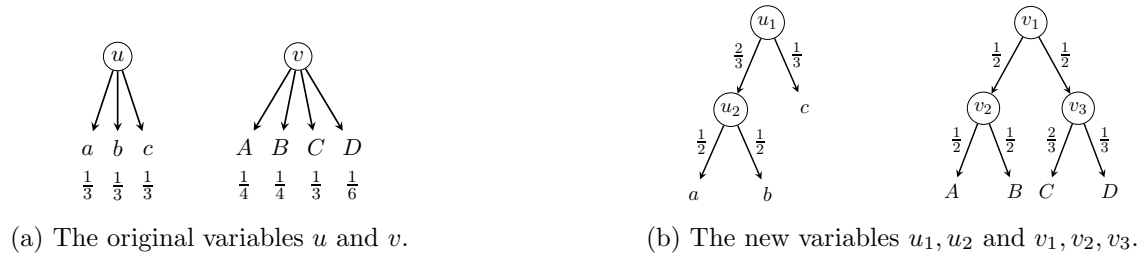
\begin{figure}[ht]
\centering
\begin{subfigure}[b]{0.4\textwidth}
\centering
\scalebox{0.8}{
\begin{tikzpicture}[emptyC/.style={draw,circle,inner sep=2pt}, outE/.style={->,>=stealth,thick}]
\node[emptyC] (v0) at (0,0) {$u$};
\node (q1) [below=of v0,xshift=-5mm] {$\phantom{A}a\phantom{A}$};
\node (qq1) [below=of q1,yshift=10mm] {$\frac13$};
\node (q2) [below=of v0] {$\phantom{A}b\phantom{A}$};
\node (qq2) [below=of q2,yshift=10mm] {$\frac13$};
\node (q3) [below=of v0,xshift=5mm] {$\phantom{A}c\phantom{A}$};
\node (qq3) [below=of q3,yshift=10mm] {$\frac13$};
\draw[outE] (v0) -- (q1);
\draw[outE] (v0) -- (q2);
\draw[outE] (v0) -- (q3);

\node[emptyC] (u0) at (3,0) {$v$};
\node (r1) [below=of u0, xshift=-9mm] {$A$};
\node (rr1) [below=of r1,yshift=10mm] {$\frac14$};
\node (r2) [below=of u0, xshift=-3mm] {$B$};
\node (rr2) [below=of r2,yshift=10mm] {$\frac14$};
\node (r3) [below=of u0, xshift=3mm] {$C$};
\node (rr3) [below=of r3,yshift=10mm] {$\frac13$};
\node (r4) [below=of u0, xshift=9mm] {$D$};
\node (rr4) [below=of r4,yshift=10mm] {$\frac16$};
\draw[outE] (u0) -- (r1);
\draw[outE] (u0) -- (r2);
\draw[outE] (u0) -- (r3);
\draw[outE] (u0) -- (r4);
\end{tikzpicture}}
\caption{The original variables $u$ and $v$.}
\end{subfigure}
\hfill
\begin{subfigure}[b]{0.5\textwidth}
\centering
\scalebox{0.7}{
\begin{tikzpicture}[emptyC/.style={draw,circle,inner sep=2pt}, outE/.style={->,>=stealth,thick}]
\node[emptyC] (v0) at (0,0) {$u_1$};
\node[emptyC] (v1) [below=of v0, xshift=-7mm] {$u_2$};
\node (q1) [below=of v1, xshift=-7mm] {$\phantom{A}a\phantom{A}$};
\node (q2) [below=of v1, xshift=7mm] {$\phantom{A}b\phantom{A}$};
\node (q3) [below=of v0, xshift=7mm] {$\phantom{A}c\phantom{A}$};
\draw[outE] (v0) --node[xshift=-2mm, yshift=2mm]{$\frac23$} (v1);
\draw[outE] (v0) --node[xshift=2mm, yshift=2mm]{$\frac13$} (q3);
\draw[outE] (v1) --node[xshift=-2mm, yshift=2mm]{$\frac12$} (q1);
\draw[outE] (v1) --node[xshift=2mm, yshift=2mm]{$\frac12$} (q2);

\node[emptyC] (u0) at (5,0) {$v_1$};
\node[emptyC] (u1) [below=of u0, xshift=-9mm] {$v_2$};
\node[emptyC] (u2) [below=of u0, xshift=9mm] {$v_3$};
\node (r1) [below=of u1, xshift=-6mm] {$A$};
\node (r2) [below=of u1, xshift=6mm] {$B$};
\node (r3) [below=of u2, xshift=-6mm] {$C$};
\node (r4) [below=of u2, xshift=6mm] {$D$};
\draw[outE] (u0) --node[xshift=-2mm, yshift=2mm]{$\frac12$} (u1);
\draw[outE] (u0) --node[xshift=2mm, yshift=2mm]{$\frac12$} (u2);
\draw[outE] (u1) --node[xshift=-2mm, yshift=2mm]{$\frac12$} (r1);
\draw[outE] (u1) --node[xshift=2mm, yshift=2mm]{$\frac12$} (r2);
\draw[outE] (u2) --node[xshift=-2mm, yshift=2mm]{$\frac23$} (r3);
\draw[outE] (u2) --node[xshift=2mm, yshift=2mm]{$\frac13$} (r4);
\end{tikzpicture}}
\caption{The new variables $u_1,u_2$ and $v_1,v_2,v_3$.}
\end{subfigure}
\caption{An example for state tensorization.}\label{fig:state_tensorization_example}
\end{figure}

Hence after the state tensorization, $C_1=\False$ iff $u_1=0,u_2=0$; and $C_2=\False$ iff $u_1=1,v_1=0,v_2=1$. Moreover, sampling the value of $u,v$ from $\Dcal_u\times\Dcal_v$ is equivalent to first sampling the value of $u_1,u_2,v_1,v_2,v_3$ from $\Dcal_{u_1}\times\Dcal_{u_2}\times\Dcal_{v_1}\times\Dcal_{v_2}\times\Dcal_{v_3}$ and then interpret the value of $u,v$ from them. Therefore, to obtain a random solution under distribution $\Dcal_u\times\Dcal_v$, it suffices to first obtain a random solution under the product distribution of $u_1,u_2,v_1,v_2,v_3$ and then interpret them back.

Most importantly, this reduction does not change the violation probability of any individual constraint, nor change the dependency relation among constraints. 
These essentially guarantees that the desired local lemma condition does not deteriorate after the reduction.
The formal description of the reduction can be found in \Cref{sec:large_domains:state_tensorization}.

We remark that state tensorization, combined with the marking $\Mcal$, generalizes the state compression technique in \cite{feng2020sampling}.
On the other hand, state tensorization is similar to standard gadget reduction in the study of complexity theory. For example by encoding large alphabets using binary bits, one can show Boolean CSPs are no easier to solve than CSPs with large variable domains for polynomial time algorithms. However we are not aware of such simple idea being used for perform sampling tasks.

\paragraph*{Organization.} 
We give formal definitions in \Cref{sec:preliminaries}. Useful subroutines are provided in \Cref{sec:useful_subroutines} and then we describe and analyze our main algorithm in \Cref{sec:the_atomiccspsampling_algorithm}. We discuss our result for different applications in \Cref{sec:applications}. 

\section{Preliminaries}\label{sec:preliminaries}

We use $\Naturale\approx2.71828$ to denote the \emph{natural base}.
We use $\log(\cdot)$ and $\ln(\cdot)$ to denote the logarithm with base $2$ and $\Naturale$ respectively. We use $[N]$ to denote $\cbra{1,2,\ldots,N}$; and use $\Zbb$ to denote the set of all integers. 
We say $V$ is a \emph{disjoint union} of $\pbra{V_i}_{i\in[s]}$ if $V=\bigcup_{i\in[s]}V_i$ and $V_i\cap V_j=\emptyset$ holds for any distinct $i,j\in[s]$. 
For positive integer $m$, $t\mod m=t-m\cdot\lfloor t/m\rfloor$ for non-negative integer $t$; and $t\mod m=t\cdot(1-m)\mod m$ for negative integer $t$.

For any index set $I$ and domains $\pbra{\Omega_i}_{i\in I}$, we use $\prod_{i\in I}\Omega_i$ to denote their product space. For some vector $\mathsf{vec}\in\prod_{i\in I}\Omega_i$, we use $\mathsf{vec}(i)\in\Omega_i$ to denote the entry of $\mathsf{vec}$ indexed by $i$; and use $\mathsf{vec}(J)\in\prod_{i\in J}\Omega_i$ to denote the entries of $\mathsf{vec}$ on indices $J\subseteq I$. 

For a finite set $\Xcal$ and a distribution $\Dcal$ over $\Xcal$, we use $x\sim\Dcal$ to denote that $x$ is a random variable sampled from $\Xcal$ according to distribution $\Dcal$. 
For two events $\Ecal_1,\Ecal_2$ with $\Pr\sbra{\Ecal_2}=0$, we define the conditional probability $\Pr\sbra{\Ecal_1(x)\mid\Ecal_2(x)}=0$. We say event $\Ecal$ happens \emph{almost surely} if $\Pr\sbra{\Ecal}=1$.

\paragraph*{Constraint Satisfaction Problems.}
Let $V$ be a set of variables with finite domains $\pbra{\Omega_v}_{v\in V}$. 
A \emph{constraint} $C$ on $V$ is a mapping $C\colon\prod_{v\in V}\Omega_v\to\binTF$.
We say $C$ \emph{depends on} $v\in V$ if there exists $\sigma_1,\sigma_2\in\prod_{v\in V}\Omega_v$ such that $C(\sigma_1)\neq C(\sigma_2)$ and $\sigma_1,\sigma_2$ differ in (and only in) $v$. 
We use $\vbl(C)$ to denote the set of variables that $C$ depends on, then $C$ can be viewed as a mapping from $\prod_{v\in\vbl(C)}\Omega_v$ to $\binTF$.

For convenience we use $\sigma_\False^C\subseteq\prod_{v\in\vbl(C)}\Omega_v$ (resp., $\sigma_\True^C\subseteq\prod_{v\in\vbl(C)}\Omega_v$) to denote the set of falsifying (resp., satisfying) assignments of $C$. 
More generally, for $\Ccal$ being a set of constraints, we use $\sigma_\False^\Ccal$ (resp., $\sigma_\True^\Ccal$) to denote the set of falsifying (resp., satisfying) assignments of $\Ccal$, i.e., $C(\sigma)=\False$ for all $\sigma\in\sigma_\False^\Ccal$ and \emph{some} $C\in\Ccal$ (resp., $C(\sigma)=\True$ for all $\sigma\in\sigma_\True^\Ccal$ and \emph{all} $C\in\Ccal$).

For sampling LLL, we also need to specify the underlying distribution. Assume each $v\in V$ has some distribution $\Dcal_v$ supported on $\Omega_v$.\footnote{This means $\Dcal_v(U)>0$ for all $U\in\Omega_v$. One natural choice is the uniform distribution.}
Define $\mu_\True^\Ccal$ as the distribution of solutions of $\Ccal$ induced by $\pbra{\Dcal_v}_{v\in V}$, i.e.,
$$
\mu_\True^\Ccal(\sigma)=\Pr_{\sigma'\sim\prod_{v\in V}\Dcal_v}\sbra{\sigma'=\sigma\mid\sigma'\in\sigma_\True^\Ccal}
\quad\text{for each $\sigma\in\prod_{v\in V}\Omega_v$}.
$$

\begin{definition}[(Atomic) Constraint Satisfaction Problem]\label{def:(atomic)_constraint_satisfaction_problem}
A \emph{constraint satisfaction problem} is specified by $\Phi=\pbra{V,\pbra{\Omega_v,\Dcal_v}_{v\in V},\Ccal}$ where $\Ccal$ is a set of constraints on $V$ and each $\Dcal_v$ is a distribution supported on $\Omega_v$.

We say $\Phi$ is \emph{atomic} if 
$\abs{\sigma_\False^C}=1$ for all $C\in\Ccal$. In this case, we abuse the notation to define $\sigma_\False^C$ as the unique falsifying assignment of $C$.

In addition, we define the following measures for $\Phi$:
\begin{itemize}
\item the \emph{width} is $k=k(\Phi)=\max_{C\in\Ccal}\abs{\vbl(C)}$;
\item the \emph{variable degree} is $d=d(\Phi)=\max_{v\in V}\abs{\cbra{C\in\Ccal\mid v\in\vbl(C)}}$;
\item the \emph{constraint degree} is $\Delta=\Delta(\Phi)=\max_{C\in\Ccal}\abs{\cbra{C'\in\Ccal\mid\vbl(C)\cap\vbl(C')\neq\emptyset}}$;\footnote{Here $\Delta$ is one plus the maximum degree of the dependency graph of $\Phi$ since $C\in\cbra{C'\in\Ccal\mid\vbl(C)\cap\vbl(C')\neq\emptyset}$.}
\item the \emph{domain size} is $Q=Q(\Phi)=\max_{v\in V}\abs{\Omega_v}$;
\item the \emph{maximal individual falsifying probability} is 
$$
p=p(\Phi)
=\max_{C\in\Ccal}\Pr_{\sigma\sim\prod_{v\in\vbl(C)}\Dcal_v}\sbra{C(\sigma)=\False}
=\max_{C\in\Ccal}\sum_{\sigma\in\sigma_\False^C}\prod_{v\in\vbl(C)}\Dcal_v(\sigma(v)).
$$
\end{itemize}
\end{definition}
We will simply use $k,d,\Delta,Q,p$ when $\Phi$ is clear from the context. In addition we assume $\Delta\ge2$, $d\ge2$, and $|V|\ge2$ since otherwise the constraints in $\Phi$ are independent and the sampling problem becomes trivial.

\paragraph*{Lov\'asz Local Lemma.}
The Lov\'asz local lemma provides sufficient conditions to guarantee the existence of a solution of CSPs. 
\begin{theorem}[\cite{EL75}]\label{thm:LLL}
Let $\Phi=\pbra{V,\pbra{\Omega_v,\Dcal_v}_{v\in V},\Ccal}$ be a CSP. If $\Naturale p\Delta\le1$, then 
$$
\Pr_{\sigma\sim\prod_{v\in V}\Dcal_v}\sbra{\sigma\in\sigma_\True^\Ccal}\ge
\pbra{1-\Naturale p}^{\abs{\Ccal}}>0.
$$
\end{theorem}

The following more general version, first stated in \cite{haeupler2011new}, can be proved with minor modification of the original proof of the Lov\'asz local lemma \cite{EL75}.
\begin{theorem}[{\cite[Theorem 2.1]{haeupler2011new}}]\label{thm:HSS_LLL}
Let $\Phi=\pbra{V,\pbra{\Omega_v,\Dcal_v}_{v\in V},\Ccal}$ be a CSP. 
If $\Naturale p\Delta\le1$, then $\sigma_\True^\Ccal\neq\emptyset$ and for any constraint $B$ (not necessarily from $\Ccal$) we have
$$
\Pr_{\sigma\sim\mu_\True^\Ccal}\sbra{B(\sigma)=\True}
\le\pbra{1-\Naturale p}^{-\abs{\Gamma(B)}}\Pr_{\sigma\sim\prod_{v\in V}\Dcal_v}\sbra{B(\sigma)=\True},
$$
where $\Gamma(B)=\cbra{C\in\Ccal\mid\vbl(C)\cap\vbl(B)\neq\emptyset}$.
\end{theorem}

\paragraph*{Hypergraphs.}
Here we give definitions related with hypergraphs. All the definitions directly translate to graphs when every edge in the hypergraph contains two vertices.

Let $H$ be a hypergraph with finite \emph{vertex set} $V(H)$ and finite \emph{edge set} $E(H)$. Each edge $e\in E(H)$ is a non-empty subset of $V(H)$. We \emph{allow} multiple occurrence of a same edge. 
For any CSP $\Phi=\pbra{V,\pbra{\Omega_v,\Dcal_v}_{v\in V},\Ccal}$ we naturally view it as a hypergraph $H(\Phi)$ where 
\begin{equation}\label{eq:csp_to_hypergraph}
V(H(\Phi))=V
\quad\text{and}\quad
E(H(\Phi))=\cbra{\vbl(C)}_{C\in\Ccal}.
\end{equation}
Similar as the measures of CSPs, we define the following measures for a hypergraph $H$:
\begin{itemize}
\item the \emph{width} is $k=k(H)=\max_{e\in E(H)}|e|$;
\item the \emph{vertex degree} is $d=d(H)=\max_{v\in V(H)}\abs{\cbra{e\in E(H)\mid v\in e}}$;
\item the \emph{edge degree} is $\Delta=\Delta(H)=\max_{e\in E(H)}\abs{\cbra{e'\in E(H)\mid e\cap e'\neq\emptyset}}$.
\end{itemize}

For any two vertices $u,v\in V(H)$, we say they are \emph{adjacent} if there exists some $e\in E(H)$ such that $u\in e$ and $v\in e$; we say they are \emph{connected} if there exists a vertex sequence $w_1,w_2,\ldots,w_d\in V(H)$ such that $w_1=u,w_d=v$ and each $w_i,w_{i+1}$ are adjacent. 
Then hypergraph $H$ is \emph{connected} if any two vertices $u,v\in V(H)$ are connected. Furthermore, we have the following basic fact regarding connected hypergraphs.
\begin{fact}\label{fct:path_in_connected_hypergraph}
Assume $H$ is a connected hypergraph. Then for any $e,e'\in E(H)$, there exists a sequence of edges $e_1,e_2,\ldots,e_\ell$ such that the following holds.
\begin{itemize}
\item $e_1=e$, $e_\ell=e'$, and $e_i\cap e_{i+1}\neq\emptyset$ for all $i\in[\ell-1]$.
\item $e_i\cap e_j=\emptyset$ for all $i,j\in[\ell]$ with $|i-j|>1$.
\end{itemize}
\end{fact}

A hypergraph $H'$ is a \emph{sub-hypergraph} of $H$ if $V(H')\subseteq V(H)$ and $E(H')\subseteq E(H)$. 
If in addition $e\cap V(H')=\emptyset$ holds for all $e\in E(H)\setminus E(H')$, we say $H'$ is an \emph{induced sub-hypergraph} of $H$.

\paragraph*{Marking.}
Apart from the CSP $\Phi=\pbra{V,\pbra{\Omega_v,\Dcal_v}_{v\in V},\Ccal}$ itself, our algorithms will also need a subset of $V$ which we call \emph{marking}. Both the correctness and efficiency of our algorithms rely on the marking. 

Assume $\Phi$ is atomic and recall our measures for $\Phi$ from \Cref{def:(atomic)_constraint_satisfaction_problem}. We define the following constants given marking $\Mcal$, the meaning of which will be clear as we proceed to the next section:
\begin{itemize}
\item The \emph{maximal conditional falsifying probability} of $(\Phi,\Mcal)$ is 
\begin{equation}\label{eq:alpha}
\alpha=\alpha(\Phi,\Mcal)=\max_{C\in\Ccal}\prod_{v\in\vbl(C)\setminus\Mcal}\Dcal_v(\sigma_\False^C(v)).
\end{equation}
\item When $\Naturale\alpha\le1$, define
\begin{itemize}
\item the \emph{multiplicative bias} of $(\Phi,\Mcal)$ as
\begin{equation}\label{eq:beta}
\beta=\beta(\Phi,\Mcal)=(1-\Naturale\alpha)^{-d};
\end{equation}
\item the \emph{maximal multiplicative-biased falsifying probability} of $(\Phi,\Mcal)$ as
\begin{equation}\label{eq:rho}
\rho=\rho(\Phi,\Mcal)=\max_{C\in\Ccal}\prod_{v\in\vbl(C)\cap\Mcal}\beta\cdot\Dcal_v(\sigma_\False^C(v));
\end{equation}
\item the \emph{maximal multiplicative-biased unpercolated probability} of $(\Phi,\Mcal)$ as
\begin{equation}\label{eq:lambda}
\lambda=\lambda(\Phi,\Mcal)=\max_{C\in\Ccal}\abs{\vbl(C)}^2\prod_{v\in\vbl(C)\cap\Mcal}\pbra{\beta\cdot\Dcal_v(\sigma_\False^C(v))+(\beta-1)\pbra{|\Omega_v|-2}}.
\end{equation}
\end{itemize}
\end{itemize}
When context is clear, we will just use $\alpha,\beta,\rho,\lambda$.

\paragraph*{Wildcard.}
We will reserve $\Qmark$ as a wildcard symbol. 
Our algorithms will use $\Qmark$ to represent all the possibilities in some domain. 

Let $\Phi=\pbra{V,\pbra{\Omega_v,\Dcal_v}_{v\in V},\Ccal}$ be a CSP.
For each $v\in V$ we assume $\Qmark\notin\Omega_v$ and define $\Omega_v^\Qmark=\Omega_v\cup\cbra{\Qmark}$.
For any $C\in\Ccal$ and $\sigma\in\prod_{v\in V}\Omega_v^\Qmark$ we abuse notation to define
\begin{equation}\label{eq:qmark_constraint}
C(\sigma)=\begin{cases}
\False & \exists\sigma'\in\prod_{v\in V}\Omega_v\text{ such that }C(\sigma')=\False\text{ and }\sigma(v)\in\cbra{\sigma'(v),\Qmark}\text{ for all }v\in V,\\
\True & \text{otherwise}.
\end{cases}
\end{equation}

Here we define CSPs projected on assignments with $\Qmark$ coordinates. Intuitively this assignment fixes (and only fixes) the non-$\Qmark$ variables for the CSP. But pedantically we provide the following definition.
\begin{definition}[Projected Constraint Satisfaction Problem]\label{def:projected_constraint_satisfaction_problem}
For any $\sigma\in\prod_{v\in V}\Omega_v^\Qmark$, we define the \emph{projected constraint satisfaction problem} $\Phi|_\sigma=\pbra{V,\pbra{\Omega_v|_\sigma,\Dcal_v|_\sigma}_{v\in V},\Ccal|_\sigma}$ by setting
$$
(\Omega_v|_\sigma,\Dcal_v|_\sigma)=\begin{cases}
(\Omega_v,\Dcal_v) & \sigma(v)=\Qmark,\\
(\cbra{\sigma(v)},\text{point distribution}) & \text{otherwise}
\end{cases}
\quad\text{for all }v\in V
$$
and $\Ccal|_\sigma=\cbra{C|_\sigma\mid C\in\Ccal,C(\sigma)=\False}$ where $C|_\sigma$ has the same evaluation rule as $C\in\Ccal$ but depends on possibly fewer variables, i.e., $\vbl(C|_\sigma)=\cbra{v\in\vbl(C)\mid\sigma(v)=\Qmark}\subseteq\vbl(C)$. Similarly we sometimes view $C|_\sigma$ as constraint only on $\vbl(C|_\sigma)$.
\end{definition}

Recall the measures defined in \Cref{def:(atomic)_constraint_satisfaction_problem}. We note the following simple fact.
\begin{fact}\label{fct:after_projection}
$k(\Phi|_\sigma)\le k(\Phi)$, $d(\Phi|_\sigma)\le d(\Phi|_\sigma)$, $\Delta(\Phi|_\sigma)\le\Delta(\Phi)$, and $Q(\Phi|_\sigma)\le Q(\Phi)$. Moreover, if $\Phi$ is an atomic CSP, then $\Phi|_\sigma$ is an atomic CSP.
\end{fact}

\section{Useful Subroutines}\label{sec:useful_subroutines}
In this section, we provide some useful subroutines for later reference in our main algorithm.

\subsection{A \textsf{Component} Subroutine}\label{sec:a_component_subroutine}

Recall notations defined in \Cref{def:projected_constraint_satisfaction_problem}.
We first set up the following \Component{$\Phi,\Mcal,\sigma,u$} subroutine, which uses $u\in V$ and current assignment $\sigma\in\prod_{v\in V}\Omega_v^\Qmark$ to (hopefully) decompose projected CSP $\Phi|_\sigma$ into two disjoint parts: One containing $u$ and one isolated from $u$. 
For our purpose, the input will guarantee that $\Phi$ is atomic, $\sigma(u)=\Qmark$, and $\sigma(v)\equiv\Qmark$ for all $v\in V\setminus\Mcal$.

\begin{algorithm}[ht]
\caption{The \textsf{Component} subroutine}\label{alg:component}
\DontPrintSemicolon
\LinesNumbered
\KwIn{an atomic CSP $\Phi=\pbra{V,\pbra{\Omega_v,\Dcal_v}_{v\in V},\Ccal}$, a marking $\Mcal\subseteq V$, an assignment $\sigma\in\pbra{\prod_{v\in\Mcal}\Omega_v^\Qmark}\times\cbra{\Qmark}^{V\setminus\Mcal}$, and $u\in V$ with $\sigma(u)=\Qmark$}
\KwOut{$\pbra{\Phi',\Token}$ where $\Phi'=\pbra{V',\pbra{\Omega_v|_\sigma,\Dcal_v|_\sigma}_{v\in V'},\Ccal'}$ and $\Token\in\binTF$}
Initialize $V'\gets\cbra{u}$ and $\Ccal'\gets\emptyset$\;
\While{$\exists C|_\sigma\in\Ccal|_\sigma\setminus\Ccal'$ with $\vbl(C|_\sigma)\cap V'\neq\emptyset$}{
\lIf{$\vbl(C|_\sigma)\subseteq\cbra{u}\cup\pbra{V\setminus\Mcal}$}{Update $V'\gets V'\cup\vbl(C|_\sigma)$ and $\Ccal'\gets\Ccal'\cup\cbra{C|_\sigma}$}
\lElse{\Return{$(\Phi',\False)$}}
}
\Return{$\pbra{\Phi',\True}$}
\end{algorithm}

Here we note the following observation regarding \Cref{alg:component}.
\begin{proposition}\label{prop:component}
The following holds for \Component{$\Phi,\Mcal,\sigma,u$}.
\begin{itemize}
\item[(1)] It runs in time $O\pbra{\Delta k|\Ccal'|+dk}$.
\item[(2)] $u\in V'\subseteq V$, $\Ccal'\subseteq\Ccal|_\sigma$, and $\sigma(v)=\Qmark$ for all $v\in V'$. Moreover, $\sigma(v)=\sigma_\False^C(v)$ holds for any $C|_\sigma\in\Ccal'$ and $v\in\pbra{\vbl(C)\cap\Mcal}\setminus\cbra{u}$. 
\item[(3)] $\Phi'$ is an atomic CSP and hypergraph $H(\Phi')$ (Defined in \Cref{eq:csp_to_hypergraph}) is connected.
\item[(4)] If $\Token=\True$, let $V''=V\setminus V'$ and $\Ccal''=\Ccal|_\sigma\setminus\Ccal'$.
Then $\Phi''=\pbra{V'',\pbra{\Omega_v|_\sigma,\Dcal_v|_\sigma}_{v\in V''},\Ccal''}$ is an atomic CSP. Moreover $\sigma_\True^{\Ccal|_\sigma}=\sigma_\True^{\Ccal'}\times\sigma_\True^{\Ccal''}$ and $\mu_\True^{\Ccal|_\sigma}=\mu_\True^{\Ccal'}\times\mu_\True^{\Ccal''}$.
\end{itemize}
\end{proposition}
\begin{proof}
Item (2) is evident from the algorithm, $\Phi$ being atomic, and \Cref{def:projected_constraint_satisfaction_problem}.

For Item (3), note that each time we add a constraint $C|_\sigma$ into $\Ccal'$, we add $\vbl(C|_\sigma)$ into $V'$. Thus $\Phi'$ is a CSP.
Since $\Phi$ is atomic, by \Cref{fct:after_projection} $\Phi'$ is also atomic. In addition, we only consider $C|_\sigma$ with $\vbl(C|_\sigma)\cap V'\neq\emptyset$ at \textsf{Line 2}, hence $H(\Phi')$ is connected.

For Item (1), the algorithm can be executed by first checking all (at most $d$) constraints related with $u$, then iteratively checking (at most $\Delta|\Ccal'|$ in total) constraints related with the newly added constraints in $\Ccal'$.
Thus the total running time is $O(k)\cdot\pbra{\Delta|\Ccal'|+d}$.

Now we focus on Item (4) when $\Token=\True$. 
The condition from \textsf{Line 2} implies for any $C|_\sigma\in\Ccal''$, $\vbl(C|_\sigma)\cap V'=\emptyset$ and thus $\vbl(C|_\sigma)\subseteq V''$. Therefore $\Phi''$ is a CSP. Since $\Phi$ is atomic, by \Cref{fct:after_projection} $\Phi''$ is also atomic.
Then the ``moreover'' part follows from $V$ (resp., $\Ccal|_\sigma$) being a disjoint union of $V'$ and $V''$ (resp., $\Ccal'$ and $\Ccal''$).
\end{proof}

\subsection{A \textsf{RejectionSampling} Subroutine}\label{sec:a_rejectionsampling_subroutine}
The following simple perfect sampler, which is based on the standard rejection sampling technique, will be a building block for our main algorithm.

\begin{algorithm}[H]
\caption{The \textsf{RejectionSampling} algorithm}\label{alg:rejectionsampling}
\DontPrintSemicolon
\LinesNumbered
\KwIn{a CSP $\Phi=\pbra{V,\pbra{\Omega_v,\Dcal_v}_{v\in V},\Ccal}$ and a randomness tape $r$}
\KwOut{an assignment $\sigma\in\mu_\True^\Ccal$}
\While{$\True$}{
Sample $\sigma\sim\prod_{v\in V}\Dcal_v$ with fresh randomness from $r$\;
\lIf{$C(\sigma)=\True$ for all $C\in\Ccal$}{\Return{$\sigma$}}
}
\end{algorithm}

We have the following result on \Cref{alg:rejectionsampling} by basic facts of geometric distributions.
\begin{fact}\label{fct:rejectionsampling}
The following holds for \RejectionSampling{$\Phi,r$} over random $r$ if $\sigma_\True^\Ccal\neq\emptyset$.
\begin{itemize}
\item It halts almost surely, and outputs $\sigma\sim\mu_\True^\Ccal$ when it halts.
\item Let $T$ be the number of \textsf{while} iterations it takes before it halts. Then 
$$
\E[T]=\frac1{\Pr_{\sigma\sim\prod_{v\in V}\Dcal_v}\sbra{\sigma\in\sigma_\True^\Ccal}}
\quad\text{and}\quad
\E\sbra{T^2}=2\cdot\pbra{\E[T]}^2-\E[T].
$$ 
\item Let $X$ be its total running time.\footnote{Each \textsf{while} iteration can be performed in time $O(k|\Ccal|+Q|V|)$ where $O(Q)\cdot|V|)$ is for \textsf{Line 2} and $O(k)\cdot|\Ccal|$ is for \textsf{Line 3}.} Then 
$$
\E[X]=O\pbra{\E[T]\cdot\pbra{k|\Ccal|+Q|V|}}
\quad\text{and}\quad
\E\sbra{X^2}=O\pbra{\pbra{\E[T]\cdot\pbra{k|\Ccal|+Q|V|}}^2}.
$$ 
\end{itemize}
\end{fact}

The following result is useful when we perform rejection sampling on a projected CSP, say $\Phi'$ from \Cref{alg:component}.
\begin{proposition}\label{prop:rejectionsampling}
Let $\Phi=\pbra{V,\pbra{\Omega_v,\Dcal_v}_{v\in V},\Ccal}$ be an atomic CSP and $\Mcal\subseteq V$ be a marking. Let $k=k(\Phi)$, $\Delta=\Delta(\Phi)$, $Q=Q(\Phi)$, $\alpha=\alpha(\Phi,\Mcal)$, and $\beta=\beta(\Phi,\Mcal)$.
Let $\sigma\in\pbra{\prod_{v\in\Mcal}\Omega_v^\Qmark}\times\cbra{\Qmark}^{V\setminus\Mcal}$ be an arbitrary assignment and $\Phi'=\pbra{V',\pbra{\Omega_v|_\sigma,\Dcal_v|_\sigma}_{v\in V'},\Ccal'}$ be an arbitrary sub-CSP of $\Phi|_\sigma$ where $V'\subseteq V$ and $\Ccal'\subseteq\Ccal|_\sigma$.

If $\Naturale\alpha\Delta\le1$, then the following holds for \RejectionSampling{$\Phi',r$} over random $r$.
\begin{itemize}
\item[(1)] Let $X$ be its running time. Then $X<+\infty$ almost surely and, if hypergraph $H(\Phi')$ (Defined in \Cref{eq:csp_to_hypergraph}) is connected, we have
$$
\E[X]=O\pbra{\frac{kQ|\Ccal'|+Q}{\pbra{1-\Naturale\alpha}^{\abs{\Ccal'}}}}
\quad\text{and}\quad
\E\sbra{X^2}=O\pbra{\frac{\pbra{kQ|\Ccal'|+Q}^2}{\pbra{1-\Naturale\alpha}^{2\cdot\abs{\Ccal'}}}}.
$$
\item[(2)] Let $\sigma'$ be its output. Then $\sigma'\sim\mu_\True^{\Ccal'}$. Moreover for any $v\in V'$ and any $q\in\Omega_v|_\sigma$, we have
$$
\Pr\sbra{\sigma'(v)=q}\le\beta\cdot\Dcal_v|_\sigma(q).
$$
\end{itemize}
\end{proposition}
\begin{proof}
Firstly by \Cref{fct:after_projection}, $\Delta(\Phi')\le\Delta$. Then for any $C|_\sigma\in\Ccal'$, we have
\begin{align*}
&\phantom{=}\Pr_{\tilde\sigma\sim\prod_{v\in\vbl(C|_\sigma)}\Dcal_v|_\sigma}\sbra{C|_\sigma(\tilde\sigma)=\False}\\
&=\prod_{v\in\vbl(C|_\sigma)}\Dcal_v(\sigma_\False^C(v))
\tag{since $\Phi$ is atomic and by \Cref{def:projected_constraint_satisfaction_problem}}\\
&\le\prod_{v\in\vbl(C)\setminus\Mcal}\Dcal_v(\sigma_\False^C(v))
\tag{since $\sigma(V\setminus\Mcal)=\Qmark^{V\setminus\Mcal}$ and thus $\vbl(C|_\sigma)\subseteq V\setminus\Mcal$}\\
&\le\alpha.
\tag{by \Cref{eq:alpha}}
\end{align*}
Thus $p(\Phi')\le\alpha$ (Recall $p$ from \Cref{def:(atomic)_constraint_satisfaction_problem}) and we apply \Cref{thm:LLL} to obtain 
$$
\Pr_{\tilde\sigma\sim\prod_{v\in V'}\Dcal_v|_\sigma}\sbra{\tilde\sigma\in\sigma_\True^{\Ccal'}}\ge\pbra{1-\Naturale\alpha}^{\abs{\Ccal'}}>0.
$$
By \Cref{fct:after_projection}, $k(\Phi')\le k$.
Then Item (1) follows immediately from \Cref{fct:rejectionsampling} by noticing $|V'|\le k|\Ccal'|+1$\footnote{The additional $+1$ is for the case $|V'|=1$ and $\Ccal'=\emptyset$.} when $H(\Phi')$ is connected. $\sigma'\sim\mu_\True^{\Ccal'}$ in Item (2) follows from \Cref{fct:rejectionsampling} as well.

Note that $\beta\ge1$. Thus by \Cref{def:projected_constraint_satisfaction_problem} we may safely assume $\sigma(v)=\Qmark$ in Item (2). Let $B(\sigma')$ be the event (i.e., constraint) ``$\sigma'(v)=q$''. Then $\vbl(B)=\cbra{v}$ and 
$$
\Pr_{\sigma'\sim\mu_\True^{\Ccal'}}\sbra{B(\sigma')}=\Pr_{\sigma'\sim\mu_\True^{\Ccal'}}\sbra{\sigma'(v)=q}
\quad\text{and}\quad
\Pr_{\sigma'\sim\prod_{v\in V'}\Dcal_v|_\sigma}\sbra{B(\sigma')}=\Dcal_v|_\sigma(q).
$$ 
By \Cref{fct:after_projection}, $d(\Phi')\le d$.
Thus Item (2) follows naturally from \Cref{thm:HSS_LLL} by the definition of $\beta$ and noticing
\begin{equation}
\abs{\cbra{C|_\sigma\in\Ccal'\mid\vbl(B)\cap\vbl(C|_\sigma)\neq\emptyset}}=\abs{\cbra{C|_\sigma\in\Ccal'\mid v\in\vbl(C|_\sigma)}}\le d(\Phi')\le d.
\tag*{\qedhere}
\end{equation}
\end{proof}

\subsection{A \textsf{SafeSampling} Subroutine}\label{sec:a_safesampling_subroutine}

The following simple \textsf{SafeSampling} complements \textsf{RejectionSampling} when there is uncertainty to update $u\in\Mcal$. Details will be clear in \Cref{sec:coupling_from_the_past_and_the_bounding_chain}.

\begin{algorithm}[ht]
\caption{The \textsf{SafeSampling} algorithm}\label{alg:safesampling}
\DontPrintSemicolon
\LinesNumbered
\KwIn{an atomic CSP $\Phi=\pbra{V,\pbra{\Omega_v,\Dcal_v}_{v\in V},\Ccal}$, a marking $\Mcal\subseteq V$, some $u\in\Mcal$, and a randomness tape $r$}
\KwOut{some value $q\in\Omega_u^\Qmark$}
Recall $\beta=\beta(\Phi,\Mcal)$ from \Cref{eq:beta}\;
Sample $c$ from $\Dcal_u^\Qmark$ using $r$ where $\Dcal_u^\Qmark$ is a distribution over $\Omega_u^\Qmark$ by
$$
\Dcal_u^\Qmark(q)=\begin{cases}
\max\cbra{0,1-\beta\cdot\pbra{1-\Dcal_u(q)}} & q\in\Omega_u,\\
1-\sum_{q'\in\Omega_u}\Dcal_u^\Qmark(q') & q=\Qmark.
\end{cases}
$$\;
\vspace{-20pt}
\Return{$c$}
\end{algorithm}

We note the following observation regarding \textsf{SafeSampling}.
\begin{proposition}\label{prop:safesampling}
The following holds for \SafeSampling{$\Phi,\Mcal,u,r$} over random $r$ if $\Naturale\alpha\Delta\le1$. 
\begin{itemize}
\item[(1)] It runs in time $O(Q)$ where $Q=Q(\Phi)$ from \Cref{def:(atomic)_constraint_satisfaction_problem} and $\Dcal_u^\Qmark$ from \textsf{Line 2} is a distribution.
\item[(2)] For any $q\in\Omega_u$, $\Dcal_u^\Qmark(q)\le\Dcal_u(q)$ and $\Dcal_u^\Qmark(\cbra{q,\Qmark})\le\beta\cdot\Dcal_u(q)+(\beta-1)(|\Omega_u|-2)$.
\item[(3)] Let $\sigma\in\pbra{\prod_{v\in\Mcal}\Omega_v^\Qmark}\times\cbra{\Qmark}^{V\setminus\Mcal}$ be an arbitrary assignment with $\sigma(u)=\Qmark$. Let $\Phi'=\pbra{V',\pbra{\Omega_v|_\sigma,\Dcal_v|_\sigma}_{v\in V'},\Ccal'}$ be an arbitrary sub-CSP of $\Phi|_\sigma$ where $V'\subseteq V$ and $\Ccal'\subseteq\Ccal|_\sigma$.
Let $\sigma'\sim\mu_\True^{\Ccal'}$.\footnote{$\mu_\True^{\Ccal'}$ is well-defined guaranteed by Item (2) of \Cref{prop:rejectionsampling}.} Then for any $q\in\Omega_u|_\sigma=\Omega_u$, we have $\Dcal_u^\Qmark(q)\le\Pr\sbra{\sigma'(u)=q}$.
\end{itemize}
\end{proposition}
\begin{proof}
First note that $\beta\ge1$.
For Item (2), observe that
$$
\Dcal_u^\Qmark(q)=\max\cbra{0,1-\beta\cdot\pbra{1-\Dcal_u(q)}}\le\max\cbra{0,1-1\cdot\pbra{1-\Dcal_u(q)}}=\Dcal_u(q)
$$
and
\begin{align*}
\Dcal_u^\Qmark(\cbra{q,\Qmark})
&=1-\sum_{q'\in\Omega_u\setminus\cbra{q}}\Dcal_u^\Qmark(q')
\le1-\sum_{q'\in\Omega_u\setminus\cbra{q}}\pbra{1-\beta\cdot(1-\Dcal_u(q'))}\\
&=1+(\beta-1)(|\Omega_u|-1)-\beta\cdot\sum_{q'\in\Omega_u\setminus\cbra{q}}\Dcal_u(q')\\
&=1-(\beta-1)(|\Omega_u|-1)-\beta\cdot(1-\Dcal_u(q))
=\beta\cdot\Dcal_u(q)+(\beta-1)(|\Omega_u|-2).
\end{align*}

For Item (1), it suffices to observe, using Item (2), that $\Dcal_u^\Qmark(\Qmark)\ge1-\sum_{q'\in\Omega_u}\Dcal_u(q')=0$.

Finally for Item (3), we simply have
\begin{align*}
\Pr\sbra{\sigma'(u)=q}
&=1-\sum_{q'\in\Omega_u\setminus\cbra{q}}\Pr\sbra{\sigma'(u)=q'}\\
&\ge1-\beta\cdot\sum_{q'\in\Omega_u\setminus\cbra{q}}\Dcal_u(q')
\tag{by Item (2) of \Cref{prop:rejectionsampling} and $\Dcal_u|_\sigma=\Dcal_u$}\\
&=1-\beta\cdot(1-\Dcal_u(q))\\
&\ge\Dcal_u(q)\ge\Dcal_u^\Qmark(q).
\tag{since $\beta\ge1$ and by Item (2)}
\end{align*}
\end{proof}

\section{The \textsf{AtomicCSPSampling} Algorithm}\label{sec:the_atomiccspsampling_algorithm}

We now formally describe our main algorithm \textsf{AtomicCSPSampling} in \Cref{alg:atomiccspsampling}. The missing subroutines will be provided as we prove the correctness and efficiency of \Cref{alg:atomiccspsampling}. 

Intuitively, the $\tilde\sigma$ after the \textsf{while} iterations will be a random assignment over $\Mcal$ (i.e., $\tilde\sigma\in\pbra{\prod_{v\in\Mcal}\Omega_v}\times\cbra{\Qmark}^{V\setminus\Mcal}$) with certain distribution; and the final output $\sigma$ simply extends the assignment of $\tilde\sigma$ to $V\setminus\Mcal$. Putting them together, we will prove $\sigma$ is distributed as $\mu_\True^\Ccal$.

\begin{algorithm}[ht]
\caption{The \textsf{AtomicCSPSampling} algorithm}\label{alg:atomiccspsampling}
\DontPrintSemicolon
\LinesNumbered
\KwIn{an atomic CSP $\Phi=\pbra{V,\pbra{\Omega_v,\Dcal_v}_{v\in V},\Ccal}$ and a marking $\Mcal\subseteq V$}
\KwOut{an assignment $\sigma\in\sigma_\True^\Ccal$}
Assign infinitely long randomness $r_i$ independently for each $i\in\Zbb$\;
Initialize $T\gets1$\;
\While{$\True$}{
$\tilde\sigma\gets\BoundingChain{$\Phi,\Mcal,-T,r_{-T},\ldots,r_{-1}$}$
\tcc*{$\tilde\sigma\in\pbra{\prod_{v\in\Mcal}\Omega_v^\Qmark}\times\cbra{\Qmark}^{V\setminus\Mcal}$}
\lIf{$\tilde\sigma(v)\neq\Qmark$ for all $v\in\Mcal$}{break}
\lElse{Update $T\gets2\cdot T$}
}
$\sigma\gets\FinalSampling{$\Phi,\Mcal,\tilde\sigma$}$\;
\Return{$\sigma$}
\end{algorithm}

Recall measures $\alpha,\rho,\lambda,k,d,\Delta,Q$ defined in \Cref{eq:alpha}, \Cref{eq:rho}, \Cref{eq:lambda}, and \Cref{def:(atomic)_constraint_satisfaction_problem}.
Now we present our main theorem, the proof of which is the focus of the rest of the section.
\begin{theorem}\label{thm:atomiccspsampling}
Let $\Phi=\pbra{V,\pbra{\Omega_v,\Dcal_v}_{v\in V},\Ccal}$ be an atomic CSP. Let $\Mcal\subseteq V$ be a marking. 
If $\Naturale\alpha\Delta\le1$, $\Naturale\Delta^2\rho\le1/32$, and $\Delta^2\lambda\le1/16$, then the following holds for \AtomicCSPSampling{$\Phi,\Mcal$}.
\begin{itemize}
\item \textsc{Correctness.} It halts almost surely and outputs $\sigma\sim\mu_\True^\Ccal$ when it halts.
\item \textsc{Efficiency.} Its expected total running time is $O\pbra{kQ\Delta^3|V|\log(|V|)}$.
\end{itemize}
\end{theorem}
\begin{remark}\label{rmk:atomiccspsampling}
We remark that $\lambda\ge\rho$ since $\beta\ge1$ and domains have at least $2$ elements\footnote{If some domain has only $1$ element, then the corresponding distribution must be a point distribution. Thus we can simply fix the variable to this value and simplify the CSP.}. Thus we can use, say, $\Delta^2\lambda\le1/100$ to dominate both $\Naturale\Delta^2\rho\le1/32$ and $\Delta^2\lambda\le1/16$ and thus simplify the conditions in \Cref{thm:atomiccspsampling}. This indeed only loses minor factors.
However, we choose to present in the current format to make the proofs cleaner as $\Naturale\Delta^2\rho$ and $\Delta^2\lambda$ comes from different places.
\end{remark}

\subsection{The \textsf{BoundingChain} Subroutine}\label{sec:the_boundingchain_subroutine}

Recall our \textsf{SafeSampling} subroutine from \Cref{sec:a_safesampling_subroutine}.
We present the \textsf{BoundingChain} subroutine.

\begin{algorithm}[ht]
\caption{The \textsf{BoundingChain} subroutine}\label{alg:boundingchain}
\DontPrintSemicolon
\LinesNumbered
\KwIn{an atomic CSP $\Phi=\pbra{V,\pbra{\Omega_v,\Dcal_v}_{v\in V},\Ccal}$, a marking $\Mcal\subseteq V$, a starting time $-T$, and randomness tapes $r_{-T},\ldots,r_{-1}$}
\KwOut{an assignment $\sigma\in\prod_{v\in V}\Omega_v^\Qmark$}
Initialize $\sigma(v)\gets\Qmark$ for all $v\in V$
\tcc*{Assume $V=\cbra{v_0,\ldots,v_{n-1}}$}
\For{$t=-T$ \KwTo $-1$}{
$i_t\gets t\mod n$, and $\sigma(v_{i_t})\gets\Qmark$
\tcc*{Update $\sigma(v_{i_t})$ in this round}
$(\Phi_t,\Token_t)\gets\Component{$\Phi,\Mcal,\sigma,v_{i_t}$}$ where $\Phi_t=\pbra{V_t,\pbra{\Omega_v|_\sigma,\Dcal_v|_\sigma}_{v\in V_t},\Ccal_t}$\;
\lIf{$v_{i_t}\in V\setminus\Mcal$}{Continue}
\uElseIf{$(v_{i_t}\in\Mcal)\land(\Token_t=\True)$}{
$\sigma'\gets\RejectionSampling{$\Phi_t,r_t$}$\;
Update $\sigma(v_{i_t})\gets\sigma'(v_{i_t})$
}
\Else(\tcc*[f]{$(v_{i_t}\in\Mcal)\land(\Token_t=\False)$}){
Update $\sigma(v_{i_t})\gets\SafeSampling{$\Phi,\Mcal,v_{i_t},r_t$}$
}
}
\Return{$\sigma$}
\end{algorithm}

\begin{remark}
\Cref{alg:boundingchain} also works if we ignore variables outside $\Mcal$. This is because $\sigma(V\setminus\Mcal)$ is always kept $\Qmark^{V\setminus\Mcal}$. However the current version is more convenient for our analysis in \Cref{sec:concentration_bounds_for_the_coalescence} and does not influence the running time much.
\end{remark}

We first note the following fact regarding each round of update.
\begin{lemma}\label{lem:single_update}
Let $t\in\cbra{-T,\ldots,-1}$ be an arbitrary time with $v_{i_t}\in\Mcal$. Let $\sigma\in\pbra{\prod_{v\in\Mcal}\Omega_v^\Qmark}\times\cbra{\Qmark}^{V\setminus\Mcal}$ be an arbitrary assignment. 
Let $q\in\Omega_{v_{i_t}}^\Qmark$ be the update of $\sigma(v_{i_t})$ in the $t$-th \textsf{for} iteration of \Cref{alg:boundingchain} over random $r_t$.
If $\Naturale\alpha\Delta\le1$, then 
\begin{itemize}
\item[(1)] $q$ is well-defined almost surely;
\item[(2)] for any $q'\in\Omega_{v_{i_t}}$ we have 
$$
\Pr\sbra{q=q'\mid\sigma,t}\le\beta\cdot\Dcal_{v_{i_t}}(q')
\quad\text{and}\quad
\Pr\sbra{q\in\cbra{q',\Qmark}\mid\sigma,t}\le\beta\cdot\Dcal_{v_{i_t}}(q')+(\beta-1)\pbra{\abs{\Omega_{v_{i_t}}}-2}.
$$
\end{itemize}
\end{lemma}
\begin{proof}
Note that $\sigma$ uniquely determines $\Token_t$, i.e., whether we update $\sigma(v_{i_t})$ by \textsf{RejectionSampling} or \textsf{SafeSampling}. Thus we only need to consider two possibilities separately. 
\begin{itemize}
\item \fbox{\textsf{RejectionSampling}.} By Item (3) of \Cref{prop:component} and Item (1) of \Cref{prop:rejectionsampling}, $q$ is well-defined almost surely.
Since $\sigma(v_{i_t})$ is set to $\Qmark$ right before the update and $q$ is never $\Qmark$ in \textsf{RejectionSampling}, we have $\Dcal_{v_{i_t}}|_\sigma=\Dcal_{v_{i_t}}$ and, by Item (2) of \Cref{prop:rejectionsampling},
$$
\Pr\sbra{q\in\cbra{q',\Qmark}\mid\sigma,t}=\Pr\sbra{q=q'\mid\sigma,t}\le\beta\cdot\Dcal_{v_{i_t}}(q').
$$
\item \fbox{\textsf{SafeSampling}.} By Item (1) of \Cref{prop:safesampling}, $q$ is always well-defined. Then the two bounds follow from Item (2) of \Cref{prop:safesampling}.
\qedhere
\end{itemize}
\end{proof}

We will show the following result for \emph{one single call} of \Cref{alg:boundingchain}.
\begin{proposition}\label{prop:boundingchain}
If $\Naturale\alpha\Delta\le1$, $\Naturale\Delta^2\rho\le1/32$, and $\Delta^2\lambda\le1/16$, then the following holds over random $r_{-T},\ldots,r_{-1}$ for \BoundingChain{$\Phi,\Mcal,-T,r_{-T},\ldots,r_{-1}$}.
\begin{itemize}
\item \textsc{Efficiency.} Let $X_t$ be the running time of the $t$-th \textsf{for} iteration. Then $X_t<+\infty$ almost surely and $\E\sbra{X_t^2}=O\pbra{dk^2\Delta^5Q^2}$.
\item \textsc{Coalescence.} Let $\Ecal$ be the event ``in the returned assignment $\sigma$, there exists some $u\in\Mcal$ such that $\sigma(u)=\Qmark$''. If $T\ge2|V|-1$, we have $\Pr\sbra{\Ecal}\le4|V|\cdot 2^{-T/|V|}$.
\end{itemize}
\end{proposition}

\subsubsection{Moment Bounds on the Running Time}\label{sec:moment_bounds_on_the_running_time}

To establish the efficiency part of \Cref{prop:boundingchain}, we need to control the size of $\Phi_t$ in each \textsf{for} iteration. This requires some additional definitions.

\begin{definition}[$2$-tree]\label{def:2-tree}
Let $G$ be an undirected graph. A set of vertices $S\subseteq V(G)$ is a \emph{$2$-tree} if the following holds.
\begin{itemize}
\item $\dist_G(u,v)\ge2$ holds for any distinct $u,v\in S$ where $\dist_G(u,v)$ is the length of the shortest path in $G$ from $u$ to $v$.\footnote{For example $\dist_G(u,u)\equiv 0$ for all $u\in V(G)$, and $\dist_G(u,v)=1$ iff $(u,v)$ is an edge in $E(G)$.}
\item If we add an edge between every $u,v\in S$ with $\dist_G(u,v)=2$, then $S$ is connected.
\end{itemize}
\end{definition}
Intuitively a $2$-tree is an independent set that is not very spread out. The following lemmas bound the number of $2$-trees and show how to extract a large $2$-tree from any connected subgraph.
\begin{lemma}[{\cite[Corollary 5.7]{FGYZ20}}]\label{lem:2-tree_number}
Let $G$ be a graph with maximum degree $d$. Then for any $v\in V(G)$ and integer $\ell\ge1$, the number of $2$-trees in $G$ of size $\ell$ containing $v$ is at most $\pbra{\Naturale d^2}^{\ell-1}/2$.
\end{lemma}
\begin{lemma}[{\cite[Lemma 4.5]{Vishesh21sampling}}]\label{lem:2-tree_size_JPV}
Let $G$ be a graph with maximum degree $d$. Let $G'$ be a connected subgraph of $G$. Then for any $v\in V(G')$, there exists a $2$-tree $S\subseteq V(G')$ with $v\in S$ and size $|S|\ge|V(G')|/(d+1)$.
\end{lemma}
\begin{lemma}[{\cite[Observation 5.5]{FGYZ20}}]\label{lem:2-tree_size_FGYZ}
If a graph $G$ has a $2$-tree of size $\ell>1$ containing $v\in V(G)$, then $G$ also has a $2$-tree of size $\ell-1$ containing $v$.
\end{lemma}
The following result is an immediate corollary of \Cref{lem:2-tree_size_JPV} and \Cref{lem:2-tree_size_FGYZ}.
\begin{corollary}\label{cor:2-tree_size}
Let $G$ be a graph with maximum degree $d$. Let $G'$ be a connected subgraph of $G$. Then for any $v\in V(G')$ and any integer $\ell\le\lceil|V(G')|/(d+1)\rceil$, there exists a $2$-tree $S\subseteq V(G')$ with $v\in S$ and size $|S|=\ell$.
\end{corollary}

Now we show the following concentration bound.
\begin{proposition}\label{prop:component_size}
Let $d=d(\Phi)$, $\Delta=\Delta(\Phi)$, $\alpha=\alpha(\Phi,\Mcal)$, and $\rho=\rho(\Phi,\Mcal)$. Assume $\Naturale\alpha\Delta\le1$.

For any $t\in\cbra{-T,\ldots,-1}$, recall $\Ccal_t$ from \textsf{Line 4} of \Cref{alg:boundingchain}.
Then we have
$$
\Pr\sbra{\abs{\Ccal_t}\ge\ell\cdot\Delta}\le\frac d2\cdot\pbra{\Naturale\Delta^2\rho}^{\ell-1}
\quad\text{for any integer $\ell\ge1$}.
$$
\end{proposition}
\begin{proof}
Construct the \emph{line graph} $\Lin(\Phi)=\pbra{V^\Phi,E^\Phi}$ of $\Phi=\pbra{V,\pbra{\Omega_v,\Dcal_v}_{v\in V},\Ccal}$ where 
$$
V^\Phi=\Ccal
\quad\text{and}\quad
E^\Phi=\cbra{\cbra{e_1,e_2}\in\Ccal\times\Ccal\mid\vbl(e_1)\cap\vbl(e_2)\neq\emptyset,e_1\neq e_2}.
$$ 
Then $\Lin(\Phi)$ is an undirected graph with maximum degree $\Delta-1$.

Let $\sigma$ and $v_{i_t}$ be the assignment and variable to update at time $t$ respectively.
Let $G$ be the subgraph of $\Lin(\Phi)$ induced by vertex set $V(G)=\cbra{C\in\Ccal\mid C|_\sigma\in\Ccal_t}$. Then by Item (3) of \Cref{prop:component}, $G$ is a connected subgraph of $\Lin(\Phi)$.\footnote{Actually Item (3) of \Cref{prop:component} says the subgraph of $\Lin(\Phi|_\sigma)$ induced by vertex set $\Ccal_t$ is connected, which implies $G$ is a connected subgraph of $\Lin(\Phi)$ as $\vbl(C|_\sigma)\subseteq\vbl(C)$.}
For any $C|_\sigma\in\Ccal_t$ with $v_{i_t}\in\vbl(C)$, by \Cref{cor:2-tree_size} there exists a $2$-tree $\tilde S\subseteq V(G)$ with $C\in\tilde S$ and size $|\tilde S|=\ell$ provided $\ell\le\lceil \abs{\Ccal_t}/\Delta\rceil$. 
Define 
$$
\Scal=\cbra{\text{$2$-tree }S\subseteq V^\Phi\mid\pbra{|S|=\ell}\land\pbra{\exists C\in S,v_{i_t}\in\vbl(C)}}.
$$
Then by \Cref{lem:2-tree_number} and noticing there are at most $d$ choices of $C$, we have
$$
\abs{\Scal}\le\frac{d\cdot\pbra{\Naturale\pbra{\Delta-1}^2}^{\ell-1}}2\le\frac{d\cdot\pbra{\Naturale\Delta^2}^{\ell-1}}2.
$$

By Item (2) of \Cref{prop:component}, $\sigma(v)=\sigma_\False^C(v)$ for any $C|_\sigma\in\Ccal_t$ and $v\in\pbra{\vbl(C)\cap\Mcal}\setminus\cbra{v_{i_t}}$.
Note that $\sigma(v)$ is initialized as $\Qmark$. Thus $\sigma(v)$ must be updated before time $t$. 
Let $\UpdTime(v,t)$ be the last update time of $v$ before the $t$-th \text{for} iteration in \Cref{alg:boundingchain}, and let $\Ecal_{v,C}$ be the event ``$\sigma(v)$ is updated to $\sigma_\False^C(v)$ in the $\UpdTime(v,t)$-th \textsf{for} iteration''.
Recall the definition of $\rho$ from \Cref{eq:rho}, then we have\footnote{We remark that for the fourth inequality below, we do \emph{not} assume any independence between $\Ecal_{v,C}$. We simply use the chain rule of conditional probability in the order of time. For example, if $\UpdTime(v_1,t)<\UpdTime(v_2,t)$, then $\Pr\sbra{\Ecal_{v_1,C_1}\land\Ecal_{v_2,C_2}}=\Pr\sbra{\Ecal_{v_1,C_1}}\cdot\Pr\sbra{\Ecal_{v_2,C_2}\mid\Ecal_{v_1,C_1}}$ and then we apply Item (2) of \Cref{lem:single_update} twice.}
\begin{align*}
\Pr\sbra{\abs{\Ccal_t}\ge\ell\cdot\Delta}
&\le\Pr\sbra{\Ecal_{v,C},\forall C|_\sigma\in\Ccal_t,\forall v\in\pbra{\vbl(C)\cap\Mcal}\setminus\cbra{v_{i_t}}}\\
&\le\Pr\sbra{\Ecal_{v,C},\forall C\in\tilde S,\forall v\in\pbra{\vbl(C)\cap\Mcal}\setminus\cbra{v_{i_t}}}\\
&\le\sum_{S\in\Scal}\Pr\sbra{\Ecal_{v,C},\forall C\in S,\forall v\in\pbra{\vbl(C)\cap\Mcal}\setminus\cbra{v_{i_t}}}
\tag{by union bound}\\
&\le\sum_{S\in\Scal}\prod_{C\in S}\prod_{v\in\pbra{\vbl(C)\cap\Mcal}\setminus\cbra{v_{i_t}}}\min\cbra{1,\beta\cdot\Dcal_v(\sigma_\False^C(v))}
\tag{since $\pbra{\vbl(C)}_{C\in S}$ are pairwise disjoint and by Item (2) of \Cref{lem:single_update}}\\
&\le\sum_{S\in\Scal}\prod_{C\in S,v_{i_t}\notin\vbl(C)}\prod_{v\in\vbl(C)\cap\Mcal}\beta\cdot\Dcal_v(\sigma_\False^C(v))\\
&\le\sum_{S\in\Scal}\rho^{|S|-1}
\le\frac d2\cdot\pbra{\Naturale\Delta^2\rho}^{\ell-1}.
\tag*{\qedhere}
\end{align*}
\end{proof}

Now we obtain moment bounds for the running time of each \textsf{for} iteration in \Cref{alg:boundingchain}.
\begin{proof}[Proof of the Efficiency Part of \Cref{prop:boundingchain}]
Let $Y_t$ and $Z_t$ be the running time of \textsf{Line 4} and \textsf{Line 5-11} respectively.
By Item (1) of \Cref{prop:component}, we have $\E\sbra{Y_t^2\mid\abs{\Ccal_t}}=O\pbra{\pbra{\Delta k\abs{\Ccal_t}+dk}^2}$.
By Item (3) of \Cref{prop:component}, Item (1) of \Cref{prop:rejectionsampling}, and Item (1) of \Cref{prop:safesampling}, we also have
$$
\E\sbra{Z_t^2\mid\abs{\Ccal_t}}=O\pbra{\max\cbra{1,Q^2,\frac{\pbra{kQ\abs{\Ccal_t}+Q}^2}{(1-\Naturale\alpha)^{2\cdot\abs{\Ccal_t}}}}}=O\pbra{\frac{\pbra{kQ\abs{\Ccal_t}+Q}^2}{(1-\Naturale\alpha)^{2\cdot\abs{\Ccal_t}}}}.
$$
By \Cref{prop:component_size}, we have 
$$
\Pr\sbra{\abs{\Ccal_t}\ge\ell\cdot\Delta}\le\frac d2\cdot\pbra{\Naturale\Delta^2\rho}^{\ell-1}\le\frac d2\cdot\pbra{\frac1{32}}^{\ell-1}
\quad\text{for any integer }\ell\ge1.
$$
Since $X_t=Y_t+Z_t+O(1)$ and $(a+b+c)^2\le4\cdot(a^2+b^2+c^2)$, we have 
\begin{align*}
\E\sbra{X_t^2}
&=O\pbra{1+\E\sbra{Y_t^2}+\E\sbra{Z_t^2}}
=O\pbra{1+\sum_{L=0}^{+\infty}\Pr\sbra{\abs{\Ccal_t}=L}\E\sbra{Y_t^2+Z_t^2\mid\abs{\Ccal_t}=L}}\\
&\le O\pbra{d^2k^2}+\sum_{\ell=0}^{+\infty}\Pr\sbra{\abs{\Ccal_t}\ge\ell\cdot\Delta}\cdot\Delta\cdot O\pbra{\Delta^2k^2\pbra{\Delta(\ell+1)}^2+\frac{\pbra{Qk\Delta(\ell+1)}^2}{\pbra{1-\Naturale\alpha}^{2(\ell+1)\Delta}}}
\tag{bucketing $L\in[\ell\cdot\Delta,(\ell+1)\cdot\Delta)$}\\
&\le O\pbra{d^2k^2}+\sum_{\ell=0}^{+\infty}\Pr\sbra{\abs{\Ccal_t}\ge\ell\cdot\Delta}\cdot\Delta\cdot O\pbra{\Delta^2k^2\pbra{\Delta(\ell+1)}^2+\pbra{Qk\Delta(\ell+1)}^2\cdot16^{\ell+1}}
\tag{since $\Naturale\alpha\ge1/\Delta$ and $\Delta\ge2$, we have $(1-\Naturale\alpha)^\Delta\ge1/4$}\\
&\le O\pbra{d^2k^2}+O\pbra{d\Delta}\sum_{\ell=1}^{+\infty}\pbra{\frac1{32}}^{\ell-1}\pbra{\Delta^2k^2\pbra{\Delta(\ell+1)}^2+\pbra{Qk\Delta(\ell+1)}^2\cdot16^{\ell+1}}\\
&=O\pbra{d^2k^2+dk^2\Delta^5+dk^2\Delta^3Q^2}=O\pbra{dk^2\Delta^5Q^2}.
\tag{since $d\le\Delta$}
\end{align*}
\end{proof}

\subsubsection{Concentration Bounds for the Coalescence}\label{sec:concentration_bounds_for_the_coalescence}

Here we analyze the coalescence part of \Cref{prop:boundingchain}, which is also the stopping condition for the \textsf{while} iterations in \Cref{alg:atomiccspsampling}. We use \emph{information percolation} argument and need additional setup. 

We follow the notation convention in \Cref{sec:moment_bounds_on_the_running_time}: $V=\cbra{v_0,\ldots,v_{n-1}}$ and $i_t=t\mod n$ for $t\in\cbra{-T,\ldots,-1}$.  
$\UpdTime(v,t)$ is the last update time of $v$ \emph{before} time $t$, i.e.,
\begin{equation}\label{eq:updtime}
\UpdTime(v,t)=\max\cbra{-T-1,\max\cbra{t'<t\mid v_{i_{t'}}=v}}.
\end{equation}
The additional $-T-1$ is to set up the boundary condition corresponding to the initialization step.
\begin{definition}[Extended Constraints]\label{def:extended_constraints}
For any $C\in\Ccal$ and $e=\cbra{t_1,\ldots,t_m}\subseteq\cbra{-T,\ldots,-1}$, we say $(e,C)$ is an \emph{extended constraint} if the following holds.
\begin{itemize}
\item[(1)] $m=|\vbl(C)|$ and $\vbl(C)=\cbra{v_{i_{t_1}},v_{i_{t_2}},\ldots,v_{i_{t_m}}}$.
\item[(2)] $e=\cbra{\UpdTime(v,t_\mathsf{max}+1)\mid v\in\vbl(C)}$ where $t_\mathsf{max}=\max_{t'\in e}t'$.
\end{itemize}
\end{definition}
Intuitively an extended constraint of $C$ is a consecutive rounds of updates in $\vbl(C)$.
\begin{fact}\label{fct:extended_constraints}
The following holds for extended constraints.
\begin{itemize}
\item[(1)] If $(e_1,C_1)$ and $(e_2,C_2)$ are two extended constraints and $\vbl(C_1)\cap\vbl(C_2)=\emptyset$, then $e_1\cap e_2=\emptyset$.
\item[(2)] If $(e,C)$ is an extended constraint, then 
\begin{itemize}
\item[(2a)] $0\le t_\mathsf{max}(e)-t_\mathsf{min}(e)<n$, where $t_\mathsf{min}(e)=\min_{t'\in e}t'$;
\item[(2b)] $e=\cbra{\UpdTime(v,t')\mid v\in\vbl(C)}$ for any $t_\mathsf{max}+1\le t'\le t_\mathsf{min}+n$;
\item[(2c)] for any $C'\in\Ccal$, we have 
$$
\abs{\cbra{e'\mid (e',C')\text{ is an extended constraint with }e\cap e'\neq\emptyset}}<2\cdot\abs{\vbl(C')}.
$$
\end{itemize}
\end{itemize}
\end{fact}
\begin{proof}
Item (1) is evident from Item (1) of \Cref{def:extended_constraints}. 

For Item (2), we assume $\abs{\vbl(C)}=m$ and $\vbl(C)=\cbra{v_{a_1},\ldots,v_{a_m}}$. Let 
$$
S=\cbra{-T,\ldots,-1}\cap\cbra{a_i-j\cdot n\mid i\in[m],j\in\Zbb}=\cbra{b_1,b_2,\ldots,b_{T(C)}}
$$
where $-T\le b_1<\cdots<b_{T(C)}\le-1$. Note that $b_i\equiv b_{i+m}\mod n$ for all $i\in[T(C)-m]$. If $(e,C)$ is an extended constraint, then by Item (2) of \Cref{def:extended_constraints}, $e$ consists of a consecutive interval of $S$, i.e., $e=\cbra{b_o,b_{o+1},\ldots,b_{o+m-1}}$ for some $o\in[T(C)-m+1]$. Thus $t_\mathsf{max}(e)-t_\mathsf{min}(e)=b_{o+m-1}-b_o<n$ (since $b_i\equiv b_{i+m}\mod n$) which verifies Item (2a). 
Since either $b_o+n=b_{o+m}$ or $b_o+n\ge0$, we know $\UpdTime(v_{a_i},t')=\UpdTime(v_{a_i},b_{o+m-1}+1)$ for all $i\in[m]$ and $b_{o+m-1}+1\le t'\le b_o+n$ which verifies Item (2b).

Now we prove Item (2c).
By Item (1), assume without loss of generality $\vbl(C)\cap\vbl(C')\neq\emptyset$.
Let $m'=|\vbl(C')|$ and  $\vbl(C')=\cbra{v_{a_1'},\ldots,v_{a_{m'}'}}$ and define
$$
S'=\cbra{-T,\ldots,-1}\cap\cbra{a_i'-j\cdot n\mid i\in[m'],j\in\Zbb}=\cbra{b_1',\ldots,b_{T(C')}'}
$$
where $-T\le b_1'<\cdots<b_{T(C')}'\le-1$. Then similarly, we have $b_i'\equiv b_{i+m'}'\mod n$ for all $i\in[T(C')-m']$ and $e'=\cbra{b_{o'}',\ldots,b_{o'+m'-1}'}$ for some $o'\in[T(C')-m'+1]$. Let $i_\mathsf{min}=\min\cbra{i\in[T(C')]\mid b_i'\in e}$ and $i_\mathsf{max}=\max\cbra{i\in[T(C')]\mid b_i'\in e}$. Then $e\cap e'\neq\emptyset$ iff $i_\mathsf{min}\le o'+m'-1$ and $i_\mathsf{max}\ge o'$. 
Therefore there are at most $i_\mathsf{max}-i_\mathsf{min}+m'$ choices of $o'$. 
Since $b_{i_\mathsf{min}}',b_{i_\mathsf{max}}'\in e$, by Item (2a) we know $b_{i_\mathsf{max}}'-b_{i_\mathsf{min}}'<n$. Hence $i_\mathsf{max}<i_\mathsf{min}+m'$. In all, there are at most $2m'-1$ choices of $e'$.
\end{proof}

\begin{definition}[Extended Hypergraph]\label{def:extended_hypergraph}
\emph{Extended hypergraph} $H^\ext=\pbra{V^\ext,E^\ext}$ has vertex set $V^\ext=\cbra{-T,\ldots,-1}$ and extended constraints as edges:
$$
E^\ext=\cbra{e\subseteq\cbra{-T,\ldots,-1}\mid (e,C)\text{ is an extended constraint}}.
$$
Moreover, we label each edge $e$ with $C$ if it is added into $E^\ext$ by extended constraint $(e,C)$. We allow multiple occurrence of the same edge but the labels are different.
\end{definition}

Define $\sigma_0$ as the final returned assignment in \Cref{alg:boundingchain}.
For each $t\in\cbra{-T,\ldots,-1}$, let $\sigma_t$ be the assignment at \textsf{Line 4} of the $t$-th \textsf{for} iteration in \Cref{alg:boundingchain}. In particular $\sigma_t(v_{i_t})=\Qmark$ due to \textsf{Line 3}.
 
Now we present the following algorithm informally described in \Cref{sec:proof_overview} to sequentially find constraints that are not satisfied during the \textsf{BoundingChain} process. We remark that this algorithm is only for our analysis, and we do \emph{not} run it during \textsf{AtomicCSPSampling}.

\begin{algorithm}[ht]
\caption{Find failed constraints during the \textsf{BoundingChain} process}\label{alg:find_failed_constraints_during_the_boundingchain_process}
\DontPrintSemicolon
\LinesNumbered
\KwIn{assignments $\pbra{\sigma_t}_{t\in\cbra{-T,\ldots,0}}$ defined above and some $u\in\Mcal$ with $\sigma_0(u)=\Qmark$}
\KwOut{$H'=\pbra{V',E'}$ where $V'\subseteq V^\ext$ and $E'\subseteq E^\ext$}
Set $t_0\gets\UpdTime(u,0)$ and initialize $V'\gets\cbra{t_0},E'\gets\emptyset$\;
\FailedConstraints{$t_0$}\;
\Return{$\pbra{V',E'}$}\;
\Procedure{\FailedConstraints{$t$}}{
\lIf(\tcc*[f]{$(v_{i_t}\in\Mcal)\land(\Token_t=\False)$}){$t<-T+n-1$}{\Return{}}
Initialize $V_t\gets\cbra{v_{i_t}}$ and $\Ccal_t\gets\emptyset$\;
\While{$\exists C|_{\sigma_t}\in\Ccal|_{\sigma_t}\setminus\Ccal_t$ with $\vbl(C|_{\sigma_t})\cap V_t\neq\emptyset$}{
$e\gets\cbra{\UpdTime(v,t+1)\mid v\in\vbl(C)}$
\tcc*{$(e,C)$ is an extended constraint}
Update $\Ccal_t\gets\Ccal_t\cup\cbra{C|_{\sigma_t}}$ and $V_t\gets V_t\cup\vbl(C|_{\sigma_t})$\;
Update $E'\gets E'\cup\cbra{e}$ and $V'\gets V'\cup e$
\tcc*{$e$ is labeled by $C$}
}
\ForEach{$v\in(V_t\cap\Mcal)\setminus\cbra{v_{i_t}}$}{
\FailedConstraints{$\UpdTime(v,t)$}
}
}
\end{algorithm}

We have the following observation regarding \Cref{alg:find_failed_constraints_during_the_boundingchain_process}.
\begin{lemma}\label{lem:find_failed_constraints_during_the_boundingchain_process}
\Cref{alg:find_failed_constraints_during_the_boundingchain_process} halts always. Furthermore, if $T\ge2n-1$ then 
\begin{itemize}
\item[(1)] for each $(e,C)$ from \textsf{Line 8} when we execute \FailedConstraints{$t$}, 
\begin{itemize}
\item[(1a)] it is an extended constraint, 
\item[(1b)] for each $v\in\vbl(C)$, the assignment on $v$ is updated to $\Qmark$ or $\sigma_\False^C(v)$ in the $\UpdTime(v,t+1)$-th \textsf{for} iteration in \Cref{alg:boundingchain};
\end{itemize}
\item[(2)] each time we call \FailedConstraints{$t$}, $t$ is already in $V'$;
\item[(3)] $H'$ is a connected sub-hypergraph of $H^\ext$;
\item[(4)] there exists some $e_0,e_1\in E'$ such that $t_\mathsf{max}(e_0)\ge-n$ and $t_\mathsf{min}(e_1)<-T+n-1$.
\end{itemize}
\end{lemma}
\begin{proof}
Since $\UpdTime(v,t)<t$ for all $v\in V$ and $t\in\cbra{-T,\ldots,0}$, \Cref{alg:find_failed_constraints_during_the_boundingchain_process} always halts.

We prove Item (1) by induction on the calls of \FailedConstraints{$t$}. The first call $t_0$ represents the final update of the assignment on $u\in\Mcal$, which results in $\sigma_0(u)=\Qmark$. 
\begin{itemize}
\item \fbox{Item (1a) for $t_0$.} 
Note that $0>t_0\ge0-n\ge-T+n-1$. Then $\UpdTime(v,t_0+1)=t_0\ge-T$ for $v=v_{i_{t_0}}$; and $\UpdTime(v,t_0+1)\ge-T$ for all $v\neq v_{i_{t_0}}$. This means $-T-1\notin e$ and thus $(e,C)$ is an extended constraint.
\item \fbox{Item (1b) for $t_0$.} 
Since $C|_{\sigma_{t_0}}\in\Ccal|_{\sigma_{t_0}}$ at \textsf{Line 7}, we know $C(\sigma_{t_0})=\False$ and, by \Cref{eq:qmark_constraint}, $\sigma_{t_0}(v)\in\cbra{\sigma_\False^C(v),\Qmark}$ for each $v\in\vbl(C)$. 
This means, if $v\neq v_{i_{t_0}}$, the assignment on $v$ is updated to such value in the $\UpdTime(v,t_0)=\UpdTime(v,t_0+1)$-th \textsf{for} iteration in \Cref{alg:boundingchain}. 
Meanwhile if $v=v_{i_{t_0}}$, then the assignment on $v$ is updated to $\Qmark$ in this $t_0=\UpdTime(v,t_0+1)$-th \textsf{for} iteration, resulting in $\sigma_0(v)=\Qmark$.
\end{itemize}
To complete the induction, we note that each later call of \textsf{FailedConstraints} relies on some $v$ from \textsf{Line 13} when we execute some \FailedConstraints{$\UpdTime(v,t)$}. This means $v\in\Mcal\setminus\cbra{v_{i_t}}$ and $\sigma_t(v)=\Qmark$ and $t\ge-T+n-1$. Thus the assignment on $v$ is updated to $\Qmark$ in the $-T\le\UpdTime(v,t)$-th \textsf{for} iteration in \Cref{alg:boundingchain}. Then the argument above also goes through with almost no change.

For Item (2), note that $V'$ is initialized as $\cbra{t_0}$. Upon \textsf{Line 13}, we have $\UpdTime(v,t)=\UpdTime(v,t+1)$ since $v\neq v_{i_t}$, which has been added into $V'$ by \textsf{Line 10} earlier.

Now we turn to Item (3). By Item (1), $E'\subseteq E^\ext$. Meanwhile by \textsf{Line 10}, $E'$ are edges over vertex set $V'$. Thus it suffices to show $H'=(V',E')$ is connected. By Item (2), we only need to show vertices added during the \textsf{while} iterations in \FailedConstraints{$t$} is connected to $t$. 
This can be proved by induction:
When we find $C|_{\sigma_t}$ satisfying the condition at \textsf{Line 7}, fix some $v'\in\vbl(C|_{\sigma_t})\cap V_t$. 
Then for the edge $e$ constructed in the next step, each $\UpdTime(v,t+1)\in e$ is connected to $\UpdTime(v',t+1)\in e$.
Note that either (A) $v'=v_{i_t}$ and thus $\UpdTime(v',t+1)=t$, or (B) $\UpdTime(v',t+1)$ was added into $V'$ in an earlier time and connected to $t$. Hence each $\UpdTime(v,t+1)\in e$ is connected to $t$ as desired.

Finally we prove Item (4). 
Each time we call \FailedConstraints{$t$}, it implies the assignment on $v_{i_t}\in\Mcal$ is updated to $\Qmark$ in the $t$-th \textsf{for} iteration in \Cref{alg:boundingchain}.
This means $\Token_t=\False$ and \textsf{SafeSampling} is performed in \Cref{alg:boundingchain}. 
By comparing \Cref{alg:component} and \textsf{FailedConstraints}, we know $v_{i_t}$ is connected by falsified constraints to some $v\in\pbra{V_t\cap\Mcal}\setminus\cbra{v_{i_t}}$ that $\sigma_t(v)=\Qmark$. This implies at least one round of \textsf{Line 13} here will be executed. Thus the recursion of \Cref{alg:find_failed_constraints_during_the_boundingchain_process} continues until $t<-T+n-1$. 
By Item (2), there exists some $t_1\in V'$ with $t_1<-T+n-1$. Meanwhile $t_0\in V'$ and $t_0\ge-n$. 
Hence by Item (3), there exists $e_0,e_1\in E'$ such that $t_0\in e_0$ and $t_1\in e_1$; and $t_\mathsf{max}(e_0)\ge t_0\ge-n$, $t_\mathsf{min}(e_1)\le t_1<-T+n-1$.
\end{proof}

Finally we complete the proof of \Cref{prop:boundingchain}.
\begin{proof}[Proof of the Coalescence Part of \Cref{prop:boundingchain}]
Since $\sigma_0$ is defined as the final returned assignment, we know event $\Ecal$ (Defined in \Cref{prop:boundingchain}) is ``there exists some $u\in\Mcal$ such that $\sigma_0(u)=\Qmark$''. 
Using \Cref{alg:find_failed_constraints_during_the_boundingchain_process} we obtain $H'=(V',E')$. 
By \Cref{fct:path_in_connected_hypergraph} and Item (3) (4) of \Cref{lem:find_failed_constraints_during_the_boundingchain_process}, there exists a path $\tilde e_1,\tilde e_2,\ldots,\tilde e_\ell\in E'\subseteq E^\ext$ with labels $\tilde C_1,\tilde C_2,\ldots,\tilde C_\ell$ such that
\begin{itemize}
\item[(a)] $\tilde e_i\cap\tilde e_{i+1}\neq\emptyset$ and thus $t_\mathsf{min}(\tilde e_i)\le t_\mathsf{max}(\tilde e_{i+1})$ for all $i\in[\ell-1]$;
\item[(b)] $\tilde e_i\cap\tilde e_j=\emptyset$ for all $i,j\in[\ell]$ with $|i-j|>1$;
\item[(c)] $-n\le t_\mathsf{max}(\tilde e_1)<0$;
\item[(d)] $-T\le t_\mathsf{min}(\tilde e_\ell)<-T+n-1$.
\end{itemize}
Meanwhile
\begin{align*}
t_\mathsf{max}(\tilde e_1)
&=t_\mathsf{max}(\tilde e_\ell)+\sum_{i=1}^{\ell-1}\pbra{t_\mathsf{max}(\tilde e_i)-t_\mathsf{max}(\tilde e_{i+1})}\\
&\le t_\mathsf{min}(\tilde e_\ell)+n-1+\sum_{i=1}^{\ell-1}\pbra{t_\mathsf{min}(\tilde e_i)+n-1-t_\mathsf{max}(\tilde e_{i+1})}
\tag{by Item (2a) of \Cref{fct:extended_constraints}}\\
&\le \pbra{-T+n-2}+\ell\cdot(n-1).
\tag{by Item (a) (d) above}
\end{align*}
Thus by Item (c) above, we have $\ell\ge-2+\lceil T/(n-1)\rceil$. 
For convenience we truncate the tail of the path so that $\ell$ is the largest even number no more than $-2+\lceil T/(n-1)\rceil$. Thus $\ell\ge-3+\lceil T/(n-1)\rceil$.
We remark that since the tail is truncated, Item (d) above is not necessarily true now.

Now for any fixed path $e_1,\ldots,e_\ell\in E^\ext$ with labels $C_1,\ldots,C_\ell$ satisfying Item (a) (b) (c) above, we bound the probability that this path, denoted by $P$, appears in $H'$. Assume 
\begin{equation}\label{eq:even_positions_assumption}
\prod_{i=1}^{\ell/2}\abs{\vbl(C_{2i-1})}\le\prod_{i=1}^{\ell/2}\abs{\vbl(C_{2i})}
\end{equation}
and the other case can be analyzed similarly. 
For each $t\in\cbra{-T,\ldots,-1}$ and $C\in\Ccal$ with $v_{i_t}\in\vbl(C)$ we define $\Ecal_{t,C}$ as the event `` the assignment on $v_{i_t}$ is updated to $\Qmark$ or $\sigma_\False^C(v)$ in the $t$-th \textsf{for} iteration in \Cref{alg:boundingchain}''.
By Item (1b) of \Cref{lem:find_failed_constraints_during_the_boundingchain_process}, $e_i$, labeled by $C_i$, appearing in $H'$ implies $\bigwedge_{t\in e_i}\Ecal_{t,C_i}$ happens.
Recall the definition of $\lambda$ from \Cref{eq:lambda}, then we have\footnote{We remark that for the third inequality below, we do \emph{not} assume any independence between $\Ecal_{t,C}$. For example, if $t_1<t_2$, then $\Pr\sbra{\Ecal_{t_1,C_4}\land\Ecal_{t_2,C_2}}=\Pr\sbra{\Ecal_{t_1,C_4}}\cdot\Pr\sbra{\Ecal_{t_2,C_2}\mid\Ecal_{t_1,C_4}}$ and then we apply Item (2) of \Cref{lem:single_update} twice.}
\begin{align*}
\Pr\sbra{P\text{ appears in }H'}
&\le\Pr\sbra{e_{2i}\text{ appears in $H'$ with label $C_{2i}$ for all $i\in[\ell/2]$}}\\
&\le\Pr\sbra{\bigwedge_{i=1}^{\ell/2}\bigwedge_{t\in e_{2i}}\Ecal_{t,C_{2i}}}
=\Pr\sbra{\bigwedge_{i=1}^{\ell/2}\bigwedge_{t\in e_{2i},v_{i_t}\in\Mcal}\Ecal_{t,C_{2i}}}\\
&\le\prod_{i=1}^{\ell/2}\prod_{v\in\vbl(C_{2i})\cap\Mcal}\pbra{\beta\cdot\Dcal_v(\sigma_\False^C(v))+(\beta-1)(|\Omega_v|-2)}
\tag{since $(e_{2i})_{i\in[\ell/2]}$ are pairwise disjoint and by Item (2) of \Cref{lem:single_update}}\\
&\le\prod_{i=1}^{\ell/2}\pbra{4\Delta\abs{\vbl(C_{2i})}}^{-2}\le\prod_{i=1}^\ell\pbra{4\Delta\abs{\vbl(C_i)}}^{-1}.
\tag{since $\Delta^2\lambda\le1/16$ and by \Cref{eq:even_positions_assumption}}
\end{align*}

Now it suffices to union bound over all possible paths, i.e., $\Pr\sbra{\Ecal}\le\sum_P\Pr\sbra{P\text{ appears in }H'}$ where $P$ is some fixed path $e_1,\ldots,e_\ell\in E^\ext$ with labels $C_1,\ldots,C_\ell$ satisfying Item (a) (b) (c) above.
First by Item (c), there are at most $n$ possible $t_\mathsf{max}(e_1)$. Then by Item (1) of \Cref{fct:extended_constraints} there are at most $d$ possible $C_1$ given $t_\mathsf{max}(e_1)$. By Item (2) of \Cref{def:extended_constraints} this determines $(e_1,C_1)$.
Now given $(e_1,C_1)$, we enumerate the rest of the path by a rooted tree $\Tcal(e_1,C_1)$ constructed as follows:
\begin{itemize}
\item $\Tcal(e_1,C_1)$ has depth $2(\ell-1)$ and the root is labeled with $(e_1,C_1)$.
\item Given a node $z$ with label $(e_i,C_i),i\in[\ell-1]$, we construct the next two layers differently.
\begin{itemize}
\item For each $C_{i+1}\in\Ccal$ with $\vbl(C_i)\cap\vbl(C_{i+1})\neq\emptyset$, we create a child node $z'$ with label $C_{i+1}$ and link to $z$. 
\begin{flushright}($z'$ has at most $\Delta$ possibilities.)\end{flushright}
\item For each $z'$, assume its label is $C_{i+1}$. We find some $e_{i+1}$ such that $e_i\cap e_{i+1}\neq\emptyset$ and $(e_{i+1},C_{i+1})$ is an extended constraint, and then create a child node $z''$ with label $(e_{i+1},C_{i+1})$ and link to $z'$.
\begin{flushright}($z''$, given $z'$, has at most $2\cdot|\vbl(C_{i+1})|$ possibilities by Item (2c) of \Cref{fct:extended_constraints}.)\end{flushright}
\item We move to each $z''$ and repeat the construction.
\end{itemize}
\end{itemize}
Each leaf of $\Tcal(e_1,C_1)$ represents a path $P$ which satisfies Item (a) (c) already, and
\begin{itemize}
\item either, $P$ does not satisfy Item (b) and thus does not contribute to the union bound;
\item or, $P$ is valid and $\Pr\sbra{P\text{ appears in }H'}\le\prod_{i=1}^\ell\pbra{4\Delta\abs{\vbl(C_i)}}^{-1}$ as above.
\end{itemize}
Now we put weight on each internal node $z$ of $\Tcal(e_1,C_1)$ as follows:
\begin{itemize}
\item If $z$ has label $(e,C)$, then its weight is $w(z)=(\sqrt2\Delta)^{-1}$.
\item Otherwise $z$ has label $C$, then its weight is $w(z)=\pbra{2\sqrt2\cdot|\vbl(C)|}^{-1}$.
\end{itemize}
This means $\Pr\sbra{P\text{ appears in }H'}\le(4\Delta)^{-1}\prod_{\text{internal node }z\text{ in }P}w(z)$ for each valid $P$ in $\Tcal(e_1,C_1)$ where the $(4\Delta)^{-1}$ factor is because there are only $\ell-1$ internal nodes for each case along the path. 
Thus
\begin{align*}
\sum_{\text{valid }P\text{ in }\Tcal(e_1,C_1)}\Pr\sbra{P\text{ appears in }H'}
&\le\sum_{\text{valid }P\text{ in }\Tcal(e_1,C_1)}(4\Delta)^{-1}\prod_{\text{internal node }z\text{ in }P}w(z)\\
&\le(4\Delta)^{-1}\sum_{P\text{ in }\Tcal(e_1,C_1)}\prod_{\text{internal node }z\text{ in }P}w(z)\\
&\le(\sqrt2)^{-2(\ell-1)}/(4\Delta)=2^{-\ell}/(2\Delta),
\end{align*}
where the last inequality can be proved by induction on the depth and noticing $(\sqrt2\cdot w(z))^{-1}$ is at least the number of child nodes of $z$ for each internal node $z\in\Tcal(e_1,C_1)$.

Putting everything together, we have
\begin{align*}
\Pr\sbra{\Ecal}
&\le\sum_{\text{valid }(e_1,C_1)}\sum_{\text{valid }P\text{ in }\Tcal(e_1,C_1)}\Pr\sbra{P\text{ appears in }H'}\\
&\le nd\cdot 2^{-\ell}/(2\Delta)
\tag{$(e_1,C_1)$ has at most $nd$ choices}\\
&\le n\cdot 2^{2-\frac T{n-1}}\le 4n\cdot 2^{-\frac Tn}.
\tag{since $\ell\ge-3+\lceil T/(n-1)\rceil$ and $d\le\Delta$}
\end{align*}
\end{proof}

\subsection{The Distribution after \textsf{BoundingChain} Subroutines}\label{sec:the_distribution_after_boundingchain_subroutines}

Recall in \AtomicCSPSampling{$\Phi,\Mcal$} (\Cref{alg:atomiccspsampling}), we keep doubling $T$ and performing the corresponding \BoundingChain{$\Phi,\Mcal,-T,r_{-T},\ldots,r_{-1}$} until the returned assignment has no $\Qmark$ on $\Mcal$. 
Thus before we present the \textsf{FinalSampling} subroutine, we pause for now to analyze these \textsf{BoundingChain} calls \emph{in a whole}.

\begin{definition}[Projected Distribution]\label{def:projected_distribution}
Given an atomic CSP $\Phi=\pbra{V,\pbra{\Omega_v,\Dcal_v}_{v\in V},\Ccal}$ and a marking $\Mcal\subseteq V$, let $\Lambda=\pbra{\prod_{v\in\Mcal}\Omega_v}\times\cbra{\Qmark}^{V\setminus\Mcal}$ and define the \emph{projected distribution} $\mu^\Mcal\in\Rbb^\Lambda$ by
$$
\mu^\Mcal(\sigma)=\Pr_{\sigma'\sim\mu_\True^\Ccal}\sbra{\sigma'(\Mcal)=\sigma(\Mcal)}
\quad\text{for all $\sigma\in\Lambda$}.
$$
\end{definition}

We will show the distribution after all \textsf{BoundingChain} subroutines is exactly $\mu^\Mcal$.
\begin{proposition}\label{prop:boundingchain_distribution}
If $\Naturale\alpha\Delta\le1$, $\Naturale\Delta^2\rho\le1/32$, and $\Delta^2\lambda\le1/16$, then the \textsf{while} iterations in \AtomicCSPSampling{$\Phi,\Mcal$} halts almost surely and the final assignment $\tilde\sigma$ has distribution $\mu^\Mcal$.
\end{proposition}

We introduce and analyze the following \SystematicScan{$\Phi,\Mcal,\sigmain,L,R,r_L,\ldots,r_R$} algorithm, then show it couples with \textsf{BoundingChain}.

\begin{algorithm}[ht]
\caption{The \textsf{SystematicScan} algorithm}\label{alg:systematicscan}
\DontPrintSemicolon
\LinesNumbered
\KwIn{an atomic CSP $\Phi=\pbra{V,\pbra{\Omega_v,\Dcal_v}_{v\in V},\Ccal}$, a marking $\Mcal\subseteq V$, an assignment $\sigmain\in\pbra{\prod_{v\in\Mcal}\Omega_v}\times\cbra{\Qmark}^{V\setminus\Mcal}$, a starting time $L$, a stopping time $R$, and randomness tapes $r_L,\ldots,r_R$.}
\KwOut{an assignment $\sigma\in\pbra{\prod_{v\in\Mcal}\Omega_v}\times\cbra{\Qmark}^{V\setminus\Mcal}$.}
Initialize $\sigma\gets\sigmain$\;
\For(\tcc*[f]{Assume $V=\cbra{v_0,\ldots,v_{n-1}}$}){$t=L$ \KwTo $R$}{
$i_t\gets t\mod n$, and $\sigma(v_{i_t})\gets\Qmark$
\tcc*{Update $\sigma(v_{i_t})$ in this round}
$\Phi_t=\pbra{V_t,\pbra{\Omega_v|_\sigma,\Dcal_v|_\sigma}_{v\in V_t},\Ccal_t}\gets\Component{$\Phi,\Mcal,\sigma,v_{i_t}$}$\;
\tcc*{Ignore the returned $\Token$ since it is always $\True$ here}
\lIf{$v_{i_t}\in V\setminus\Mcal$}{Continue}
\Else{
$\sigma'\gets\RejectionSampling{$\Phi_t,r_t$}$\;
Update $\sigma\gets\sigma'(v_{i_t})$
}
}
\Return{$\sigma$}
\end{algorithm}

\subsubsection{Convergence of \textsf{SystematicScan}}\label{sec:convergence_of_systematicscan}

We first show \textsf{SystematicScan} converges to $\mu^\Mcal$. We set up basic notations for Markov chains.

Let $\Lambda$ be a finite state space; for our purpose it will be $\pbra{\prod_{v\in\Mcal}\Omega_v}\times\cbra{\Qmark}^{V\setminus\Mcal}$. We view any distribution $\mu$ over $\Lambda$ as a \emph{horizontal} vector in $\Rbb^\Lambda$ where $\sum_{a\in\Lambda}\mu(a)=1$ and $\mu(a)\ge0$ holds for all $a\in\Lambda$.
We denote $\indicator_a\in\Rbb^\Lambda$ as the point distribution of $a\in\Lambda$, i.e., $\indicator_a(b)=1$ iff $a=b$. 

A \emph{Markov chain} $\pbra{X_t}_{t\ge0}$ over $\Lambda$ is given by its transition matrices $\pbra{\Psf_t}_{t\ge0}$ where each $\Psf_t\in\Rbb^{\Lambda\times\Lambda}$ has non-negative entries and $\sum_{b\in\Lambda}\Psf_t(a,b)\equiv1$ for all $a\in\Lambda$. Then $X_t=X_0\Psf_0\Psf_1\cdots\Psf_{t-1}$ where $X_0\in\Rbb^\Lambda$ is the starting distribution. In particular, 
\begin{itemize}
\item if $\Psf_t\equiv\Psf$ for all $t\ge0$, then $\pbra{X_t}_{t\ge0}$ is a \emph{time homogeneous Markov chain} given by $\Psf$;
\item if $\pbra{\Psf_t}_{t\ge0}$ are possibly different, then $\pbra{X_t}_{t\ge0}$ is a \emph{time inhomogeneous Markov chain}.
\end{itemize}

Assume $\pbra{X_t}_{t\ge0}$ is a time homogeneous Markov chain over $\Lambda$ given by transition matrix $\Psf$. We say $\Psf$ is 
\begin{itemize}
\item \emph{irreducible} if for any $a,b\in\Lambda$, there exists some integer $t\ge0$ such that $\Psf^t(a,b)>0$;
\item \emph{aperiodic} if for any $a\in\Lambda$, $\mathsf{gcd}\cbra{\text{integer }t>0\mid\Psf^t(a,a)>0}=1$;\footnote{$\mathsf{gcd}$ stands for greatest common divisor.}
\item \emph{stationary with respect to distribution $\mu$} if $\mu\Psf=\mu$;
\item \emph{reversible with respect to distribution $\mu$} if $\mu(a)\Psf(a,b)=\mu(b)\Psf(b,a)$ holds for all $a,b\in\Lambda$.
\end{itemize}
Here we note the following two classical results.
\begin{fact}[e.g., {\cite[Proposition 1.20]{levin2017markov}}]\label{fct:reversible_to_stationary}
If $\Psf$ is reversible with respect to $\mu$, then $\Psf$ is also stationary with respect to $\mu$.
\end{fact}
\begin{theorem}[The Convergence Theorem, e.g., {\cite[Theorem 4.9]{levin2017markov}}]\label{thm:the_convergence_theorem}
Suppose $\pbra{X_t}_{t\ge0}$ is an irreducible and aperiodic time homogeneous Markov chain over finite state space $\Lambda$ with stationary distribution $\mu$ and transition matrix $\Psf$. Then for any $X_0$, we have $\lim_{t\to+\infty}X_t=\mu$.\footnote{The convergence is entry-wise.}
\end{theorem}

Now we turn to \textsf{SystematicScan} and follow the notation convention in \Cref{alg:systematicscan}: $V=\cbra{v_0,\ldots,v_{n-1}}$ and $i_t=t\mod n$. Recall we also fix $\Lambda=\pbra{\prod_{v\in\Mcal}\Omega_v}\times\cbra{\Qmark}^{V\setminus\Mcal}$.

\begin{definition}[One-step Transition Matrix]\label{def:one-step_transition_matrix}
For any $i\in\cbra{0,\ldots,n-1}$, define the \emph{one-step transition matrix} on $v_i\in V$ as $\Psf_i\in\Rbb^{\Lambda\times\Lambda}$ where
$$
\Psf_i(\sigma_1,\sigma_2)=
\Pr_{\sigma'\sim\mu_\True^\Ccal}\sbra{\sigma'(\Mcal)=\sigma_2(\Mcal)\mid \sigma'(\Mcal\setminus\cbra{v_i})=\sigma_1(\Mcal\setminus\cbra{v_i})}.
$$
\end{definition}

We first prove some useful facts and also connect one-step transition matrices to \textsf{SystematicScan}. 
\begin{proposition}\label{prop:systematicscan}
If $\Naturale\alpha\Delta\le1$, then the following holds.
\begin{itemize}
\item[(1)] Each $\Psf_i$ is well-defined. 
\item[(2)] For any $i\in\cbra{0,\ldots,n-1}$ and $\sigma_1,\sigma_2\in\Lambda$, we have $\Psf_i(\sigma_1,\sigma_2)\equiv0$ if $\sigma_1(\Mcal\setminus\cbra{v_i})\neq\sigma_2(\Mcal\setminus\cbra{v_i})$; and $\Psf_i(\sigma_1,\sigma_2)>0$ if otherwise.
\item[(3)] $\mu^\Mcal\Psf_{i_1}\cdots\Psf_{i_m}=\mu^\Mcal$ holds for any sequence $i_1,i_2,\ldots,i_m\in\cbra{0,\ldots,n-1}$ of finite length $m$.
\item[(4)] For any $L\le R$ and $\sigmain$, $\SystematicScan{$\Phi,\Mcal,\sigmain,L,R,r_L,\ldots,r_R$}$ halts almost surely over random $r_L,\ldots,r_R$ and its output distribution is $\mu=\indicator_\sigmain\Psf_{i_L}\Psf_{i_{L+1}}\cdots\Psf_{i_R}$.
\end{itemize}
\end{proposition}
\begin{proof}
First we prove Item (1) (2).
For any fixed $i\in\cbra{0,\ldots,n-1}$ and $\sigma_1\in\Lambda$, define assignment $\sigma$ by setting $\sigma(V\setminus\cbra{v_i})=\sigma_1(V\setminus\cbra{v_i})$ and $\sigma(v_i)=\Qmark$.
Then for any $\sigma_2\in\Lambda$ we have
$$
\Psf_i(\sigma_1,\sigma_2)
=\Pr_{\sigma'\sim\mu_\True^{\Ccal|_\sigma}}\sbra{\sigma'(\Mcal)=\sigma_2(\Mcal)},
$$
where $\Ccal|_\sigma$ is defined in \Cref{def:projected_constraint_satisfaction_problem}.
Thus Item (1) is equivalent to $\mu_\True^{\Ccal|_\sigma}$ being well-defined, which follows from \Cref{prop:rejectionsampling}.
Since $\sigma(\Mcal)$ has no $\Qmark$, $\sigma'(\Mcal\setminus\cbra{v_i})\equiv\sigma(\Mcal\setminus\cbra{v_i})\equiv\sigma_1(\Mcal\setminus\cbra{v_i})$. Thus the first part of Item (2) follows naturally. 
As for the second part, if $\sigma_1(V\setminus\cbra{v_i})=\sigma_2(V\setminus\cbra{v_i})$, then 
\begin{equation}\label{eq:simplified_transition}
\Psf_i(\sigma_1,\sigma_2)
=\frac{\Pr_{\sigma'\sim\mu_\True^\Ccal}\sbra{\sigma'(\Mcal)=\sigma_2(\Mcal)}}{\Pr_{\sigma'\sim\mu_\True^\Ccal}\sbra{\sigma'(\Mcal\setminus\cbra{v_i})=\sigma_1(\Mcal\setminus\cbra{v_i})}}.
\end{equation}
Thus $\Psf_i(\sigma_1,\sigma_2)>0$ iff the enumerator is positive, which is equivalent to $\mu_\True^{\Ccal|_{\sigma_2}}$ being well-defined and is, again, guaranteed by \Cref{prop:rejectionsampling}.

Now we prove Item (3). Since $\mu^\Mcal\Psf_{i_1}\cdots\Psf_{i_m}=\pbra{\mu^\Mcal\Psf_{i_1}}\Psf_{i_2}\cdots\Psf_{i_m}$, it suffices to show for $m=1$ and then apply induction. By \Cref{fct:reversible_to_stationary}, it suffices to show for each $i\in\cbra{0,\ldots,n-1}$, $\Psf_i$ is reversible with respect to $\mu^{\Phi|_\pibm}$, i.e.,
\begin{equation}\label{eq:systematicscan_stationary}
\mu^\Mcal(\sigma_1)\Psf_i(\sigma_1,\sigma_2)=\mu^\Mcal(\sigma_2)\Psf_i(\sigma_2,\sigma_1)
\quad\text{for any $\sigma_1,\sigma_2\in\Lambda$}.
\end{equation}
By Item (2), we may safely assume $\sigma_1(\Mcal\setminus\cbra{v_i})=\sigma_2(\Mcal\setminus\cbra{v_i})$ and observe that
\begin{align*}
&\phantom{=}\mu^\Mcal(\sigma_1)\Psf_i(\sigma_1,\sigma_2)\\
&=\Pr_{\sigma'\sim\mu_\True^\Ccal}\sbra{\sigma'(\Mcal)=\sigma_1(\Mcal)}\cdot\frac{\Pr_{\sigma'\sim\mu_\True^\Ccal}\sbra{\sigma'(\Mcal)=\sigma_2(\Mcal)}}{\Pr_{\sigma'\sim\mu_\True^\Ccal}\sbra{\sigma'(\Mcal\setminus\cbra{v_i})=\sigma_1(\Mcal\setminus\cbra{v_i})}}
\tag{by \Cref{eq:simplified_transition}}\\
&=\Pr_{\sigma'\sim\mu_\True^\Ccal}\sbra{\sigma'(\Mcal)=\sigma_2(\Mcal)}\cdot\frac{\Pr_{\sigma'\sim\mu_\True^\Ccal}\sbra{\sigma'(\Mcal)=\sigma_1(\Mcal)}}{\Pr_{\sigma'\sim\mu_\True^\Ccal}\sbra{\sigma'(\Mcal\setminus\cbra{v_i})=\sigma_2(\Mcal\setminus\cbra{v_i})}}\\
&=\mu^\Mcal(\sigma_2)\Psf_i(\sigma_2,\sigma_1).
\end{align*}

Finally we turn to Item (4). By induction on $R$, it suffices to verify for $L=R=t\in\Zbb$ that $\indicator_\sigmain\Psf_{i_t}=\Psf_{i_t}(\sigmain,\cdot)$ is the distribution of $\sigmaout\gets\SystematicScan{$\Phi,\Mcal,\sigmain,t,t,r_t$}$. 
Similarly as above, define assignment $\tilde\sigma$ by setting $\tilde\sigma(V\setminus\cbra{v_{i_t}})=\sigmain(V\setminus\cbra{v_{i_t}})$ and $\tilde\sigma(v_{i_t})=\Qmark$. 
Then for each $\hat\sigma\in\Lambda$ we have
$$
\Psf_{i_t}(\sigmain,\hat\sigma)=\Pr_{\sigma'\sim\mu_\True^{\Ccal|_{\tilde\sigma}}}\sbra{\sigma'(\Mcal)=\hat\sigma(\Mcal)}.
$$
Now we consider two separate cases.
\begin{itemize}
\item \fbox{$v_{i_t}\in V\setminus\Mcal$.} 
Then $\sigma'(\Mcal)\equiv\tilde\sigma(\Mcal)$. Thus
$\Psf_{i_t}(\sigmain,\hat\sigma)$ equals $1$ if $\hat\sigma=\sigmain$; and equals $0$ if otherwise. This agrees with $\sigmaout\equiv\sigmain$ from the algorithm.
\item \fbox{$v_i\in\Mcal$.}
Let $\Phi_t=\pbra{V_t,\pbra{\Omega_v|_{\tilde\sigma},\Dcal_v|_{\tilde\sigma}}_{v\in V_t},\Ccal_t}$ be from \textsf{Line 4}. Then by Item (2) of \Cref{prop:rejectionsampling}, we have
$$
\Pr\sbra{\sigmaout=\hat\sigma}=\begin{cases}
\Pr_{\sigma'\sim\mu_\True^{\Ccal_t}}\sbra{\sigma'(v_{i_t})=\hat\sigma(v_{i_t})} & \hat\sigma(V\setminus\cbra{v_{i_t}})=\tilde\sigma(V\setminus\cbra{v_{i_t}}),\\
0 & \text{otherwise}.
\end{cases}
$$
By Item (4) of \Cref{prop:component}, $\Pr_{\sigma'\sim\mu_\True^{\Ccal_t}}\sbra{\sigma'(v_{i_t})=\hat\sigma(v_{i_t})}=\Pr_{\sigma'\sim\mu_\True^{\Ccal|_{\tilde\sigma}}}\sbra{\sigma'(v_{i_t})=\hat\sigma(v_{i_t})}$.
Thus
\begin{align*}
\Pr\sbra{\sigmaout=\hat\sigma}
&=\begin{cases}
\Pr_{\sigma'\sim\mu_\True^{\Ccal|_{\tilde\sigma}}}\sbra{\sigma'(v_{i_t})=\hat\sigma(v_{i_t})} & \hat\sigma(V\setminus\cbra{v_{i_t}})=\tilde\sigma(V\setminus\cbra{v_{i_t}}),\\
0 & \text{otherwise}
\end{cases}\\
&=\Pr_{\sigma'\sim\mu_\True^{\Ccal|_{\tilde\sigma}}}\sbra{\sigma'(\Mcal)=\hat\sigma(\Mcal)}
\tag{since $\hat\sigma\in\Lambda$ and $\sigma'(\Mcal\setminus\cbra{v_{i_t}})\equiv\tilde\sigma(\Mcal\setminus\cbra{v_{i_t}})$}\\
&=\Psf_{i_t}(\sigmain,\hat\sigma).
\tag*{\qedhere}
\end{align*}
\end{itemize}
\end{proof}

Item (4) of \Cref{prop:systematicscan} shows \textsf{SystematicScan} is a \emph{time inhomogeneous Markov chain}. General theory regarding time inhomogeneous Markov chains can be much more complicated but luckily we can embed this one into a time homogeneous Markov chain.

\begin{lemma}\label{lem:systematicscan_final_distribution}
Assume $\Naturale\alpha\Delta\le1$.
Let $L$ and $\sigmain\in\Lambda$ be arbitrary. 
Define $\mu_R=\indicator_\sigmain\Psf_{i_L}\cdots\Psf_{i_R}$. Then $\lim_{R\to+\infty}\mu_R=\mu^\Mcal$.
\end{lemma}
\begin{proof}
Let $\Fsf=\Psf_{i_L}\cdots\Psf_{i_{L+n-1}}$. Since $i_t=t\mod n$, the one-step transition matrices repeatedly applied to $\indicator_\sigmain$ has period $n$. Hence $\mu_R=\indicator_\sigmain\Fsf^m$ if $R=L+m\cdot n-1$ and $m\ge1$.
Let $Y_0=\indicator_\sigmain$ and $Y_t=Y_0\Fsf^i$ for $t\ge1$, then $\pbra{Y_t}_{t\ge0}$ is a time homogeneous Markov chain with transition matrix $\Fsf$. Here we verify the following properties of $\Fsf$.
\begin{itemize}
\item \textsc{Stationary with respect to $\mu^\Mcal$.} This follows immediately from Item (3) of \Cref{prop:systematicscan}.
\item \textsc{Aperiodic.} By Item (2) of \Cref{prop:systematicscan}, for any $i\in\cbra{0,\ldots,n-1}$ and any $\sigma\in\Lambda$ we have $\Psf_i(\sigma,\sigma)>0$. Thus $\Fsf(\sigma,\sigma)>0$ which implies $\Fsf$ is aperiodic.
\item \textsc{Irreducible.} Let $\sigma_1,\sigma_2\in\Lambda$ be arbitrary. For each $j\in\cbra{0,\ldots,n}$, define $\sigma^j\in\Lambda$ by
$$
\sigma^j(v_{i'})=\begin{cases}
\sigma_2(v_{i'}) & i'\in\cbra{i_L,\ldots,i_{L+j-1}},\\
\sigma_1(v_{i'}) & \text{otherwise}.
\end{cases}
$$
Then $\sigma_1=\sigma^0$ and $\sigma_2=\sigma^n$. Moreover $\Psf_{i_{L+j}}(\sigma^j,\sigma^{j+1})>0$ for all $j\in\cbra{0,\ldots,n-1}$ by Item (2) of \Cref{prop:systematicscan}. Thus $\Fsf(\sigma_1,\sigma_2)\ge\Psf_{i_L}(\sigma^0,\sigma^1)\Psf_{i_{L+1}}(\sigma^1,\sigma^2)\cdots\Psf_{i_{L+n-1}}(\sigma^{n-1},\sigma^n)>0$.
\end{itemize}
Therefore by \Cref{thm:the_convergence_theorem}, $\lim_{m\to+\infty}\mu_{L+m\cdot n-1}=\lim_{t\to+\infty}Y_t=\mu^\Mcal$.
Since each $\Psf_i$ is stationary with respect to $\mu^\Mcal$ by Item (3) of \Cref{prop:systematicscan}, for any finite integer $o\ge0$
$$
\lim_{m\to+\infty}\mu_{L+m\cdot n-1+o}
=\pbra{\lim_{m\to+\infty}\mu_{L+m\cdot n-1}}\Psf_{i_L}\cdots\Psf_{i_{L+o-1}}
=\mu^\Mcal\Psf_{i_L}\cdots\Psf_{i_{L+o-1}}
=\mu^\Mcal.
$$
Hence $\lim_{R\to+\infty}\mu_R=\mu^\Mcal$.
\end{proof}

\subsubsection{Coupling from the Past and the Bounding Chain}\label{sec:coupling_from_the_past_and_the_bounding_chain}
We have showed \textsf{SystematicScan} \emph{converges} to distribution $\mu^\Mcal$, but to obtain a sample distributed \emph{exactly} according to $\mu^\Mcal$ we need to run for infinite time.
The trick for making it finite is to think \emph{backwards}. 
That is the idea of \emph{coupling from the past} \cite{PW96}; then the \emph{bounding chain} \cite{huber1998exact,haggstrom1999exact} is used to make the process more computationally efficient.

Let $\Psf\in\Rbb^{\Lambda\times\Lambda}$ be some transition matrix. We say $f\colon\Lambda\times[0,1]\to\Lambda$ is a \emph{coupling} of $\Psf$ if for all $a,b\in\Lambda$, $\Pr_{r\sim[0,1]}\sbra{f(a,r)=b}=\Psf(a,b)$.
We use random function $f^r\colon\Lambda\to\Lambda$ to denote the coupling $f$ with randomness $r$, i.e., $f^r(a)=f(a,r)$ for all $a\in\Lambda$.

Recall our definition of $\Psf_i$ from \Cref{def:one-step_transition_matrix} and $i_t=t\mod n$ from \Cref{alg:systematicscan}.
\begin{lemma}[Coupling from the Past]\label{lem:coupling_from_the_past}
Let $f_t\colon\Lambda\times[0,1]\to\Lambda$ be a coupling of $\Psf_{i_t}$ for all $t\in\Zbb$. Define random functions $F_{L,R}\colon\Lambda\to\Lambda$ over random $\pbra{r_t}_{t\in\Zbb}$ for $-\infty<L\le R<+\infty$ as
$$
F_{L,R}(a)=f_R^{r_R}\pbra{f_{R-1}^{r_{R-1}}\pbra{\cdots f_L^{r_L}(a)\cdots}}
\quad\text{for all $a\in\Lambda$}.
$$
Let $M\ge1$ be the smallest integer such that $F_{-M,-1}$ is a constant function. Let $A=F_{-M,-1}(\Lambda)$ be the corresponding constant.
Then $F_{-M',-1}(\Lambda)\equiv A$ for any $M'>M$, and $A$ is distributed as $\mu^\Mcal$ if $M<+\infty$ almost surely.
\end{lemma}
\begin{proof}
Since $\Psf_i=\Psf_{i+n}$, for any integer $\ell\ge1$ and all $a,b\in\Lambda$ we have 
$$
\Pr\sbra{F_{-\ell\cdot n,-1}(a,b)}=\indicator_a\pbra{\Psf_{-n}\Psf_{-n+1}\cdots\Psf_{-1}}^\ell\indicator_b^\top=\indicator_a\pbra{\Psf_0\Psf_1\cdots\Psf_{n-1}}^\ell\indicator_b^\top=\Pr\sbra{F_{0,\ell\cdot n-1}(a,b)}.
$$
Thus for any $a,b\in\Lambda$, by \Cref{lem:systematicscan_final_distribution} we have
$$
\lim_{\ell\to+\infty}\Pr\sbra{F_{-\ell\cdot n,-1}(a)=b}=\lim_{\ell\to+\infty}\Pr\sbra{F_{0,\ell\cdot n-1}(a)=b}=\mu^\Mcal(b).
$$
On the other hand for $M'>M$, $F_{-M',-1}(\Lambda)=F_{-M,-1}\pbra{F_{-M',-M-1}(\Lambda)}=A$. If $M<+\infty$ almost surely, then we have
\begin{equation*}
\Pr\sbra{A=b}=\lim_{\ell\to+\infty}\Pr\sbra{F_{-\ell\cdot n,-1}(a)=b}=\mu^\Mcal(b).
\tag*{\qedhere}
\end{equation*}
\end{proof}

Therefore to obtain a perfect sample from $\mu^\Mcal$ we only need to (1) design a coupling, then (2) sample randomness tapes $\pbra{r_t}_{t\ge0}$, and lastly (3) find some $M'\ge1$ such that $F_{-M',-1}$ is a constant function and output the corresponding constant. Now we show in the following \Cref{alg:coupling} that \Cref{alg:boundingchain} implicitly provides a coupling for (1) and an efficient way to check (3).

Recall $\Lambda=\pbra{\prod_{v\in\Mcal}\Omega_v}\times\cbra{\Qmark}^{V\setminus\Mcal}$. We define $\Lambda^\Qmark=\pbra{\prod_{v\in\Mcal}\Omega_v^\Qmark}\times\cbra{\Qmark}^{V\setminus\Mcal}$.
Then we construct $g_i\colon\Lambda^\Qmark\times[0,1]\to\Lambda^\Qmark$ for each $i\in\cbra{0,\ldots,n-1}$ as follows. We also remark that this coupling is information-theoretic and is for analysis only.

\begin{algorithm}[ht]
\caption{A coupling $g_i\colon\Lambda^\Qmark\times[0,1]\to\Lambda^\Qmark$ for each $i\in\cbra{0,\ldots,n-1}$}\label{alg:coupling}
\DontPrintSemicolon
\LinesNumbered
\KwIn{an assignment $\sigmain\in\Lambda^\Qmark$ and a randomness tape $r\in[0,1]$.}
\KwOut{an assignment $\sigmaout\in\Lambda^\Qmark$.}
\KwData{an atomic CSP $\Phi=\pbra{V,\pbra{\Omega_v,\Dcal_v}_{v\in V},\Ccal}$, a marking $\Mcal\subseteq V$, and an index $i\in\cbra{0,\ldots,n-1}$. Assume $V=\cbra{v_0,\ldots,v_{n-1}}$.}
Initialize $\sigmaout\gets\sigmain$ and divide $r$ into two independent parts $r_1,r_2$\;
\lIf{$v_i\in V\setminus\Mcal$}{\Return{$\sigmaout$}}
Recall $\beta=\beta(\Phi,\Mcal)$ from \Cref{eq:beta} and define distribution $\Dcal_{v_i}^\Qmark$ over $\Omega_{v_i}^\Qmark$ by setting
$$
\Dcal_{v_i}^\Qmark(q)=\begin{cases}
\max\cbra{0,1-\beta\cdot(1-\Dcal_{v_i}(q))} & q\in\Omega_{v_i},\\
1-\sum_{q'\in\Omega_{v_i}}\Dcal_{v_i}^\Qmark(q') & q=\Qmark.
\end{cases}
$$\;
\vspace{-20pt}
Sample $c_1\sim\Dcal_{v_i}^\Qmark$ using $r_1$ and update $\sigmaout(v_i)\gets c_1$\;
\lIf{$c_1\neq\Qmark$}{\Return{$\sigmaout$}}
$\pbra{\Phi',\Token}\gets\Component{$\Phi,\Mcal,\sigmaout,v_i$}$ where $\Phi'=\pbra{V',\pbra{\Omega_v|_{\sigmaout},\Dcal_v|_{\sigmaout}}_{v\in V'},\Ccal'}$\;
\lIf{$\Token=\False$}{\Return{$\sigmaout$}}
Define distribution $\Dcal_{v_i}^\dag$ over $\Omega_{v_i}$ by setting $\Dcal_{v_i}^\dag(q)=\Pr_{\sigma'\sim\mu_\True^{\Ccal'}}\sbra{\sigma'(v_i)=q}
\quad\text{for $q\in\Omega_{v_i}$}$.\;
Define distribution $\Dcal_{v_i}'$ over $\Omega_{v_i}$ by setting $\Dcal_{v_i}'(q)=\frac{\Dcal_{v_i}^\dag(q)-\Dcal_{v_i}^\Qmark(q)}{\Dcal_{v_i}^\Qmark(\Qmark)}$ for $q\in\Omega_{v_i}$\;
Sample $c_2\sim\Dcal_{v_i}'$ using $r_2$ and update $\sigmaout(v_i)\gets c_2$\;
\Return{$\sigmaout$}
\end{algorithm}

We say $\sigma_1\in\sigma_2$ for some $\sigma_1,\sigma_2\in\Lambda^\Qmark$ if $\sigma_2(v)\in\cbra{\sigma_1(v),\Qmark}$ for all $v\in V$. 
\begin{fact}\label{fct:bounding}
When $\sigma_1,\sigma_2\in\Lambda$, $\sigma_1\in\sigma_2$ iff $\sigma_1=\sigma_2$.
\end{fact}

Now we verify the following properties of \Cref{alg:coupling}. 
\begin{proposition}\label{prop:coupling}
If $\Naturale\alpha\Delta\le1$ then the following holds for $g_i\colon\Lambda^\Qmark\times[0,1]\to\Lambda^\Qmark$ from \Cref{alg:coupling}.
\begin{itemize}
\item[(1)] All the distributions in \Cref{alg:coupling} are well-defined.
\item[(2)] $g_i(\sigma_1,r)\in g_i(\sigma_2,r)$ holds for any $r\in[0,1]$ and $\sigma_1,\sigma_2\in\Lambda^\Qmark$ with $\sigma_1\in\sigma_2$.
\item[(3)] $g_i$ restricted on $\Lambda$ is a coupling of $\Psf_i$.
\item[(4)] For any $t\equiv i\mod n$, $g_i$ is the same update procedure as the $t$-th \textsf{for} iteration in \Cref{alg:boundingchain}.
\end{itemize}
\end{proposition}
\begin{proof}
First we prove Item (1).
Note that $\Dcal_{v_i}^\Qmark$ is the same one as in \SafeSampling{$\Phi,\Mcal,v_i,\cdot$}. Thus it is well-defined by Item (1) of \Cref{prop:safesampling}.
Now we assume $\Token=\True$ to check $\Dcal_{v_i}^\dag$ and $\Dcal_{v_i}'$. By Item (2) of \Cref{prop:rejectionsampling}, $\Dcal_{v_i}^\dag$ is well-defined since $\mu_\True^{\Ccal'}$ is well-defined. By Item (3) of \Cref{prop:safesampling}, $\Dcal_{v_i}^\Qmark(q)\le\Dcal_{v_i}^\dag(q)$ for any $q\in\Omega_{v_i}$. Hence $\Dcal_{v_i}'$ is also well-defined.

For Item (2), we have the following cases, each of which satisfies $g_i(\sigma_1,r)\in g_i(\sigma_2,r)$.
\begin{itemize}
\item \fbox{$c_1\neq\Qmark$.} Then both $\sigma_1(v_i)$ and $\sigma_2(v_i)$ are updated to $c_1$.
\item \fbox{$c_1=\Qmark$ and $\Token=\False$ for $\sigma_2$.} Then $\sigma_2(v_i)$ is updated to $\Qmark$.
\item \fbox{$c_1=\Qmark$ and $\Token=\True$ for $\sigma_2$.} Since $\sigma_1\in\sigma_2$, $\Token$ also equals $\True$ for $\sigma_1$. Moreover they get the same CSP from \textsf{Line 6}. Thus $\sigma_1(v_i)$ and $\sigma_2(v_i)$ are updated to the same value $c_2$.
\end{itemize}

Item (3) is obviously true if $v_i\in V\setminus\Mcal$ thus we assume $v_i\in\Mcal$.
Observe that $\Dcal_{v_i}^\dag(q)=\Dcal_{v_i}'(q)\cdot\Dcal_{v_i}^\Qmark(\Qmark)+\Dcal_{v_i}^\Qmark(q)$ for any $q\in\Omega_{v_i}$. Then Item (3) follows from Item (4) of \Cref{prop:systematicscan} with $L=R=i$.

Finally we prove Item (4). Assume $v_i\in\Mcal$ since otherwise it is trivial.
To obtain the pseudocode in \Cref{alg:boundingchain}, we reorganize \Cref{alg:coupling} by first set $\sigmaout\gets\Qmark$ and call \Component{$\Phi,\Mcal,\sigmaout,v_i$}, then based on the value of $\Token$ we either (A) sample $c_1$ only then update $\sigmaout(v_i)\gets c_1$, or (B) sample both $c_1$ and $c_2$ then update $\sigmaout(v_i)\gets c_1$ if $c_1\neq\Qmark$; and $\sigmaout(v_i)\gets c_2$ if otherwise. The former is \textsf{SafeSampling}, and the latter, executed jointly, is exactly \textsf{RejectionSampling} as we analyzed for Item (3).
\end{proof}

One more ingredient we need is the following well-known Borel-Cantelli theorem.
\begin{theorem}[Borel-Cantelli Theorem, e.g., {\cite[Section 7.3]{GrimmettS20}}]\label{thm:borel-cantelli_theorem}
Let $T$ be a non-negative random variable. If $\sum_{i=0}^{+\infty}\Pr\sbra{T>i}<+\infty$ then $T<+\infty$ almost surely.
\end{theorem} 

Now we are ready to prove \Cref{prop:boundingchain_distribution}.
\begin{proof}[Proof of \Cref{prop:boundingchain_distribution}]
Assume the \textsf{while} iterations in \Cref{alg:atomiccspsampling} stop at $T=T_\Final$ or $T_\Final=+\infty$ if iterations never end. Since each iteration halts almost surely by \Cref{prop:boundingchain}, $T_\Final$ is well-defined almost surely.

Let $V=\pbra{v_0,\ldots,v_{n-1}}$ and define $i_t=t\mod n$ for all $t\in\Zbb$ as in \Cref{alg:boundingchain}.
Define random functions $G_{L,R}\colon\Lambda^\Qmark\to\Lambda^\Qmark$ over random $\pbra{r_t}_{t\in\Zbb}$ for $-\infty<L\le R<+\infty$ as 
$$
G_{L,R}(a)=g_R^{r_R}\pbra{g_{R-1}^{r_{R-1}}\pbra{\cdots g_L^{r_L}(a)\cdots}}
\quad\text{for all $a\in\Lambda^\Qmark$},
$$
where $g_t=g_{i_t}\colon\Lambda^\Qmark\times[0,1]\to\Lambda^\Qmark$ is from \Cref{alg:coupling} and $g_t^r(\cdot)=g_t(\cdot,r)$.

Let $M_1\ge1$ be the smallest integer such that $G_{-M_1,-1}$ is a constant function on $\Lambda$, i.e., $G_{-M_1,-1}(\Lambda)\equiv A$.
By Item (3) of \Cref{prop:coupling} and \Cref{lem:coupling_from_the_past}, $G_{-M_1',-1}(\Lambda)\equiv A$ for all $M_1'\ge M_1$, and $A$ is distributed as $\mu^\Mcal$ if $M_1<+\infty$ almost surely.

Let $M_2\ge1$ be the smallest integer such that $G_{-M_2,-1}(\Qmark^V)\in\Lambda$. Iteratively applying Item (2) of \Cref{prop:coupling}, we know $G_{-M_2,-1}(\Lambda)\in G_{-M_2,-1}(\Qmark^V)$. Then by \Cref{fct:bounding}, $G_{-M_2,-1}(\Lambda)$ is constant. Thus $M_2\ge M_1$ and $G_{-M_2,-1}(\Qmark^V)=A$.

By Item (4) of \Cref{prop:coupling}, \BoundingChain{$\Phi,\Mcal,-T,r_{-T},\ldots,r_{-1}$} equals $G_{-T,-1}(\Qmark^V)$, which means $T_\Final\ge M_2$. Thus the final assignment $\tilde\sigma$ equals $A$ and has distribution $\mu^\Mcal$ provided $T_\Final<+\infty$ almost surely.

Now we only need to show $T_\Final<+\infty$ almost surely. Note that either $T_\Final=1$ or, by \Cref{alg:atomiccspsampling}, $\BoundingChain{$\Phi,\Mcal,-T_\Final/2,r_{-T_\Final/2},\ldots,r_{-1}$}=G_{-T_\Final/2,-1}(\Qmark^V)\notin\Lambda$. Thus $T_\Final\le 2\cdot M_2$, which means it suffices to show $M_2<+\infty$ almost surely. 
By Item (3) of \Cref{prop:coupling} and the analysis above, $G_{-i,-1}(\Qmark^V)=A\in\Lambda$ for all $i\ge M_2$; thus
\begin{align*}
\sum_{i=0}^{+\infty}\Pr\sbra{M_2>i}
&\le 2n-1+\sum_{i=2n-1}^{+\infty}\Pr\sbra{G_{-i,-1}(\Qmark^V)\notin\Lambda}\\
&=2n-1+\sum_{i=2n-1}^{+\infty}\Pr\sbra{\BoundingChain{$\Phi,\Mcal,-i,r_{-i},\ldots,r_{-1}$}\notin\Lambda}\\
&\le2n-1+\sum_{i=2n-1}^{+\infty}4n\cdot2^{-\frac in}<+\infty,
\tag{by \Cref{prop:boundingchain}}
\end{align*}
as desired by \Cref{thm:borel-cantelli_theorem}.
\end{proof}

\subsection{The \textsf{FinalSampling} Subroutine}\label{sec:the_finalsampling_subroutine}

Finally we give the missing \FinalSampling{$\Phi,\Mcal,\tilde\sigma$} subroutine, which simply completes the assignment on $V\setminus\Mcal$ for $\tilde\sigma$.

\begin{algorithm}[ht]
\caption{The \textsf{FinalSampling} subroutine}\label{alg:finalsampling}
\DontPrintSemicolon
\LinesNumbered
\KwIn{an atomic CSP $\Phi=\pbra{V,\pbra{\Omega_v,\Dcal_v}_{v\in V},\Ccal}$, a marking $\Mcal\subseteq V$, and an assignment $\tilde\sigma\in\pbra{\prod_{v\in\Mcal}\Omega_v}\times\cbra{\Qmark}^{V\setminus\Mcal}$}
\KwOut{an assignment $\sigma\in\sigma_\True^\Ccal$}
$\Phi_v=\pbra{V_v,\pbra{\Omega_v|_{\tilde\sigma},\Dcal_v|_{\tilde\sigma}}_{v\in V_v},\Ccal_v}\gets\Component{$\Phi,\Mcal,\tilde\sigma,v$}$ for all $v\in V$\;
\tcc*{Ignore the returned $\Token$ since it is always $\True$ here}
Initialize $V'\gets\emptyset$\;
\While{$\exists v\in V\setminus V'$}{
$\sigma_v\gets\RejectionSampling{$\Phi_v,r_v$}$
\tcc*{$r_v$ is a fresh new randomness tape}
Assign $\sigma(V_v)\gets\sigma_v(V_v)$ and update $V'\gets V'\cup V_v$
}
\Return{$\sigma$}
\end{algorithm}

We observe the following results regarding \Cref{alg:finalsampling}.
\begin{lemma}\label{lem:finalsampling}
If $\Naturale\alpha\Delta\le1$, then the following holds for \FinalSampling{$\Phi,\Mcal,\tilde\sigma$}.
\begin{itemize}
\item It halts almost surely, and outputs $\sigma\sim\mu_\True^{\Ccal|_{\tilde\sigma}}$ when it halts.
\item Its expected total running time is at most $O\pbra{kdQ\sum_{v\in V}\pbra{1-\Naturale\alpha}^{-\abs{\Ccal_v}}}$.
\end{itemize}
\end{lemma}
\begin{proof}
All the $\Phi_v$ can be easily constructed with one pass of $\Ccal$ and $V$ which takes time $O\pbra{k|\Ccal|+|V|}$.

By \Cref{prop:rejectionsampling}, each iteration of \textsf{Line 4} halts almost surely and generates $\sigma_v\sim\mu_\True^{\Ccal_v}$ in expected time $O\pbra{\pbra{kQ\abs{\Ccal_v}+Q}\cdot\pbra{1-\Naturale\alpha}^{-\abs{\Ccal_v}}}$. 
Therefore \textsf{FinalSampling} halts almost surely and its expected total running time is at most
\begin{align*}
&\phantom{\le}O\pbra{k|\Ccal|+|V|+\sum_{v\text{ needed by \textsf{Line 4}}}\pbra{kQ\abs{\Ccal_v}+Q}\cdot\pbra{1-\Naturale\alpha}^{-\abs{\Ccal_v}}}\\
&=O\pbra{k|\Ccal|+|V|+\sum_{v\in V}\frac{kQ\abs{\Ccal_v}+Q}{|V_v|}\cdot\pbra{1-\Naturale\alpha}^{-\abs{\Ccal_v}}}\\
&\le O\pbra{kdQ\sum_{v\in V}\pbra{1-\Naturale\alpha}^{-\abs{\Ccal_v}}}.
\tag{since $\abs{\Ccal_v}\le d\abs{V_v}$ and $|\Ccal|\le d|V|$}
\end{align*}
Moreover when \textsf{FinalSampling} halts, by iteratively applying Item (4) of \Cref{prop:component} we have
\begin{equation*}
\sigma\sim\prod_{v\text{ needed by \textsf{Line 4}}}\mu_\True^{\Ccal_v}=\mu_\True^{\Ccal|_{\tilde\sigma}}.
\tag*{\qedhere}
\end{equation*}
\end{proof}

Combining \Cref{prop:boundingchain_distribution}, we analyze the performance of \Cref{alg:finalsampling} in \Cref{alg:atomiccspsampling}.
\begin{proposition}\label{prop:finalsampling}
If $\Naturale\alpha\Delta\le1$, $\Naturale\Delta^2\rho\le1/32$, and $\Delta^2\lambda\le1/16$, then the following holds for the \textsf{FinalSampling} in \Cref{alg:atomiccspsampling}.
\begin{itemize}
\item It halts almost surely, and outputs $\sigma\sim\mu_\True^\Ccal$ when it halts.  
\item Its expected running time is at most $O\pbra{d^2kQ\Delta|V|}$.
\end{itemize}
\end{proposition}
\begin{proof}
By Item (1) of \Cref{lem:finalsampling}, it halts almost surely. 
Fix an arbitrary $\sigma'\in\sigma_\True^\Ccal$. Define $\tilde\sigma$ by setting $\tilde\sigma(\Mcal)=\sigma'(\Mcal)$ and $\tilde\sigma(V\setminus\Mcal)=\Qmark^{V\setminus\Mcal}$. 
Combining \Cref{prop:boundingchain_distribution} and \Cref{def:projected_distribution}, we have
\begin{align*}
\Pr\sbra{\sigma=\sigma'}
&=\mu^\Mcal(\tilde\sigma)\cdot\mu_\True^{\Ccal|_{\tilde\sigma}}(\sigma')
=\Pr_{\hat\sigma\sim\mu_\True^\Ccal}\sbra{\hat\sigma(\Mcal)=\sigma'(\Mcal)}\cdot\mu_\True^{\Ccal|_{\tilde\sigma}}(\sigma')\\
&=\Pr_{\hat\sigma\sim\prod_{v\in V}\Dcal_v}\sbra{\hat\sigma(\Mcal)=\sigma'(\Mcal)\mid\hat\sigma\in\sigma_\True^\Ccal}\cdot\Pr_{\sigma''\sim\prod_{v\in V}\Dcal_v|_{\tilde\sigma}}\sbra{\sigma''=\sigma'\mid\sigma''\in\sigma_\True^{\Ccal|_{\tilde\sigma}}}\\
&=\frac{\Pr_{\hat\sigma\sim\prod_{v\in V}\Dcal_v}\sbra{\hat\sigma(\Mcal)=\sigma'(\Mcal),\hat\sigma\in\sigma_\True^\Ccal}}{\Pr_{\hat\sigma\sim\prod_{v\in V}\Dcal_v}\sbra{\hat\sigma\in\sigma_\True^\Ccal}}
\cdot\frac{\Pr_{\sigma''\sim\prod_{v\in V}\Dcal_v|_{\tilde\sigma}}\sbra{\sigma''=\sigma'}}{\Pr_{\sigma''\sim\prod_{v\in V}\Dcal_v|_{\tilde\sigma}}\sbra{\sigma''\in\sigma_\True^{\Ccal|_{\tilde\sigma}}}}.
\end{align*}
By \Cref{def:projected_constraint_satisfaction_problem}, $\Dcal_v|_{\tilde\sigma}$ equals $\Dcal_v$ if $v\in V\setminus\Mcal$; and equals the point distribution of $\sigma'(v)$ if $v\in\Mcal$. Let $\sigma_1\sim\prod_{v\in V\setminus\Mcal}\Dcal_v$ and $\sigma_2\sim\prod_{v\in\Mcal}\Dcal_v$ be independent.
Then
\begin{align*}
\frac{\Pr_{\sigma''\sim\prod_{v\in V}\Dcal_v|_{\tilde\sigma}}\sbra{\sigma''=\sigma'}}{\Pr_{\sigma''\sim\prod_{v\in V}\Dcal_v|_{\tilde\sigma}}\sbra{\sigma''\in\sigma_\True^{\Ccal|_{\tilde\sigma}}}}
&=\frac{\Pr_{\sigma_1}\sbra{\sigma_1=\sigma'(V\setminus\Mcal)}}{\Pr_{\sigma_1}\sbra{\sigma_1\circ\sigma'(\Mcal)\in\sigma_\True^{\Ccal|_{\tilde\sigma}}}}
\tag{$\circ$ represents vector concatenation}\\
&=\frac{\Pr_{\sigma_1,\sigma_2}\sbra{\sigma_1=\sigma'(V\setminus\Mcal),\sigma_2=\sigma'(\Mcal)}}{\Pr_{\sigma_1,\sigma_2}\sbra{\sigma_1\circ\sigma'(\Mcal)\in\sigma_\True^{\Ccal|_{\tilde\sigma}},\sigma_2=\sigma'(\Mcal)}}
\tag{since $\sigma_1,\sigma_2$ are independent}\\
&=\frac{\Pr_{\sigma_1,\sigma_2}\sbra{\sigma_1\circ\sigma_2=\sigma'}}{\Pr_{\sigma_1,\sigma_2}\sbra{\sigma_1\circ\sigma_2\in\sigma_\True^{\Ccal|_{\tilde\sigma}},\sigma_2=\sigma'(\Mcal)}}\\
&=\frac{\Pr_{\hat\sigma\sim\prod_{v\in V}\Dcal_v}\sbra{\hat\sigma=\sigma'}}{\Pr_{\hat\sigma\sim\prod_{v\in V}\Dcal_v}\sbra{\hat\sigma\in\sigma_\True^{\Ccal|_{\tilde\sigma}},\hat\sigma(\Mcal)=\sigma'(\Mcal)}}\\
&=\frac{\Pr_{\hat\sigma\sim\prod_{v\in V}\Dcal_v}\sbra{\hat\sigma=\sigma'}}{\Pr_{\hat\sigma\sim\prod_{v\in V}\Dcal_v}\sbra{\hat\sigma\in\sigma_\True^\Ccal,\hat\sigma(\Mcal)=\sigma'(\Mcal)}}
\end{align*}
where the last equality is because when $\hat\sigma(\Mcal)=\sigma'(\Mcal)$, we have $\hat\sigma(\Mcal)=\tilde\sigma(\Mcal)$ and thus $\hat\sigma\in\sigma_\True^{\Ccal|_{\tilde\sigma}}$ iff $\hat\sigma\in\sigma_\True^\Ccal$.
Hence in all, we have 
$$
\Pr\sbra{\sigma=\sigma'}
=\frac{\Pr_{\hat\sigma\sim\prod_{v\in V}\Dcal_v}\sbra{\hat\sigma=\sigma'}}{\Pr_{\hat\sigma\sim\prod_{v\in V}\Dcal_v}\sbra{\hat\sigma\in\sigma_\True^\Ccal}}
=\Pr_{\hat\sigma\sim\prod_{v\in V}\Dcal_v}\sbra{\hat\sigma=\sigma'\mid\hat\sigma\in\sigma_\True^\Ccal}
=\mu_\True^\Ccal(\sigma')
$$
as desired.

Note that $\mu^\Mcal$ is a stationary distribution for $\Psf_i$ by Item (3)  \Cref{prop:systematicscan}. Meanwhile by Item (4) of \Cref{prop:coupling}, the $t$-th \textsf{for} iteration in \textsf{BoundingChain} is a coupling of $\Psf_{i_t}$. 
Thus in \textsf{FinalSampling}, upon receiving $\tilde\sigma$ which has distribution $\mu^\Mcal$, we can execute $|V|$ more rounds of \textsf{Line 3-11} in \Cref{alg:boundingchain} on $\tilde\sigma$ using fresh randomness, and the resulted assignment still has distribution $\mu^\Mcal$.
In other words, we may safely assume the last $|V|$ rounds of update in the final \textsf{BoundingChain} procedure are all using \textsf{RejectionSampling}.
Thus each $\abs{\Ccal_v}$ in \Cref{lem:finalsampling} also satisfies the concentration bound in \Cref{prop:component_size}. 
By a similar calculation in the proof of the efficiency part of \Cref{prop:boundingchain} and noticing $|\Ccal|\le d|V|$, the expected running time here is at most
\begin{equation*}
O\pbra{d^2kQ|V|\sum_{\ell=1}^{+\infty}\pbra{\frac1{32}}^{\ell-1}\cdot\Delta\cdot4^{\ell+1}}
=O\pbra{d^2kQ\Delta|V|}.
\tag*{\qedhere}
\end{equation*}
\end{proof}

\subsection{Putting Everything Together}\label{sec:putting_everything_together}

Now we put everything together to prove our main theorem.
\begin{proof}[Proof of \Cref{thm:atomiccspsampling}]
The correctness part follows immediately from \Cref{prop:finalsampling} and \Cref{prop:boundingchain_distribution}.
Thus recall measures defined in \Cref{def:(atomic)_constraint_satisfaction_problem} and we focus on the efficiency part. 
\begin{itemize}
\item Let $X$ be the total running time of \AtomicCSPSampling{$\Phi,\pibm$}.
\item Let $A$ be the time for computing $\beta(\Phi,\Mcal)$ in \Cref{eq:beta}. Then $A=O(k|\Ccal|)\le O(dk|V|)$.
\item For integer $i\ge1$ and $j\in[i]$, let $X_{i,j}$ be the running time of the $(-j)$-th \textsf{for} iteration in \BoundingChain{$\Phi,\Mcal,-i,r_{-i},\ldots,r_{-1}$}. Then $\E\sbra{X_{i,j}^2}=O\pbra{dk^2\Delta^5Q^2}$ by \Cref{prop:boundingchain}.
\item Let $T_\Final$ be the $T$ when the \textsf{while} iterations stop. Then by \Cref{prop:boundingchain}, we have $\Pr\sbra{T_\Final\ge t}\le4|V|\cdot2^{-t/|V|}$  for $t\ge2|V|-1$.
\item Let $Y$ be the running time of the \textsf{FinalSampling} in the end. Then by \Cref{prop:finalsampling} we have $\E[Y]=O\pbra{d^2kQ\Delta|V|}$.
\end{itemize}
Therefore we have $X=A+\sum_{t=0}^{\log\pbra{T_\Final}}\sum_{j=1}^{2^t}X_{2^t,j}+Y$\footnote{Technically we also need to initialize the randomness at \textsf{Line 1} of \Cref{alg:atomiccspsampling}, and check for \textsf{Line 5} in \Cref{alg:atomiccspsampling}, and initialize the assignment at \textsf{Line 1} of \Cref{alg:boundingchain}. However these can be done on the fly and their cost will be minor compared with the parts we listed.} and
\begin{align*}
\E[X]
&=\E[A]+\E[Y]+\sum_{t=0}^{+\infty}\sum_{j=1}^{2^t}\E\sbra{X_{2^t,j}\cdot[T_\Final\ge 2^t]}\\
&\le\E[A]+\E[Y]+\sum_{t=0}^{+\infty}\sum_{j=1}^{2^t}\sqrt{\E\sbra{X^2_{2^t,j}}\Pr\sbra{T_\Final\ge 2^t}}
\tag{by Cauchy-Schwarz inequality}\\
&\le\E[A]+\E[Y]+\sum_{t=0}^m\sum_{j=1}^{2^t}\sqrt{\E\sbra{X^2_{2^t,j}}}+\sum_{t=m+1}^{+\infty}\sum_{j=1}^{2^t}\sqrt{\E\sbra{X^2_{2^t,j}}\Pr\sbra{T_\Final\ge 2^t}}
\tag{$m\ge\lfloor\log(2|V|-1)\rfloor$ to be determined later}\\
&\le O\pbra{d^2kQ\Delta|V|+\sqrt{dk^2\Delta^5Q^2}\cdot\pbra{2^m+\sqrt{|V|}\sum_{t=m+1}^{+\infty}2^t\cdot2^{-\frac{2^t}{|V|}}}}.
\end{align*}
We pick $m=\lceil\log(|V|)+\log\log(|V|)+10\rceil$ then 
\begin{align*}
2^m+\sqrt{|V|}\sum_{t=m+1}^{+\infty}2^{t-\frac{2^t}{|V|}}
&\le2^m+\sqrt{|V|}\int_m^{+\infty}2^{x-\frac{2^x}{|V|}}\sd\!x
\tag{$2^{x-\frac{2^x}n}$ is decreasing when $2^x\ge\frac n{\ln(2)}$}\\
&=2^m+\sqrt{|V|}\cdot\frac{|V|}{\ln^2(2)}\cdot2^{-\frac{2^m}{|V|}}
\tag{since $\pbra{\frac{-n}{\ln^2(2)}\cdot2^{-\frac{2^x}n}}'=2^{x-\frac{2^x}n}$}\\
&=O\pbra{|V|\log(|V|)}.
\end{align*}
Since $d\le\Delta$, we have $\E[X]=O\pbra{kQ\Delta^3|V|\log(|V|)}$.
\end{proof}

\begin{remark}
Computing higher moments of $X_{i,j},Y$ and using possibly stronger assumption, one can improve the dependency on $k,\Delta,Q$ in the expected running time. However we view these as constants compared with $|V|$. Thus we do not make the effort here.
\end{remark}

\section{Applications}\label{sec:applications}

Here we instantiate \Cref{thm:atomiccspsampling} to special CSPs. 
We will use the following algorithmic Lov\'asz local lemma for constructing the marking $\Mcal$.
\begin{theorem}[\cite{moser2010constructive}]\label{thm:MT_LLL}
Let $\Phi=\pbra{V,\pbra{\Omega_v,\Dcal_v}_{v\in V},\Ccal}$ be a CSP. If $\Naturale p\Delta\le1$, then $\sigma_\True^\Ccal\neq\emptyset$ and there exists a randomized algorithm which outputs some $\sigma\in\sigma_\True^\Ccal$ in time $O(k\Delta|V|)$ with probability at least $0.99$.
\end{theorem}

We define the \emph{smooth parameter} of a CSP $\Phi$ by 
\begin{equation}\label{eq:kappa}
\kappa=\kappa(\Phi)=\max_{v\in V}\max_{q,q'\in\Omega_v}\frac{\Dcal_v(q)}{\Dcal_v(q')}.
\end{equation}
Note that $\kappa(\Phi)\ge1$ always.
When context is clear, we will simply write $\kappa$.

\subsection{Binary Domains}\label{sec:binary_domains}

Let $\Mcal\subseteq V$ be a marking. We specialize $\alpha,\beta,\lambda$ when all domains are of size $2$.
\begin{itemize}
\item $\alpha=\alpha(\Phi,\Mcal)=\max_{C\in\Ccal}\alpha(\Phi,\Mcal,C)$ where 
$$
\alpha(\Phi,\Mcal,C)=\prod_{v\in\vbl(C)\setminus\Mcal}\Dcal_v(\sigma_\False^C(v));
$$
\item When $\Naturale\alpha\le1$, define
\begin{itemize}
\item $\beta=\beta(\Phi,\Mcal)=(1-\Naturale\alpha)^{-d}\le(1-\Naturale\alpha)^{-\Delta}$;
\item $\lambda=\lambda(\Phi,\Mcal)=\max_{C\in\Ccal}\lambda(\Phi,\Mcal,C)$ where 
$$
\lambda(\Phi,\Mcal,C)
=\abs{\vbl(C)}^2\cdot\beta^{|\vbl(C)\cap\Mcal|}\cdot\prod_{v\in\vbl(C)\cap\Mcal}\Dcal_v(\sigma_\False^C(v)).
$$
\end{itemize}
\end{itemize}

By \Cref{rmk:atomiccspsampling}, we will use the following version of \Cref{thm:atomiccspsampling}.
\begin{theorem}[\Cref{thm:atomiccspsampling}, Binary Domains]\label{thm:binary_domains}
Let $\Phi=\pbra{V,\pbra{\Omega_v,\Dcal_v}_{v\in V},\Ccal}$ be an atomic CSP such that $|\Omega_v|=2$ for all $v\in V$. Let $\Mcal\subseteq V$ be a marking.
If 
$$
\Naturale\cdot\alpha(\Phi,\Mcal,C)\cdot\Delta\le1
\quad\text{and}\quad
\Delta^2\cdot\lambda(\Phi,\Mcal,C)\le1/100
\quad\text{for all }C\in\Ccal,
$$
then the following holds for \AtomicCSPSampling{$\Phi,\Mcal$}.
\begin{itemize}
\item \textsc{Correctness.} It halts almost surely and outputs $\sigma\sim\mu_\True^\Ccal$ when it halts.
\item \textsc{Efficiency.} Its expected total running time is $O\pbra{k\Delta^3|V|\log(|V|)}$.
\end{itemize}
\end{theorem}

We now construct a valid marking when the underlying distributions are arbitrary.
\begin{lemma}\label{lem:binary_domains}
Let $\Phi=\pbra{V,\pbra{\Omega_v,\Dcal_v}_{v\in V},\Ccal}$ be an atomic CSP such that $|\Omega_v|=2$ for all $v\in V$. 
For any $\zeta\in(0,1)$, if $p^\gamma\cdot\Delta\le0.01\cdot\zeta/\kappa$ where 
$$
\gamma=\frac{3-9\zeta+\ln(\kappa+1)-\sqrt{\ln^2(\kappa+1)+6\cdot(1-3\zeta)\ln(\kappa+1)}}9,
$$ 
then there exists a marking $\Mcal\subseteq V$ such that 
$$
\Naturale\cdot\alpha(\Phi,\Mcal,C)\cdot\Delta\le1
\quad\text{and}\quad
\Delta^2\cdot\lambda(\Phi,\Mcal,C)\le1/100
\quad\text{for all }C\in\Ccal.
$$ 
Moreover $\Mcal$ can be constructed in time $O(k\Delta|V|)$ with success probability at least $0.99$.
\end{lemma}
\begin{proof}
For each $C\in\Ccal$, define $p_C=\prod_{v\in\vbl(C)}\Dcal_v(\sigma_\False^C(v))$. 
Then $p=p(\Phi)=\max_{C\in\Ccal}p_C$.
Meanwhile by the definition of $\kappa$, we know $(1/(\kappa+1))^{|\vbl(C)|}\le p_C\le(\kappa/(\kappa+1))^{|\vbl(C)|}$.
Thus 
\begin{equation}\label{eq:vbl(C)}
|\vbl(C)|\le\frac{\ln(1/p_C)}{\ln(1+1/\kappa)}.
\end{equation}
Since $x^\zeta\ge\zeta\ln(x)$ holds for any $x>0$, we also have
\begin{equation}\label{eq:giant}
\ln(1/p_C)\le p_C^{-\zeta}/\zeta.
\end{equation}

Let $\eta,\tau\in(0,1)$ be parameters and $\Mcal$ be the marking to be determined later. We will ensure $1-\eta-\tau\ge0$ and $\eta-\tau-3\zeta\ge0$.
For each $C\in\Ccal$, let $\Ecal_C$ be the event (i.e., constraint) 
$$
\Ecal_C=\text{`` }
\abs{\ln\pbra{\prod_{v\in\vbl(C)\cap\Mcal}\Dcal_v(\sigma_\False^C(v))}-\eta\ln(p_C)}>\tau\ln(1/p_C)
\text{ ''.}
$$
Now we check $\Naturale\cdot\alpha(\Phi,\Mcal,C)\cdot\Delta\le1$ and $\Delta^2\cdot\lambda(\Phi,\Mcal,C)\le1/100$ assuming no $\Ecal_C$ happens.
Since $\vbl(C)$ is the disjoint union of $\vbl(C)\cap\Mcal$ and $\vbl(C)\setminus\Mcal$, if $\Ecal_C$ does not happen, then 
$$
\alpha(\Phi,\Mcal,C)
\le p_C^{1-\eta-\tau}\le p^{1-\eta-\tau}.
$$
Since $(1-x)^{-1/x}\le4$ holds for any $x\in(0,1/2]$, we have 
\begin{equation}\label{eq:beta_binary}
\beta
\le\pbra{1-\Naturale\cdot p^{1-\eta-\tau}}^{-\Delta}
\le4^{\Naturale\cdot p^{1-\eta-\tau}\cdot\Delta}
\le\pbra{\frac{\kappa+1}\kappa}^\zeta
\quad\text{if}\quad
2\Naturale\ln(2)\cdot p^{1-\eta-\tau}\cdot\Delta\le\zeta\cdot\ln\pbra{1+1/\kappa}.
\end{equation}
Note that $2\Naturale\ln(2)\cdot p^{1-\eta-\tau}\cdot\Delta\le\zeta\cdot\ln\pbra{1+1/\kappa}$ already implies $\Naturale\cdot\alpha(\Phi,\Mcal,C)\cdot\Delta\le1$.
Combining \Cref{eq:vbl(C)}, \Cref{eq:giant}, and \Cref{eq:beta_binary}, we also have
$$
\lambda(\Phi,\Mcal,C)
\le|\vbl(C)|^2\cdot\beta^{|\vbl(C)|}p_C^{\eta-\tau}
\le\frac{\ln^2(1/p_C)\cdot p_C^{\eta-\tau-\zeta}}{\ln^2(1+1/\kappa)}
\le\frac{p_C^{\eta-\tau-3\zeta}}{\zeta^2\ln^2(1+1/\kappa)}
\le\frac{p^{\eta-\tau-3\zeta}}{\zeta^2\ln^2(1+1/\kappa)}.
$$
Therefore $\Delta^2\cdot\lambda(\Phi,\Mcal,C)\le1/100$ is reduced to $\Delta^2\cdot p^{\eta-\tau-3\zeta}\le0.01\cdot\zeta^2\cdot\ln^2(1+1/\kappa)$.
In all, it suffices to make sure $1-\eta-\tau\ge0$, $\eta-\tau-3\zeta\ge0$, and
\begin{equation}\label{eq:generalcsp_arbitrary}
2\Naturale\ln(2)\cdot p^{1-\eta-\tau}\cdot\Delta\le\zeta\cdot\ln(1+1/\kappa),
\quad\text{and}\quad
\Delta^2\cdot p^{\eta-\tau-3\zeta}\le0.01\cdot\zeta^2\cdot\ln^2(1+1/\kappa).
\end{equation}

Now we show how to set $\eta,\tau$ and construct $\Mcal$ to make sure no $\Ecal_C$ happens.
We put each $v\in V$ into $\Mcal$ independently with probability $\eta$. For each $v\in V$, let $x_v\in\bin$ be the indicator for whether $v$ is in $\Mcal$. Then
\begin{align*}
\Ecal_C=\text{`` }\abs{\sum_{v\in\vbl(C)}x_v\ln\pbra{1/\Dcal_v(\sigma_\False^C(v))}-\E\sbra{\sum_{v\in\vbl(C)}x_v\ln\pbra{1/\Dcal_v(\sigma_\False^C(v))}}}>\tau\ln(1/p_C)\text{ ''.}
\end{align*}
By Hoeffding's inequality \cite[Theorem 2]{hoeffding1994probability}, we have
\begin{align*}
\Pr\sbra{\Ecal_C}
&\le2\exp\cbra{\frac{-2\tau^2\ln^2(1/p_C)}{\sum_{v\in\vbl(C)}\ln^2\pbra{1/\Dcal_v(\sigma_\False^C(v))}}}\\
&\le2\exp\cbra{\frac{-2\tau^2\ln^2(1/p_C)}{\ln(\kappa+1)\cdot\sum_{v\in\vbl(C)}\ln\pbra{1/\Dcal_v(\sigma_\False^C(v))}}}
\tag{by the definition of $\kappa$}\\
&=2\exp\cbra{\frac{-2\tau^2\ln(1/p_C)}{\ln(\kappa+1)}}
\tag{since $\ln(p_C)=\sum_{v\in\vbl(C)}\ln\pbra{\Dcal_v(\sigma_\False^C(v))}$}\\
&=2\cdot p_C^{2\tau^2/\ln(\kappa+1)}\le2\cdot p^{2\tau^2/\ln(\kappa+1)}.
\end{align*}
Since $\vbl(\Ecal_C)=\vbl(C)$ and thus it correlates with $\Delta$ many $\Ecal_{C'}$ (including itself), by \Cref{thm:MT_LLL} we can construct $\Mcal$ (i.e., $\pbra{x_v}_{v\in V}$) to avoid all $\Ecal_C$ in time $O(k\Delta|V|)$ with probability at least 0.99 as long as 
$$
2\Naturale\cdot p^{2\tau^2/\ln(\kappa+1)}\cdot\Delta\le1.
$$

Now we set $\eta=(2-\tau+3\zeta)/3$ and 
$$
\tau=\frac{-\ln(\kappa+1)+\sqrt{\ln^2(\kappa+1)+6\cdot(1-3\zeta)\ln(\kappa+1)}}6<\frac{1-3\zeta}2
$$
so that $1-\eta-\tau$, $(\eta-\tau-3\zeta)/2$, and $2\tau^2/\ln(\kappa+1)$ all equal
$$
\gamma=\frac13-\zeta-\frac{-\ln(\kappa+1)+\sqrt{\ln^2(\kappa+1)+6\cdot(1-3\zeta)\ln(\kappa+1)}}9>0.
$$
Then all the conditions in \Cref{eq:generalcsp_arbitrary} boil down to $p^\gamma\cdot\Delta\le0.1\cdot\zeta\cdot\ln(1+1/\kappa)$. 
Since $\kappa\ge1$ and $\ln(1+x)\ge0.1x$ for all $x\in[0,1]$, we can safely replace $\ln(1+1/\kappa)$ with $0.1/\kappa$ as desired in the statement.
\end{proof}

Note that whether $\Mcal$ satisfies the conditions in \Cref{thm:binary_domains} can be easily checked in time $O(k|\Ccal|)=O(k\Delta|V|)$ by computing $\alpha(\Phi,\Mcal,C)$ and $\lambda(\Phi,\Mcal,C)$ for each $C\in\Ccal$. Thus we can keep performing \Cref{lem:binary_domains} until the marking $\Mcal$ is valid and then we run \AtomicCSPSampling{$\Phi,\Mcal$}. This provides a Las Vegas algorithm as below.
\begin{corollary}\label{cor:binary_domains}
There exists a Las Vegas algorithm which takes as input an atomic CSP $\Phi=\pbra{V,\pbra{\Omega_v,\Dcal_v}_{v\in V},\Ccal}$ and a parameter $\zeta\in(0,1)$ such that the following holds.

If $|\Omega_v|=2$ holds for all $v\in V$ and $p^\gamma\cdot\Delta\le0.01\cdot\zeta/\kappa$ where
$$
\gamma=\frac{3-9\zeta+\ln(\kappa+1)-\sqrt{\ln^2(\kappa+1)+6\cdot(1-3\zeta)\ln(\kappa+1)}}9,
$$
then the algorithm outputs a random solution of $\Phi$ distributed perfectly as $\mu_\True^\Ccal$ in expected time $O\pbra{k\Delta^3|V|\log(|V|)}$.
\end{corollary}
\begin{remark}\label{rmk:binary_domains}
One natural choice of the underlying distributions is the uniform distribution. In this case $\kappa=1$. By setting $\zeta\to0$, we have
$$
\gamma\to\frac{3+\ln(2)-\sqrt{\ln^2(2)+6\ln(2)}}9>0.171.
$$
For example we can set $\zeta=10^{-10}$ and the local lemma condition is simply $p^{0.171}\cdot\Delta\le10^{-12}/\kappa$.
In \Cref{lem:generalcsp_uniform_simple} we will optimize it to $0.175$ by a tighter concentration bound.
\end{remark}

\subsection{Large Domains: State Tensorization}\label{sec:large_domains:state_tensorization}

Here we formally introduce the \emph{state tensorization} technique, generalizing \emph{state compression} from \cite{feng2020sampling}.
This, as emphasized in \Cref{sec:proof_overview}, allows us to transform a large domain into a product of binary domains.

Let $\Omega$ be a finite domain of size at least $2$ and $\Dcal$ be a distribution supported on $\Omega$.
A \emph{state tensorization} for $(\Omega,\Dcal)$ (See \Cref{fig:state_tensorization} for a concrete example) is a rooted tree $\Tcal$ where 
\begin{itemize}
\item $\Tcal$ has $|\Omega|$ leaves and each internal node of $\Tcal$ has at least two child nodes;
\item the leaves of $\Tcal$ have a one-to-one correspondence with elements in $\Omega$.
\end{itemize}
For each node $z\in\Tcal$, let $\leafs(z)$ be the set of leaves in the sub-tree of $z$. Then $\leafs(\mathsf{rt})=\Omega$ for root $\mathsf{rt}$.
For each internal node $z$, we use $\childs(z)$ to denote its child nodes. For any $z'\in\childs(z)$, we use $z\to z'$ to denote the edge from $z$ to $z'$. Moreover, we define the weight of $z\to z'$ as 
\begin{equation}\label{eq:edge_weight}
W(z\to z')=\frac{\sum_{q\in\leafs(z')}\Dcal(q)}{\sum_{q\in\leafs(z)}\Dcal(q)}.
\end{equation}
It is easy to see the total weight of outgoing edges of any internal node is $1$.

If $|\Omega|=1$ and thus $\Dcal$ is the point distribution, then the state tensorization $\Tcal$ for $(\Omega,\Dcal)$ has two nodes $z$ and $z'$ where $z$ is the root and $z'$ is the only leaf and $W(z\to z')=1$.

\begin{figure}[ht]
\centering
\begin{tikzpicture}[emptyC/.style={draw,circle,inner sep=2pt}, outE/.style={->,>=stealth,thick}]

\node (vv) at (-8,-2) {State tensorization $\Tcal$ for $(\Omega,\Dcal)$};
\node[emptyC] (v0) at (0,0) {$z_0$};
\node[emptyC] (v1) [below left=of v0, xshift=0cm] {$z_1$};
\node[emptyC] (v2) [below right=of v0, xshift=0cm] {$z_2$};
\node[emptyC] (v3) [below left=of v1,xshift=0cm] {$z_3$};

\node (q1) [below right=of v1,xshift=-1cm] {$d$};
\node (q2) [below left=of v2,xshift=1cm] {$e$};
\node (q3) [below right=of v2] {$f$};
\node (q4) [below left=of v3] {$\phantom{b}a\phantom{b}$};
\node (q5) [below=of v3] {$\phantom{b}b\phantom{b}$};
\node (q6) [below right=of v3] {$\phantom{b}c\phantom{b}$};

\draw[outE] (v0) -- node[xshift=-2mm, yshift=2mm]{$\frac35$} (v1);
\draw[outE] (v0) -- node[xshift=2mm, yshift=2mm]{$\frac25$} (v2);
\draw[outE] (v1) -- node[xshift=-2mm, yshift=2mm]{$\frac56$} (v3);
\draw[outE] (v1) -- node[xshift=2mm, yshift=2mm]{$\frac16$} (q1);
\draw[outE] (v2) -- node[xshift=-2mm, yshift=2mm]{$\frac14$} (q2);
\draw[outE] (v2) -- node[xshift=2mm, yshift=2mm]{$\frac34$} (q3);
\draw[outE] (v3) -- node[xshift=-2mm, yshift=2mm]{$\frac13$} (q4);
\draw[outE] (v3) -- node[xshift=2mm, yshift=2mm]{$\frac13$} (q5);
\draw[outE] (v3) -- node[xshift=2mm, yshift=2mm]{$\frac13$} (q6);
\end{tikzpicture}
\caption{One example of $\Tcal$ for $(\Omega,\Dcal)$ where $\Omega=\cbra{a,b,c,d,e,f}$ and $\Dcal(a)=\Dcal(b)=\Dcal(c)=1/6,\Dcal(d)=\Dcal(e)=1/10,\Dcal(f)=3/10$. We omit the leaf nodes.}\label{fig:state_tensorization}
\end{figure}
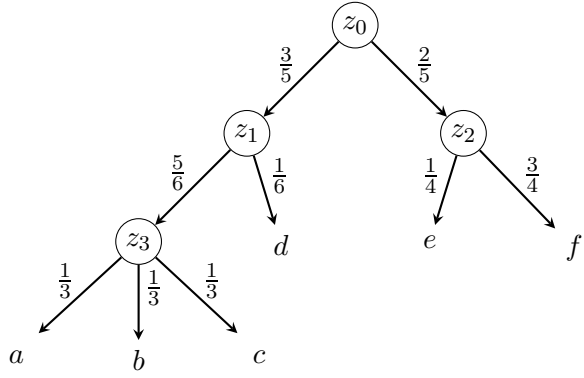

For each $q\in\Omega$, we use $\tpath(q,\Tcal)$ to denote the set of internal nodes in $\Tcal$ on the path from the root to the leaf $z$ that corresponds to $q$. 
For example in \Cref{fig:state_tensorization}, $\tpath(b,\Tcal)=\cbra{z_0,z_1,z_3}$. 

We first observe the following fact regarding edge weights in $\Tcal$.
\begin{fact}\label{fct:edge_weight}
Let $q\in\Omega$ be arbitrary. Let $\tpath(q,\Tcal)=\cbra{z_0,\ldots,z_\ell}$ and $z_{\ell+1}$ be the leaf node corresponding to $q$. Assume $z_0,\ldots,z_{\ell+1}$ is in the root-to-leaf order. Then $\Dcal(q)=\prod_{i=0}^\ell W(z_i\to z_{i+1})$.
\end{fact}
\begin{proof}
By \Cref{eq:edge_weight}, we have
\begin{equation*}
\prod_{i=0}^\ell W(z_i\to z_{i+1})
=\prod_{i=0}^\ell\frac{\sum_{q'\in\leafs(z_{i+1})}\Dcal(q')}{\sum_{q'\in\leafs(z_i)}\Dcal(q')}
=\frac{\sum_{q'\in\leafs(z_{\ell+1})}\Dcal(q')}{\sum_{q'\in\leafs(z_0)}\Dcal(q')}
=\Dcal(q)
\tag*{\qedhere}
\end{equation*}
\end{proof}

Now we move to atomic CSPs and show formally how state tensorization helps reduce domain sizes.
\begin{definition}[Tensorized Atomic Constraint Satisfaction Problem]\label{def:tensorized_atomic_constraint_satisfaction_problem}
Let $\Phi=\pbra{V,\pbra{\Omega_v,\Dcal_v}_{v\in V},\Ccal}$ be an atomic CSP.
Let $\pbra{\Tcal_v}_{v\in V}$ be state tensorizations where each $\Tcal_v$ is a state tensorization for $(\Omega_v,\Dcal_v)$.\footnote{We require each $\Tcal_v$ uses different set of tree nodes. So there will be no confusion when using $z$ without explicitly providing $\Tcal_v\ni z$.}
We construct $\Phi^\otimes=\pbra{Z,\pbra{\Omega_z,\Dcal_z}_{z\in Z},\Ccal^\otimes}$ as the \emph{tensorized atomic constraint satisfaction problem}:\footnote{Technically $\Phi^\otimes$ depends on $\pbra{\Tcal_v}_{v\in V}$ which we omit here for simplicity. In addition we remark $\Phi^\otimes$ may have redundant variables that do not appear in any constraint; this is indeed consistent with our definition of CSPs as we never require them to shave redundant variables.}
\begin{itemize}
\item $Z$ is the set of internal nodes of all $\Tcal_v$.
\item For each $z\in Z$, $\Omega_z=\childs(z)$ and $\Dcal_z$ is a distribution supported on $\Omega_z$ by setting $\Dcal_z(z')=W(z\to z')$ for all $z'\in\childs(z)$.
\item For each $C\in\Ccal$, we construct $C^\otimes\in\Ccal^\otimes$ by setting $$
\vbl(C^\otimes)=\bigcup_{v\in\vbl(C)}\tpath(\sigma_\False^C(v),\Tcal_v)
$$
and $C^\otimes(\sigma)=\False$ iff $\sigma(z)=z(\sigma_\False^C(v))$ for all $v\in\vbl(C)$ and $z\in\tpath(\sigma_\False^C(v))$ where $z(q)\in\childs(z)$ is the child node of $z$ such that $q\in\leafs(z(q))$.\footnote{Assume $\tpath(q,\Tcal_v)=\cbra{z_0,\ldots,z_\ell}$ where $z_0,\ldots,z_\ell$ is in the top-down order of $\Tcal_v$. Let $z_{\ell+1}$ be the leaf node corresponding to $q$. Then $z_i(q)=z_{i+1}$ for each $i\in\cbra{0,\ldots,\ell}$.}
\end{itemize}
\end{definition}

For example in \Cref{fig:state_tensorization}, we have $\Omega_{z_0}=\cbra{z_1,z_2}$ and $\Dcal_{z_0}(z_1)=3/5,\Dcal_{z_0}(z_2)=2/5$. Assume constraint $C$ is false iff $d$ is assigned. Then $C^\otimes$ will have $\vbl(C^\otimes)=\tpath(d,\Tcal)=\cbra{z_0,z_1}$ and, for $\sigma\in\Omega_{z_0}\times\Omega_{z_1}\times\Omega_{z_2}\times\Omega_{z_3}$, $C^\otimes(\sigma)$ iff $\sigma(z_0)=z_1$ and $\sigma(z_1)$ equals the leaf node corresponding to $d$.

Recall measures defined in \Cref{def:(atomic)_constraint_satisfaction_problem}. Here we note some basic facts for \Cref{def:tensorized_atomic_constraint_satisfaction_problem}.
\begin{fact}\label{fct:tensorized_atomic_constraint_satisfaction_problem}
The following holds for $\Phi^\otimes$.
\begin{itemize}
\item[(1)] $\Phi^\otimes$ is an atomic CSP where $|\Ccal^\otimes|=|\Ccal|$ and $|Z|=\sum_{v\in V}|\cbra{\text{internal nodes in }\Tcal_v}|\le Q(\Phi)|V|$.
\item[(2)] $Q(\Phi^\otimes)=\max_{z\in Z}|\childs(z)|$ and $k(\Phi^\otimes)\le k(\Phi)\cdot\max_{C\in\Ccal,v\in\vbl(C)}\abs{\tpath(\sigma_\False^C(v),\Tcal_v)}$
\item[(3)] $\Delta(\Phi^\otimes)=\Delta(\Phi)$, $d(\Phi^\otimes)=d(\Phi)$, and $p(\Phi^\otimes)=p(\Phi)$. Moreover for all $C\in\Ccal$ we have
$$
\Pr_{\sigma\sim\prod_{v\in\vbl(C)}\Dcal_v}\sbra{C(\sigma)=\False}=\Pr_{\sigma'\sim\prod_{z\in\vbl(C^\otimes)}\Dcal_z}\sbra{C^\otimes(\sigma')=\False}.
$$
\end{itemize}
\end{fact}
\begin{proof}
Item (1) (2) are obvious. Now we focus on Item (3). 
Note that for any $C\in\Ccal$ and $v\in V$, $v\in\vbl(C)$ iff the root of $\Tcal_v$ is in $\vbl(C^\otimes)$. On the other hand if some internal node $z$ is in $\vbl(C^\otimes)$ then all the ancestors of $z$ are also in $\vbl(C^\otimes)$.
Thus $\Delta(\Phi^\otimes)=\Delta(\Phi)$ and $d(\Phi^\otimes)=d(\Phi)$.
The ``moreover'' part and $p(\Phi^\otimes)=p(\Phi)$ follow from \Cref{fct:edge_weight}.
\end{proof}

Now we formally describe how to translate an assignment for $\Phi^\otimes$ into an assignment for $\Phi$. For any $\sigma\in\prod_{z\in Z}\Omega_z$, we define $\sigma^\Trans\in\prod_{v\in V}\Omega_v$ by the following process for each $v\in V$:
\begin{itemize}
\item We start from the root of $\Tcal_v$.
\item If we are at an internal node $z$ of $\Tcal_v$, then proceed to its child node $\sigma(z)$ and repeat.
\item Otherwise we are at a leaf node, then set $\sigma^\Trans(v)$ by its corresponding value in $\Omega_v$.
\end{itemize}

Finally we prove the following simple but powerful reduction result.
\begin{proposition}\label{prop:state_tensorization}
If $\sigma\sim\mu_\True^{\Ccal^\otimes}$, then $\sigma^\Trans\sim\mu_\True^\Ccal$. 

Therefore to obtain a random solution of $\Phi$ distributed perfectly as $\mu_\True^\Ccal$, it suffices to have a random solution of $\Phi^\otimes$ distributed perfectly as $\mu_\True^{\Ccal^\otimes}$ and then perform $\Trans$ operation.
\end{proposition}
\begin{proof}
Recall \RejectionSampling{$\Phi^\otimes,\cdot$} in \Cref{alg:rejectionsampling}:
\begin{itemize}
\item[(1)] Sample $\sigma\sim\prod_{z\in Z}\Dcal_z$.
\item[(2)] If $C^\otimes(\sigma)=\True$ for all $C^\otimes\in\Ccal^\otimes$, then accept $\sigma$; otherwise resample $\sigma$.
\end{itemize}
Obviously $\sigma\sim\mu_\True^{\Ccal^\otimes}$. 
By the definition of $C^\otimes$ in \Cref{def:tensorized_atomic_constraint_satisfaction_problem} and the definition of $\Trans$ above, we have $C^\otimes(\sigma)=\True$ iff $C(\sigma^\Trans)=\True$. Thus we can safely replace Step (2) by
\begin{itemize}
\item[(2a)] If $C(\sigma^\Trans)=\True$ for all $C\in\Ccal$, then accept $\sigma$; otherwise resample $\sigma$.
\end{itemize}
On the other hand for Step (1), we can first sample the $\sigma^\Trans$ part and then complete it to $\sigma$:
\begin{itemize}
\item[(1a)] For each $v\in V$, we start from the root of $\Tcal_v$.
\begin{itemize}
\item If we are at an internal node $z$ of $\Tcal_v$, then sample $\sigma(z)\sim\Dcal_z$ and proceed to $\sigma(z)$.
\item Otherwise we are at a leaf node, then move to the next variable in $V$.
\end{itemize}
\item[(1b)] For each $z\in Z$ that $\sigma(z)$ is not sampled in Step (1a), we complete it by $\sigma(z)\sim\Dcal_z$.
\end{itemize}
By \Cref{fct:edge_weight} and the definition of $\Trans$, Step (1a) is equivalent to sample $\sigma^\Trans\sim\prod_{v\in V}\Dcal_v$ and $\sigma^\Trans$ does not depend on the values sampled in Step (1b).
More intuitively, we express this hybrid argument on rejection sampling through the following equivalence of flow charts:
\begin{align*}
&\text{Step (1)}\leftrightarrows\text{Step (2)}\rightarrow\sigma^\Trans\\
\text{equals}\quad&\text{Step (1a) (1b)}\leftrightarrows\text{Step (2a)}\rightarrow\sigma^\Trans\\
\text{equals}\quad&\text{Step (1a)}\leftrightarrows\text{Step (2a)}\rightarrow\text{Step (1b)}\rightarrow\sigma^\Trans\\
\text{equals}\quad&\text{Step (1a)}\leftrightarrows\text{Step (2a)}\rightarrow\sigma^\Trans,
\end{align*}
where the last line is exactly \RejectionSampling{$\Phi,\cdot$}. Thus $\sigma^\Trans\sim\mu_\True^\Ccal$.
\end{proof}

\subsection{General Atomic Constraint Satisfaction Problem: Arbitrary Distribution}\label{sec:general_atomic_constraint_satisfaction_problem:arbitrary_distribution}

Now we deal with general atomic CSPs where domains may be large. Firstly we prove the following lemma which describes a balanced way to construct state tensorization.
\begin{lemma}\label{lem:construct_state_tensorization}
There exists a deterministic algorithm such that the following holds.
Let $\Omega$ be a finite domain of size at least $2$ and $\Dcal$ be a distribution supported on $\Omega$. Let $\kappa=\max_{q,q'\in\Omega}\Dcal(q)/\Dcal(q')$.
The algorithm constructs a state tensorization $\Tcal$ for $(\Omega,\Dcal)$ where
\begin{itemize}
\item[(1)] the algorithm runs in time $O(|\Omega|\log(|\Omega|))$ and $\Tcal$ is a binary tree;
\item[(2)] for any internal node $z\in\Tcal$ and $\cbra{z_1,z_2}=\childs(z)$, we have $\frac{W(z\to z_1)}{W(z\to z_2)}\le\max\cbra{\kappa,2}$.
\end{itemize}
\end{lemma}
\begin{proof}
$\Tcal$ is constructed like a Huffman tree as follows:
\begin{itemize}
\item For each $q\in\Omega$ create a node $z_q$ with value $\mathsf{val}(z_q)=\Dcal(q)$. Initialize set $S=\cbra{z_q\mid q\in\Omega}$.
\item While $|S|\ge2$, 
\begin{itemize}
\item select two distinct nodes $z_1,z_2\in S$ with minimum value, i.e., $\mathsf{val}(z_1),\mathsf{val}(z_2)$ are the smallest among all nodes in $S$,
\item create a parent node $z$ of $z_1$ and $z_2$; then set $\mathsf{val}(z)=\mathsf{val}(z_1)+\mathsf{val}(z_2)$ and update $S\gets\pbra{S\cup\cbra{z}}\setminus\cbra{z_1,z_2}$.
Note that $W(z\to z_1)=\frac{\mathsf{val}(z_1)}{\mathsf{val}(z)}$ and $W(z\to z_2)=\frac{\mathsf{val}(z_2)}{\mathsf{val}(z)}$.
\end{itemize}
\item The final node in $S$ when $|S|=1$ is the root of $\Tcal$.
\end{itemize}
Then Item (1) is obvious if $S$ is implemented as a balanced binary search tree or a heap.

Now we turn to Item (2). Define $\kappa(S)=\max_{z,z'\in S}\mathsf{val}(z)/\mathsf{val}(z')$ when $|S|\ge2$. Then each time we select $z_1,z_2\in S$ and link them to $z$, we have $\frac{W(z\to z_1)}{W(z\to z_2)}\le\kappa(S)$. Therefore it suffices to show $\kappa(S)\le\max\cbra{\kappa,2}$ throughout the construction. 
\begin{itemize}
\item \fbox{Initialization.} 
Then we simply have $\kappa(S)=\kappa$.
\item \fbox{Afterwards.}
Assume $S=\cbra{z_1,z_2,\ldots,z_\ell}$ for $\ell\ge2$ where $\mathsf{val}(z_1)\le\mathsf{val}(z_2)\le\cdots\le\mathsf{val}(z_\ell)$. Let $S'=\cbra{z,z_3,\ldots,z_\ell}$ be $S$ after the update where $\mathsf{val}(z)=\mathsf{val}(z_1)+\mathsf{val}(z_2)$. Now we have two possible cases.
\begin{itemize}
\item If $\mathsf{val}(z)\le\mathsf{val}(z_\ell)$, then 
$$
\kappa(S')
=\max\cbra{\frac{\mathsf{val}(z_\ell)}{\mathsf{val}(z)},\frac{\mathsf{val}(z_\ell)}{\mathsf{val}(z_3)}}
\le\frac{\mathsf{val}(z_\ell)}{\mathsf{val}(z_1)}
=\kappa(S)\le\max\cbra{\kappa,2}.
$$
\item If $\mathsf{val}(z)>\mathsf{val}(z_\ell)$, then 
\begin{equation*}
\kappa(S')
=\frac{\mathsf{val}(z)}{\mathsf{val}(z_3)}
=\frac{\mathsf{val}(z_1)+\mathsf{val}(z_2)}{\mathsf{val}(z_3)}
\le2.
\tag*{\qedhere}
\end{equation*}
\end{itemize}
\end{itemize}
\end{proof}

By \Cref{lem:construct_state_tensorization}, \Cref{prop:state_tensorization}, and \Cref{cor:binary_domains}, we have the following theorem.
\begin{theorem}\label{thm:generalcsp_arbitrary}
There exists a Las Vegas algorithm which takes as input an atomic CSP $\Phi=\pbra{V,\pbra{\Omega_v,\Dcal_v}_{v\in V},\Ccal}$ and a parameter $\zeta\in(0,1)$ such that the following holds.

Let $\tilde\kappa=\max\cbra{\kappa,2}$ (Recall $\kappa=\kappa(\Phi)$ from \Cref{eq:kappa}). 
If $p^\gamma\cdot\Delta\le0.01\cdot\zeta/\tilde\kappa$ where
$$
\gamma=\frac{3-9\zeta+\ln(\tilde\kappa+1)-\sqrt{\ln^2(\tilde\kappa+1)+6\cdot(1-3\zeta)\ln(\tilde\kappa+1)}}9,
$$
then the algorithm outputs a random solution of $\Phi$ distributed perfectly as $\mu_\True^\Ccal$ in expected time $O\pbra{k\Delta^3Q^2|V|\log(Q|V|)}$.
\end{theorem}
\begin{proof}
By fixing the variable which has domain size $1$, we may safely assume $|\Omega_v|\ge2$ for all $v\in V$.

We use \Cref{lem:construct_state_tensorization} to construct $\Tcal_v$ for each $(\Omega_v,\Dcal_v)$. This takes time $O(Q(\Phi)\log(Q(\Phi)))\cdot|V|$.
Then let $\Phi^\otimes=\pbra{Z,\pbra{\Omega_z,\Dcal_z}_{z\in Z},\Ccal^\otimes}$ be the tensorized atomic CSP (Defined in \Cref{def:tensorized_atomic_constraint_satisfaction_problem}).
By Item (2) of \Cref{lem:construct_state_tensorization}, we have $\kappa(\Phi^\otimes)\le\max\cbra{\kappa(\Phi),2}$.
Also by \Cref{fct:tensorized_atomic_constraint_satisfaction_problem}, we have $\Delta(\Phi^\otimes)=\Delta(\Phi)$, $p(\Phi^\otimes)=p(\Phi)$, $k(\Phi^\otimes)\le k(\Phi)\cdot Q(\Phi)$\footnote{It is possible to get a better bound on $k(\Phi^\otimes)$ (for example $k(\Phi^\otimes)\le k(\Phi)\cdot\frac{\ln(\kappa(\Phi)Q(\Phi)+1)}{\ln(1+1/\kappa(\Phi))}$) by analyzing the depth of each $\Tcal_v$. This is because the construction of $\Tcal_v$ is ``balanced'' and the support of $\Dcal_v$ has size $|\Omega_v|\le Q(\Phi)$ only. However this only slightly improves the bound on the running time which is not our main focus.}, and $|Z|\le Q(\Phi)|V|$.
By \Cref{prop:state_tensorization}, it suffices to obtain a random solution of $\Phi^\otimes$ distributed perfectly as $\mu_\True^{\Ccal^\otimes}$, which, by \Cref{cor:binary_domains}, gives the claimed bounds.
\end{proof}

\begin{remark}\label{rmk:generalcsp_arbitrary}
Applying \Cref{thm:generalcsp_arbitrary} to the uniform distributions, i.e., $\kappa=1$, and setting $\zeta\to0$, we have
$$
\gamma\to\frac{3+\ln(3)-\sqrt{\ln^2(3)+6\ln(3)}}9>0.145.
$$
This already beats $1/7$ from \cite{Vishesh20towards} and $0.142$ from \cite{Vishesh21sampling}.
In \Cref{cor:generalcsp_uniform} we will further optimize it to $0.175$ with a more refined construction.
\end{remark}

\subsection{Hypergraph Coloring}\label{sec:hypergraph_coloring}

The previous bounds can be improved if the domains are large enough and the underlying distributions are smooth. Here we take the hypergraph coloring problem as an example.

\begin{definition}[Hypergraph Coloring]\label{def:hypergraph_coloring}
Let $Q$ and $k$ be positive integers. Let $H=(V,E)$ be a $k$-uniform hypergraph, i.e., each edge $e\in E$ contains exactly $k$ distinct variables.
We associate it with an atomic CSP $\Phi=\Phi(H,Q)=\pbra{V,\pbra{[Q],\Ucal}^V,\Ccal}$ where $\Ucal$ is the uniform distribution over $[Q]$ and $\Ccal=\cbra{C_{e,i}\colon[Q]^V\to\binTF\mid e\in E,i\in[Q]}$ and $C_{e,i}(\sigma)=\False$ iff $\sigma(v)=i$ for all $v\in e$. A solution to $\Phi$ is called a \emph{proper coloring} for $H$.
\end{definition}

To avoid confusion, $k,d,\Delta$ will only be referred to $k(H),d(H),\Delta(H)$.
It is easy to see $k(\Phi)=k$, $d(\Phi)=Q\cdot d$, $\Delta(\Phi)=Q\cdot\Delta$, $Q(\Phi)=Q$, and $p(\Phi)=Q^{-k}$. 

\begin{theorem}\label{thm:hypergraph_coloring}
There exists a Las Vegas algorithm which takes as input a $k$-uniform hypergraph $H=(V,E)$ and an integer $Q$. If $\Delta=\Delta(H)\le Q^{\pbra{1/3-o_{Q,k}(1)}k}$, then the algorithm outputs a perfect uniform random proper coloring for $H$ in expected time $O\pbra{k\Delta^3Q^4\log(Q)|V|\log(Q|V|)}$.
\end{theorem}
\begin{proof}
For each $v\in V$, we construct the state tensorization $\Tcal_v$ for $(\Omega_v,\Dcal_v)=([Q],\Ucal)$ as a complete binary tree, i.e., $\Tcal_v$ has depth $D=\lceil\log(Q)\rceil$ and $\Tcal_v$ has $2^i$ nodes at level $i$ for all $0\le i<\lceil\log(Q)\rceil$.

Given the state tensorizations, we obtain the tensorized atomic CSP $\Phi^\otimes=\pbra{Z,\pbra{\Omega_z,\Dcal_z}_{z\in Z},\Ccal^\otimes}$ for $\Phi=\Phi(H,Q)$.
By \Cref{def:hypergraph_coloring} and \Cref{prop:state_tensorization} it suffices to obtain $\sigma\sim\mu_\True^{\Ccal^\otimes}$.

Define $R=\lfloor\frac23\log(Q)\rfloor$. 
We construct the marking $\Mcal$ for $\Phi^\otimes$ by putting all internal nodes in $\Tcal_v$ of level at least $R$ into $\Mcal$ for each $v\in V$.
To apply \Cref{thm:binary_domains} to $(\Phi^\otimes,\Mcal)$, we compute the constants $\alpha,\beta,\lambda$ (See in \Cref{sec:binary_domains}) for $(\Phi^\otimes,\Mcal)$.

Fix any $C^\otimes\in\Ccal^\otimes$ and $v\in\vbl(C)$. Assume $\tpath(\sigma_\False^C(v),\Tcal_v)=\cbra{z_0,\ldots,z_\ell}$ where $z_0,\ldots,z_\ell$ is in the top-down order of $\Tcal_v$. 
Then $\ell\in\cbra{D-2,D-1}$ and $z_R,\ldots,z_\ell\in\Mcal$.
Since $\Tcal_v$ is a complete binary tree, by \Cref{eq:edge_weight} and the definition of $\Dcal_z$ in \Cref{def:tensorized_atomic_constraint_satisfaction_problem}, if $\ell\ge R$ then we have
$$
\prod_{i=R}^\ell\Dcal_{z_i}(\sigma_\False^{C^\otimes}(z_i))=\frac1{|\leafs(z_R)|}\le2^{-(\ell-R)}\le2^{-(D-2-R)}\le 4/Q^{1/3}.
$$
Thus 
\begin{equation}\label{eq:coloring_alpha}
\alpha(\Phi^\otimes,\Mcal,C^\otimes)\le\pbra{4/Q^{1/3}}^k.
\end{equation}
By Item (3) of \Cref{fct:tensorized_atomic_constraint_satisfaction_problem}, $\Delta(\Phi^\otimes)=\Delta(\Phi)=Q\cdot\Delta$. Since $(1-x)^{-1/x}\le4$ holds for any $x\in(0,1/2]$, we have
\begin{equation}\label{eq:coloring_beta}
\beta\le 4^{\Naturale\cdot\pbra{4/Q^{1/3}}^k\cdot Q\Delta}\le4^{\frac1{kD}}
\quad\text{if}\quad\Naturale\cdot\pbra{4/Q^{1/3}}^k\cdot Q\Delta\le\frac1{kD}.
\end{equation}
Meanwhile
$$
\prod_{i=0}^{R-1}\Dcal_{z_i}(\sigma_\False^{C^\otimes}(z_i))=\frac{|\leafs(z_R)|}Q\le\frac{2^{\ell-R+2}}Q\le\frac{2^{D-R+1}}Q\le8/Q^{2/3}.
$$
By Item (2) of \Cref{fct:tensorized_atomic_constraint_satisfaction_problem}, $k(\Phi^\otimes)\le k(\Phi)\cdot D=kD$.
Thus combining \Cref{eq:coloring_beta}, we have
\begin{equation}\label{eq:coloring_lambda}
\lambda(\Phi^\otimes,\Mcal,C^\otimes)\le k^2D^2\cdot4\cdot\pbra{8/Q^{2/3}}^k.
\end{equation}
Therefore, it suffices to make sure
$$
D-2\ge R,\quad
\Naturale\cdot\pbra{4/Q^{1/3}}^k\cdot Q\Delta\le\frac1{kD},\quad\text{and}\quad
k^2D^2\cdot4\cdot\pbra{8/Q^{2/3}}^k\cdot(Q\Delta)^2\le1/100.
$$
Since $D=\lceil\log(Q)\rceil$ and $R=\lfloor\frac23\log(Q)\rfloor$, it suffices to ensure
$$
Q\ge5\quad\text{and}\quad\Delta\le\pbra{\frac{Q^{1/3}}4}^k\cdot\frac1{40kQ\log(Q)}=Q^{\pbra{1/3-o_{Q,k}(1)}k}.
$$

Now we compute the running time. Note that the reduction to $\Phi^\otimes$ and the construction of $\Mcal$ only take $O(Q|V|)$ time.
Since $|Z|\le Q|V|$, $k(\Phi^\otimes)=O(k\log(Q))$, and $\Delta(\Phi^\otimes)=Q\Delta$, by \Cref{thm:binary_domains} the algorithm runs in time $O\pbra{k\Delta^3Q^4\log(Q)|V|\log(Q|V|)}$.
\end{proof}

\subsection{General Atomic Constraint Satisfaction Problem: Uniform Distribution}\label{sec:general_atomic_constraint_satisfaction_problem:uniform_distribution}

The analysis in \Cref{sec:general_atomic_constraint_satisfaction_problem:arbitrary_distribution} is purely a black-box reduction from large domains to binary domains: the construction of the marking does not use the fact that the CSP after reduction is actually a tensorized atomic CSP. 
Here we provide a unified construction for the state tensorizations and the marking when the underlying distributions are the uniform distribution.

Let $\Phi=\pbra{V,\pbra{\Omega_v,\Dcal_v}_{v\in V},\Ccal}$ be an atomic CSP such that $|\Omega_v|\ge2$ and $\Dcal_v$ is the uniform distribution for all $v\in V$. 
Our strategy is similar as before:
\begin{itemize}
\item[(1)] Construct state tensorizations $\pbra{\Tcal_v}_{v\in V}$ to obtain $\Phi^\otimes=\pbra{Z,\pbra{\Omega_z,\Dcal_z}_{z\in Z},\Ccal^\otimes}$ as the tensorized atomic CSP for $\Phi$. We will make sure each $\Omega_z$ is a binary domain, i.e., $Q(\Phi^\otimes)=2$, to apply \Cref{thm:binary_domains}.
\item[(2)] Construct a marking $\Mcal\subseteq Z$ to satisfy 
$$
\Naturale\cdot\alpha(\Phi^\otimes,\Mcal,C^\otimes)\cdot\Delta(\Phi^\otimes)\le1
\quad\text{and}\quad
\Delta(\Phi^\otimes)^2\cdot\lambda(\Phi^\otimes,\Mcal,C^\otimes)\le1/100
\quad\text{for all }C^\otimes\in\Ccal^\otimes,
$$
where we recall $\alpha,\beta,\lambda$ from \Cref{sec:binary_domains}.
\item[(3)] Apply \Cref{thm:binary_domains} to $\Phi^\otimes$ and $\Mcal$. By \Cref{prop:state_tensorization}, this provides a perfect uniform solution to $\Phi$ after performing $\Trans$ operation.
\end{itemize}
For each $C\in\Ccal$, define $p_C=\prod_{v\in\vbl(C)}\Dcal_v(\sigma_\False^C(v))=\prod_{v\in\vbl(C)}|\Omega_v|^{-1}$. 
For each $C^\otimes\in\Ccal^\otimes$, define $p_{C^\otimes}=\prod_{z\in\vbl(C^\otimes)}\Dcal_z(\sigma_\False^{C^\otimes}(z))$. 
To avoid confusion and by Item (3) of \Cref{fct:tensorized_atomic_constraint_satisfaction_problem}, we will use $\Delta$ for both $\Delta(\Phi^\otimes)$ and $\Delta(\Phi)$; use $p$ for both $p(\Phi^\otimes)$ and $p(\Phi)$, and use $p_C$ for both $p_C$ and $p_{C^\otimes}$.

We construct the state tensorization $\Tcal_v$ for each $(\Omega_v,\Dcal_v)$ in a random and independent way. The marking $\Mcal\subseteq Z$ will be constructed along with each $\Tcal_v$. This will be similar as \Cref{lem:binary_domains}: $\pbra{\Tcal_v}_{v\in V}$ and $\Mcal$ are constructed with high success probability using \Cref{thm:MT_LLL}.

Recall $\tpath(\cdot,\cdot)$ from \Cref{sec:large_domains:state_tensorization} and $z(\cdot)$ from \Cref{def:tensorized_atomic_constraint_satisfaction_problem}.
For each $v\in V$ and $q\in\Omega_v$, define 
$$
X(v,q)=\log\pbra{\prod_{z\in\tpath(q,\Tcal_v)\cap\Mcal}\Dcal_z(z(q))}.
$$
By \Cref{fct:edge_weight} and noticing $\Dcal_v$ is uniform, we know 
$$
X(v,q)=\log\pbra{|\Omega_v|^{-1}}-\log\pbra{\prod_{z\in\tpath(q,\Tcal_v)\setminus\Mcal}\Dcal_z(z(q))}.
$$

Recall the definition of $C^\otimes$ from \Cref{def:tensorized_atomic_constraint_satisfaction_problem}, then we have
\begin{equation}\label{eq:generalcsp_arbitrary_alpha}
\log\pbra{\alpha(\Phi^\otimes,\Mcal,C^\otimes)}=\log(p_C)-\sum_{v\in\vbl(C)}X(v,\sigma_\False^C(v))
\end{equation}
and $\beta=\beta(\Phi^\otimes,\Mcal)
\le\pbra{1-\Naturale\cdot\max_{C^\otimes\in\Ccal^\otimes}\alpha\pbra{\Phi^\otimes,\Mcal,C^\otimes}}^{-\Delta}$.
Then we also have
$$
\log\pbra{\lambda(\Phi^\otimes,\Mcal,C^\otimes)}\le\log\pbra{\abs{\vbl(C^\otimes)}^2\cdot\beta^{|\vbl(C^\otimes)|}}+\sum_{v\in\vbl(C)}X(v,\sigma_\False^C(v)).
$$

Let $\eta,\tau_1,\tau_2,\zeta\in(0,1)$ be parameters to be determined later. We will ensure $1-\eta-\tau_1\ge0$ and $\eta-\tau_2-3\zeta\ge0$.
For each $C\in\Ccal$, let $\Ecal_C^{(1)}$ be the event (i.e., constraint)
$$
\Ecal_C^{(1)}=\text{`` }
\pbra{\sum_{v\in\vbl(C)}X(v,\sigma_\False^C(v))}-\eta\log(p_C)<-\tau_1\log(1/p_C)
\text{ ''}
$$
and let $\Ecal_C^{(2)}$ be 
$$
\Ecal_C^{(2)}=\text{`` }
\pbra{\sum_{v\in\vbl(C)}X(v,\sigma_\False^C(v))}-\eta\log(p_C)>\tau_2\log(1/p_C)
\text{ ''}.
$$
Now we check $\Naturale\cdot\alpha(\Phi^\otimes,\Mcal,C^\otimes)\cdot\Delta\le1$ and $\Delta^2\cdot\lambda(\Phi^\otimes,\Mcal,C^\otimes)\le1/100$ assuming no $\Ecal_C^{(1)}$ or $\Ecal_C^{(2)}$ happens.
Then firstly $\alpha(\Phi^\otimes,\Mcal,C^\otimes)\le p_C^{1-\eta-\tau_1}\le p^{1-\eta-\tau_1}$.
Since $(1-x)^{-1/x}\le4$ holds for any $x\in(0,1/2]$, we have
$$
\beta\le 4^{\Naturale\cdot p^{1-\eta-\tau_1}\cdot\Delta}\le(3/2)^\zeta
\quad\text{if}\quad
2\Naturale\ln(2)\cdot p^{1-\eta-\tau_1}\cdot\Delta\le\zeta\cdot\ln(3/2).
$$
Note that $2\Naturale\ln(2)\cdot p^{1-\eta-\tau_1}\cdot\Delta\le\zeta\cdot\ln(3/2)$ already implies $\Naturale\cdot\alpha(\Phi^\otimes,\Mcal,C^\otimes)\cdot\Delta\le1$.

Now we turn to $\lambda$.
Since $x^\zeta\ge\zeta\ln(x)$ holds for any $x>0$, we have $\ln(1/p_C)\le p_C^{-\zeta}/\zeta$.
Our state tensorizations will have $\kappa(\Phi^\otimes)\le2$ (Recall $\kappa(\cdot)$ from \Cref{eq:kappa}), thus $p_C\le(2/3)^{|\vbl(C^\otimes)|}$ and
$$
\lambda(\Phi^\otimes,\Mcal,C^\otimes)
\le\abs{\vbl(C^\otimes)}^2\cdot\beta^{|\vbl(C^\otimes)|}p_C^{\eta-\tau_2}
\le\frac{\ln^2(1/p_C)\cdot p_C^{\eta-\tau_2-\zeta}}{\zeta^2\ln^2(3/2)}
\le\frac{p_C^{\eta-\tau_2-3\zeta}}{\zeta^2\ln^2(3/2)}
\le\frac{p^{\eta-\tau_2-3\zeta}}{\zeta^2\ln^2(3/2)}.
$$
Therefore $\Delta^2\cdot\lambda(\Phi^\otimes,\Mcal,C^\otimes)\le1/100$ is reduced to $\Delta^2\cdot p^{\eta-\tau_2-3\zeta}\le0.01\cdot\zeta^2\cdot\ln^2(3/2)$. 

In all, it suffices to make sure we can construct the state tensorizations and the marking to avoid all $\Ecal_C^{(1)}$ and $\Ecal_C^{(2)}$ with the following additional conditions:
\begin{equation}\label{eq:generalcsp_uniform_all}
\begin{gathered}
1-\eta-\tau_1\ge0,\quad\eta-\tau_2-3\zeta\ge0,\\
2\Naturale\ln(2)\cdot p^{1-\eta-\tau_1}\cdot\Delta\le\zeta\cdot\ln(3/2),
\quad\text{and}\quad
\Delta^2\cdot p^{\eta-\tau_2-3\zeta}\le0.01\cdot\zeta^2\cdot\ln^2(3/2).
\end{gathered}
\end{equation}
 
We first consider the simple case where all domains are binary to verify \Cref{rmk:binary_domains}.
\begin{lemma}\label{lem:generalcsp_uniform_simple}
Assume $|\Omega_v|=2$ for all $v\in V$. 
If $p^{0.175}\cdot\Delta\le10^{-7}$, then $\Phi^\otimes$ and $\Mcal$ can be constructed in time $O(k(\Phi)\cdot\Delta|V|)$ with success probability at least $0.99$ such that $Q(\Phi^\otimes)=2$, $\kappa(\Phi^\otimes)=1$, and 
$$
\Naturale\cdot\alpha(\Phi^\otimes,\Mcal,C^\otimes)\cdot\Delta\le1
\quad\text{and}\quad
\Delta^2\cdot\lambda(\Phi^\otimes,\Mcal,C^\otimes)\le1/100
\quad\text{for all }C^\otimes\in\Ccal^\otimes.
$$
\end{lemma}
\begin{proof}
Here we simply let each $\Tcal_v$ have one root with two child nodes. Then $\Phi^\otimes$ is essentially $\Phi$ itself.
We put each variable into $\Mcal$ with probability $\eta$ independently. 
Then for each $v\in V$ and $q\in\Omega_v$, we have
$$
X(v,q)=\log\pbra{\frac12}\cdot[v\in\Mcal]=\begin{cases}
-1 & \text{with probability }\eta,\\
0 & \text{with probability }1-\eta.
\end{cases}
$$
Hence for any $t\in\Rbb$, we have $\E\sbra{\Naturale^{t\cdot X(v,q)}}=1-\eta+\eta\cdot\Naturale^{-t}$.
Note that $p_C=2^{-|\vbl(C)|}$. 
Let $t_1,t_2>0$ be some parameters to be optimized soon. Then we have
\begin{align*}
\Pr\sbra{\Ecal_C^{(1)}}
&=\Pr\sbra{-t_1\cdot\sum_{v\in\vbl(C)}X(v,\sigma_\False^C(v))>(\eta+\tau_1)\cdot t_1\cdot|\vbl(C)|}\\
&=\Pr\sbra{\exp\cbra{-t_1\cdot\sum_{v\in\vbl(C)}X(v,\sigma_\False^C(v))}>\exp\cbra{(\eta+\tau_1)\cdot t_1\cdot|\vbl(C)|}}\\
&\le\E\sbra{\exp\cbra{-t_1\cdot\sum_{v\in\vbl(C)}X(v,\sigma_\False^C(v))}}\cdot\Naturale^{-(\eta+\tau_1)\cdot t_1\cdot|\vbl(C)|}
\tag{by Markov's inequality}\\
&=\pbra{\pbra{1-\eta+\eta\cdot\Naturale^{t_1}}\cdot\Naturale^{-(\eta+\tau_1)\cdot t_1}}^{|\vbl(C)|}
\tag{since $X(v,q)$'s are independent for different $v$}\\
&=\pbra{\pbra{\frac{1-\eta}{1-\eta-\tau_1}}^{1-\eta-\tau_1}\pbra{\frac\eta{\eta+\tau_1}}^{\eta+\tau_1}}^{|\vbl(C)|}
\tag{setting $\Naturale^{t_1}=\frac{(1-\eta)(\eta+\tau_1)}{\eta\cdot(1-\eta-\tau_1)}$}\\
&=2^{-|\vbl(C)|\cdot\KL(\eta+\tau_1\|\eta)}\le p^{\KL(\eta+\tau_1\|\eta)},
\tag{since $2^{-|\vbl(C)|}=p_C\le p$}
\end{align*}
where $\KL(a\|b)=a\log\pbra{\frac ab}+(1-a)\log\pbra{\frac{1-a}{1-b}}$ is the \emph{Kullback-Leibler divergence}.
Similarly
\begin{align*}
\Pr\sbra{\Ecal_C^{(2)}}
&=\Pr\sbra{t_2\cdot\sum_{v\in\vbl(C)}X(v,\sigma_\False^C(v))>-(\eta-\tau_2)\cdot t_2\cdot|\vbl(C)|}\\
&\le\pbra{\pbra{1-\eta+\eta\cdot\Naturale^{-t_2}}\cdot\Naturale^{(\eta-\tau_2)\cdot t_2}}^{|\vbl(C)|}\\
&=\pbra{\pbra{\frac{1-\eta}{1-\eta+\tau_2}}^{1-\eta+\tau_2}\pbra{\frac\eta{\eta-\tau_2}}^{\eta-\tau_2}}^{|\vbl(C)|}
\tag{setting $\Naturale^{-t_2}=\frac{(1-\eta)(\eta-\tau_2)}{\eta\cdot(1-\eta+\tau_2)}$}\\
&\le p^{\KL(\eta-\tau_2\|\eta)}.
\end{align*}
Let $\Ecal_C=\Ecal_C^{(1)}\lor\Ecal_C^{(2)}$.
Since each variable is put into $\Mcal$ independently, $\Ecal_C$ depends only on the constructions over $\vbl(C)$ and thus correlates with at most $\Delta$ many $\Ecal_{C'}$ (including itself) and
$$
\Pr\sbra{\Ecal_C}\le2\cdot p^{\min\cbra{\KL(\eta+\tau_1\|\eta),\KL(\eta-\tau_2\|\eta)}}.
$$
To apply \Cref{thm:MT_LLL}, it suffices to make sure
\begin{equation}\label{eq:generalcsp_uniform_binary}
2\Naturale\cdot p^{\min\cbra{\KL(\eta+\tau_1\|\eta),\KL(\eta-\tau_2\|\eta)}}\cdot\Delta\le1.
\end{equation}
Combining \Cref{eq:generalcsp_uniform_all} and \Cref{eq:generalcsp_uniform_binary}, we can pick $\eta,\tau_1,\tau_2,\zeta$ satisfying $1-\eta-\tau_1\ge0, \eta-\tau_2-3\zeta\ge0$; then let
$$
\gamma=\min\cbra{1-\eta-\tau_1,\frac{\eta-\tau_2-3\zeta}2,\KL(\eta+\tau_1\|\eta),\KL(\eta-\tau_2\|\eta)}
$$
and all the conditions boil down to $p^\gamma\cdot\Delta\le0.01\cdot\zeta$.

Numerically maximizing $\gamma$, we have $\gamma=0.175$, together with $\eta=0.595$, $\tau_1=0.23$, $\tau_2=0.245-3\cdot10^{-5}$, and $\zeta=10^{-5}$.
\end{proof}

Now we proceed to the general case. Restricted by the binary case, we cannot hope for a better bound than $p^{0.175}\cdot\Delta\lesssim1$. Therefore our goal is to show this bound is obtainable using the numerical constants $\eta,\tau_1,\tau_2,\zeta$ determined above.

\begin{lemma}\label{lem:generalcsp_uniform_full}
Assume $|\Omega_v|\ge2$ for all $v\in V$. 
If $p^{0.175}\cdot\Delta\le10^{-7}$, then $\Phi^\otimes$ and $\Mcal$ can be constructed in time $O(k(\Phi)Q(\Phi)\cdot\Delta|V|)$ with success probability at least $0.99$ such that $Q(\Phi^\otimes)=2$, $\kappa(\Phi^\otimes)\le2$, and 
$$
\Naturale\cdot\alpha(\Phi^\otimes,\Mcal,C^\otimes)\cdot\Delta\le1
\quad\text{and}\quad
\Delta^2\cdot\lambda(\Phi^\otimes,\Mcal,C^\otimes)\le1/100
\quad\text{for all }C^\otimes\in\Ccal^\otimes.
$$
\end{lemma}
\begin{proof}
Note that the choice of $\eta,\tau_1,\tau_2,\zeta$ above and the assumption $p^{0.175}\cdot\Delta\le10^{-7}$ already ensure \Cref{eq:generalcsp_uniform_all}, which guarantees $\Naturale\cdot\alpha(\Phi^\otimes,\Mcal,C^\otimes)\cdot\Delta\le1$ and $\Delta^2\cdot\lambda(\Phi^\otimes,\Mcal,C^\otimes)\le1/100$ for all $C^\otimes\in\Ccal^\otimes$. 
Thus in the following we only need to show how to construct the state tensorizations (and thus $\Phi^\otimes$) and the marking $\Mcal$ so that $Q(\Phi^\otimes)=2$, $\kappa(\Phi^\otimes)\le2$, and no $\Ecal_C^{(1)}$ or $\Ecal_C^{(2)}$ happens.

For each $C\in\Ccal$ and integer $m\ge2$, define $S(C,m)=\cbra{v\in\vbl(C)\mid|\Omega_v|=m}$. 
Then $p_C=\prod_{m=2}^{+\infty}m^{-|S(C,m)|}$.
Let $t_1,t_2>0$ be the constants used in the proof of \Cref{lem:generalcsp_uniform_simple}, i.e.,
$$
t_1=\ln\pbra{\frac{(1-\eta)(\eta+\tau_1)}{\eta\cdot(1-\eta-\tau_1)}}\in\sbra{1.1659,1.1660}
\quad\text{and}\quad
t_2=\ln\pbra{\frac{\eta\cdot(1-\eta+\tau_2)}{(1-\eta)(\eta-\tau_2)}}\in\sbra{1.0035,1.0036}.
$$

Our construction will satisfy the following proposition. We will provide its detail and proof soon.
The key part is Item (3) which intuitively says the simple binary case is actually the worst case.

\begin{proposition}\label{prop:binary_is_the_worst}
For each $v\in V$, the construction of the state tensorization $\Tcal_v$ and the marking on internal nodes of $\Tcal_v$ is randomized and satisfies the following properties.
\begin{itemize}
\item[(1)] For each $v\in V$, the construction is independent and takes $O(|\Omega_v|)=O(Q(\Phi))$ time.
\item[(2)] Each possible $\Tcal_v$ is a binary tree where $\frac{W(z\to z_1)}{W(z\to z_2)}\le2$ (Recall $W(\cdot)$ from \Cref{eq:edge_weight}) holds for any internal node $z\in\Tcal_v$ and $\cbra{z_1,z_2}=\childs(z)$.
\item[(3)] For each $v\in V$ and $q\in\Omega_v$, we have
$$
\E\sbra{\Naturale^{-t_1\cdot X(v,q)-(\eta+\tau_1)\cdot t_1\cdot\log(|\Omega_v|)}}\le|\Omega_v|^{-\KL(\eta+\tau_1\|\eta)}
$$
and
$$
\E\sbra{\Naturale^{t_2\cdot X(v,q)+(\eta-\tau_2)\cdot t_2\cdot\log(|\Omega_v|)}}\le|\Omega_v|^{-\KL(\eta-\tau_2\|\eta)}.
$$
\end{itemize}
\end{proposition}

We first finish the proof assuming \Cref{prop:binary_is_the_worst}. Note that both $Q(\Phi^\otimes)=2$ and $\kappa(\Phi^\otimes)\le2$ follow from Item (2) of \Cref{prop:binary_is_the_worst}. Thus we focus on how to make sure no $\Ecal_C^{(1)}$ or $\Ecal_C^{(2)}$ happens.
Similarly as in the proof of \Cref{lem:generalcsp_uniform_full}, we have
\begin{align*}
\Pr\sbra{\Ecal_C^{(1)}}
&=\Pr\sbra{-t_1\cdot\sum_{m=2}^{+\infty}\sum_{v\in S(C,m)}X(v,\sigma_\False^C(v))>(\eta+\tau_1)\cdot t_1\cdot\sum_{m=2}^{+\infty}|S(C,m)|\cdot\log(m)}\\
&\le\E\sbra{\exp\cbra{-t_1\cdot\sum_{m=2}^{+\infty}\sum_{v\in S(C,m)}X(v,\sigma_\False^C(v))-(\eta+\tau_1)\cdot t_1\cdot\sum_{m=2}^{+\infty}|S(C,m)|\cdot\log(m)}}\\
&=\prod_{m=2}^{+\infty}\prod_{v\in S(C,m)}\E\sbra{\Naturale^{-t_1\cdot X(v,\sigma_\False^C(v))-(\eta+\tau_1)\cdot t_1\cdot\log(m)}}
\tag{by Item (1) of \Cref{prop:binary_is_the_worst}}\\
&\le\prod_{m=2}^{+\infty}\prod_{v\in S(C,m)}m^{-\KL(\eta+\tau_1\|\eta)}
\tag{by Item (3) of \Cref{prop:binary_is_the_worst}}\\
&=p_C^{\KL(\eta+\tau_1\|\eta)}\le p^{\KL(\eta+\tau_1\|\eta)}.
\tag{since $p_C=\prod_{m=2}^{+\infty}m^{-|S(C,m)|}$}
\end{align*}
We also have
\begin{align*}
\Pr\sbra{\Ecal_C^{(2)}}
&=\Pr\sbra{t_2\cdot\sum_{m=2}^{+\infty}\sum_{v\in S(C,m)}X(v,\sigma_\False^C(v))>-(\eta-\tau_2)\cdot t_2\cdot\sum_{m=2}^{+\infty}|S(C,m)|\cdot\log(m)}\\
&\le\prod_{m=2}^{+\infty}\prod_{v\in S(C,m)}\E\sbra{\Naturale^{t_2\cdot X(v,\sigma_\False^C(v))+(\eta-\tau_2)\cdot t_2\cdot\log(m)}}\\
&\le p^{\KL(\eta-\tau_2\|\eta)}.
\end{align*}
Let $\Ecal_C=\Ecal_C^{(1)}\lor\Ecal_C^{(2)}$. 
By Item (1) of \Cref{prop:binary_is_the_worst}, the constructions are independent for each variable. Therefore $\Ecal_C$ depends only on the constructions over $\vbl(C)$ and thus correlates with at most $\Delta$ many $\Ecal_{C'}$ (including itself) and
$$
\Pr\sbra{\Ecal_C}\le2\cdot p^{\min\cbra{\KL(\eta+\tau_1\|eta),\KL(\eta-\tau_2\|\eta)}}\le2\cdot p^{0.175}.
$$
By our assumption $p^{0.175}\cdot\Delta\le10^{-7}$, we apply \Cref{thm:MT_LLL} to find a construction of the state tensorizations (and thus $\Phi^\otimes$) and the marking $\Mcal$ to avoid all $\Ecal_C=\Ecal_C^{(1)}\lor\Ecal_C^{(2)}$. The running time is $O(k(\Phi)\Delta|V|)\cdot O(Q(\Phi))$ where $O(Q(\Phi))$ is from Item (1) of \Cref{prop:binary_is_the_worst}.
\end{proof}

Putting everything together, we have the following corollary which justifies \Cref{rmk:generalcsp_arbitrary}.
\begin{corollary}\label{cor:generalcsp_uniform}
There exists a Las Vegas algorithm which takes as input an atomic CSP $\phi=\pbra{V,\pbra{\Omega_v,\Dcal_v}_{v\in V},\Ccal}$ such that the following holds.

If each $\Dcal_v$ is the uniform distribution and $p^{0.175}\cdot\Delta\le10^{-7}$, then the algorithm outputs a perfect uniform random solution of $\Phi$ in expected time $O\pbra{k\Delta^3Q^2|V|\log(Q|V|)}$.
\end{corollary}
\begin{proof}
By fixing the variable which has domain size $1$, we may safely assume $|\Omega_v|\ge2$ for all $v\in V$.

We keep using \Cref{lem:generalcsp_uniform_full} to construct $\Phi^\otimes=\pbra{Z,\pbra{\Omega_z,\Dcal_z}_{z\in Z},\Ccal^\otimes}$ and $\Mcal$ until they satisfy the conditions in \Cref{thm:binary_domains}. Then we run the algorithm in \Cref{thm:binary_domains} to obtain a random solution of $\Phi^\otimes$ distributed perfectly as $\mu_\True^{\Ccal^\otimes}$. Then by \Cref{prop:state_tensorization}, we perform $\Trans$ operation to obtain $\sigma^\Trans\sim\mu_\True^\Ccal$ which is just a perfect uniform random solution of $\Phi$.

To check the running time, by \Cref{fct:tensorized_atomic_constraint_satisfaction_problem} we have $\Delta(\Phi^\otimes)=\Delta(\Phi)$, $p(\Phi^\otimes)=p(\Phi)$, $k(\Phi^\otimes)\le k(\Phi)\cdot Q(\Phi)$\footnote{Actually we can upper bound $k(\Phi^\otimes)$ by $k(\Phi)\cdot O(\log(Q(\Phi)))$. This is because the state tensorizations here are all ``balanced'' binary trees of depth $O(\log(Q(\Phi)))$. However this only slightly improves the bound on the running time which is not our main focus.}, and $|Z|\le Q(\Phi)|V|$.
Therefore by \Cref{lem:generalcsp_uniform_full} and \Cref{thm:binary_domains}, it takes $O(kQ\Delta|V|)$ time in expectation to construct $\Phi^\otimes$ and $\Mcal$, and $O\pbra{k\Delta^3Q^2|V|\log(Q|V|)}$ time in expectation to do the sampling.
\end{proof}

Before proving \Cref{prop:binary_is_the_worst}, we set up some technical lemmas.
\begin{fact}[e.g., {\cite[Lemma 1 and Equation (4.16)]{hoeffding1994probability}}]\label{fct:sub-gaussian}
Let $X$ be an arbitrary random variable with $a\le X\le b$ almost surely. Let $c=\E[X]$. 
Then for any $t\in\Rbb$, we have
$$
\E\sbra{\Naturale^{t\cdot X}}\le\frac{b-c}{b-a}\cdot\Naturale^{a\cdot t}+\frac{c-a}{b-a}\cdot\Naturale^{b\cdot t}\le\Naturale^{\frac{(b-a)^2t^2}8}\cdot\Naturale^{c\cdot t}.
$$
\end{fact}
\begin{fact}\label{fct:numerical}
Our choice of $\eta,\tau_1,\tau_2,t_1,t_2$ satisfies
$$
\tfrac{\log^2(3)t_1^2}8\le\pbra{\tau_1t_1-\tfrac{\KL(\eta+\tau_1\|\eta)}{\log(\Naturale)}}\log(x)
\text{ and }
\tfrac{\log^2(3)t_2^2}8\le\pbra{\tau_2t_2-\tfrac{\KL(\eta-\tau_2\|\eta)}{\log(\Naturale)}}\log(x)
\text{ for all }x\ge8,
$$
and
$$
\tfrac{t_1^2}8\le\pbra{\tau_1t_1-\tfrac{\KL(\eta+\tau_1\|\eta)}{\log(\Naturale)}}\log(x)
\text{ and }
\tfrac{t_2^2}8\le\pbra{\tau_2t_2-\tfrac{\KL(\eta-\tau_2\|\eta)}{\log(\Naturale)}}\log(x)
\text{ for all }x\in\cbra{3,4,6,7},
$$
and
$$
\tfrac{\log^2(5/2)t_1^2}8\le\pbra{\tau_1t_1-\tfrac{\KL(\eta+\tau_1\|\eta)}{\log(\Naturale)}}\log(x)
\text{ and }
\tfrac{\log^2(5/2)t_2^2}8\le\pbra{\tau_2t_2-\tfrac{\KL(\eta-\tau_2\|\eta)}{\log(\Naturale)}}\log(x)
\text{ for }x=5.
$$
\end{fact}
\begin{proof}
Since $\log(x)$ is an increasing function for $x>0$, it only needs to verify the first two inequalities at $x=8$ and the middle two at $x=3$. Then all of them can be verified numerically.
\end{proof}

Now we give the proof of \Cref{prop:binary_is_the_worst}.
\begin{proof}[Proof of \Cref{prop:binary_is_the_worst}]
Our construction will be independent for each $v\in V$ and depend only on $|\Omega_v|$. Item (1) (2) will be evident as we describe the construction.

Let $N=|\Omega_v|\ge2$. 
The simplest case $N=2$ is essentially \Cref{lem:generalcsp_uniform_simple}: $\Tcal_v$ has one root with two leaf nodes, and we put the root into $\Mcal$ with probability $\eta$. By the choice of $t_1,t_2$ (See the calculation in \Cref{lem:generalcsp_uniform_simple}), the two inequalities in Item (3) are actually equality.

With foresight, our construction for $N\ge3$ will satisfy the following conditions.
\begin{itemize}
\item[(A)] $\E\sbra{X(v,q)}=\eta\log(1/N)$ for each $q\in\Omega_v$.
\item[(B)] If $N=5$, then $X(v,q)\in\sbra{a_N,b_N}$ always holds for all $q\in\Omega_v$ where $b_N-a_N\le\log(5/2)$.
\item[(C)] If $N\in\cbra{3,4,6,7}$, then $X(v,q)\in\sbra{a_N,b_N}$ always holds for all $q\in\Omega_v$ where $b_N-a_N\le1$.
\item[(D)] If $N\ge8$, then $X(v,q)\in\sbra{a_N,b_N}$ always holds for all $q\in\Omega_v$ where $b_N-a_N\le\log(3)$.
\end{itemize}
Then we verify Item (3) given Item (A) (B) (C) (D) here:
\begin{align*}
\E\sbra{\Naturale^{-t_1\cdot X(v,q)-(\eta+\tau_1)\cdot t_1\cdot\log(N)}}
&=\E\sbra{\Naturale^{-t_1\cdot\pbra{X(v,q)-\eta\log(1/N)}-\tau_1\cdot t_1\cdot\log(N)}}\\
&\le\exp\cbra{\tfrac{\pbra{b_N-a_N}^2t_1^2}8-\tau_1\cdot t_1\cdot\log(N)}
\tag{by \Cref{fct:sub-gaussian} and Item (A)}\\
&\le\exp\cbra{-\KL(\eta+\tau_1\|\eta)\cdot\tfrac{\log(N)}{\log(\Naturale)}}
=N^{-\KL(\eta+\tau_1\|\eta)}
\tag{by Item (B) (C) (D) and \Cref{fct:numerical}}
\end{align*}
and similarly
\begin{align*}
\E\sbra{\Naturale^{t_2\cdot X(v,q)+(\eta-\tau_2)\cdot t_2\cdot\log(N)}}
&=\E\sbra{\Naturale^{t_2\cdot\pbra{X(v,q)-\eta\log(1/N)}-\tau_2\cdot t_2\cdot\log(N)}}\\
&\le\exp\cbra{\tfrac{\pbra{b_N-a_N}^2t_2^2}8-\tau_2\cdot t_2\cdot\log(N)}
\tag{by \Cref{fct:sub-gaussian} and Item (A)}\\
&\le N^{-\KL(\eta-\tau_2\|\eta)}.
\tag{by Item (B) (C) (D) and \Cref{fct:numerical}}
\end{align*}

Now we give the construction for $N\ge3$ and check Item (A) (B) (C) (D).

\paragraph*{\fbox{$N=5$.}}
Let $p_5\in[0,1]$ be a constant to be determined. The construction for $N=5$ is as follows.
\begin{itemize}
\item In \Cref{fig:N=5}, we select the left tree with probability $p_5$ and the right with probability $1-p_5$.
\item After fixing the tree, we assign elements in $\Omega_v$ to the leaf nodes uniformly to obtain $\Tcal_v$.
\item Finally we put the internal nodes boxed by dashed squares into the marking $\Mcal$.
\end{itemize}

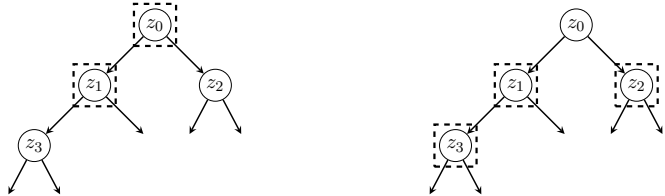
\begin{figure}[ht]
\centering
\scalebox{0.7}{
\begin{tikzpicture}[emptyC/.style={draw,circle,inner sep=2pt}, outE/.style={->,>=stealth,thick},node distance=1cm]
\node[emptyC] (v0) at (0,0) {$z_0$};
\node[draw,rectangle,dashed,minimum size=8mm,very thick] (vv0) at (v0) {};
\node[emptyC] (v1) [below left=of v0] {$z_1$};
\node[draw,rectangle,dashed,minimum size=8mm,very thick] (vv1) at (v1) {};
\node[emptyC] (v2) [below right=of v0] {$z_2$};
\node[emptyC] (v3) [below left=of v1] {$z_3$};

\node (q1) [below right=of v1] {};
\node (q2) [below left=of v2, xshift=5mm] {};
\node (q3) [below right=of v2, xshift=-5mm] {};
\node (q4) [below left=of v3, xshift=5mm] {};
\node (q5) [below right=of v3, xshift=-5mm] {};

\draw[outE] (v0) -- (v1);
\draw[outE] (v0) -- (v2);
\draw[outE] (v1) -- (v3);
\draw[outE] (v1) -- (q1);
\draw[outE] (v2) -- (q2);
\draw[outE] (v2) -- (q3);
\draw[outE] (v3) -- (q4);
\draw[outE] (v3) -- (q5);

\begin{scope}[xshift=8cm]
\node[emptyC] (u0) at (0,0) {$z_0$};
\node[emptyC] (u1) [below left=of u0] {$z_1$};
\node[draw,rectangle,dashed,minimum size=8mm,very thick] (uu1) at (u1) {};
\node[emptyC] (u2) [below right=of u0] {$z_2$};
\node[draw,rectangle,dashed,minimum size=8mm,very thick] (uu2) at (u2) {};
\node[emptyC] (u3) [below left=of u1] {$z_3$};
\node[draw,rectangle,dashed,minimum size=8mm,very thick] (uu3) at (u3) {};

\node (r1) [below right=of u1] {};
\node (r2) [below left=of u2, xshift=5mm] {};
\node (r3) [below right=of u2, xshift=-5mm] {};
\node (r4) [below left=of u3, xshift=5mm] {};
\node (r5) [below right=of u3, xshift=-5mm] {};

\draw[outE] (u0) -- (u1);
\draw[outE] (u0) -- (u2);
\draw[outE] (u1) -- (u3);
\draw[outE] (u1) -- (r1);
\draw[outE] (u2) -- (r2);
\draw[outE] (u2) -- (r3);
\draw[outE] (u3) -- (r4);
\draw[outE] (u3) -- (r5);
\end{scope}
\end{tikzpicture}}
\caption{The construction for $N=5$. Internal nodes boxed by dashed squares are put in $\Mcal$.}\label{fig:N=5}
\end{figure}

Recall $X(v,q)=\sum_{z\in\tpath(q,\Tcal_v)\cap\Mcal}\log(\Dcal_z(z(q)))$ where $z(q)\in\childs(z)$ is the child node of $z$ such that $q\in\leafs(z(q))$.
Since the leaf nodes are uniform, for any fixed $q\in\Omega_v$ we have
$$
\E\sbra{X(v,q)\mid\text{the left tree}}=\pbra{4\log(2/5)+\log(1/5)}/5<\eta\log(1/5)
$$
and
$$
\E\sbra{X(v,q)\mid\text{the right tree}}=\pbra{3\log(1/3)+2\log(1/2)}/5>\eta\log(1/5).
$$
Thus we can easily set $p_5$ to make sure 
$$
\E\sbra{X(v,q)}=p_5\cdot\E\sbra{X(v,q)\mid\text{the left tree}}+(1-p_5)\cdot\E\sbra{X(v,q)\mid\text{the right tree}}=\eta\log(1/5),
$$ 
which proves Item (A). As for Item (B), it suffices to observe 
$$
X(v,q)\in\cbra{\log(1/5),\log(1/3),\log(2/5),\log(1/2)}\subseteq[\log(1/5),\log(1/2)].
$$

\paragraph*{\fbox{$N\in\cbra{3,4,6,7}$.}}
The construction for $N\in\cbra{3,4,6,7}$ is the same as $N=5$ above except that we now use \Cref{fig:N=3467}: we will mix the left and right trees to make sure $\E\sbra{X(v,q)}=\eta\log(1/N)$ as demanded by Item (A). To do so it suffices to check $\eta\log(1/N)$ is sandwiched between $\E\sbra{X(v,q)\mid\text{the left tree}}$ and $\E\sbra{X(v,q)\mid\text{the right tree}}$.

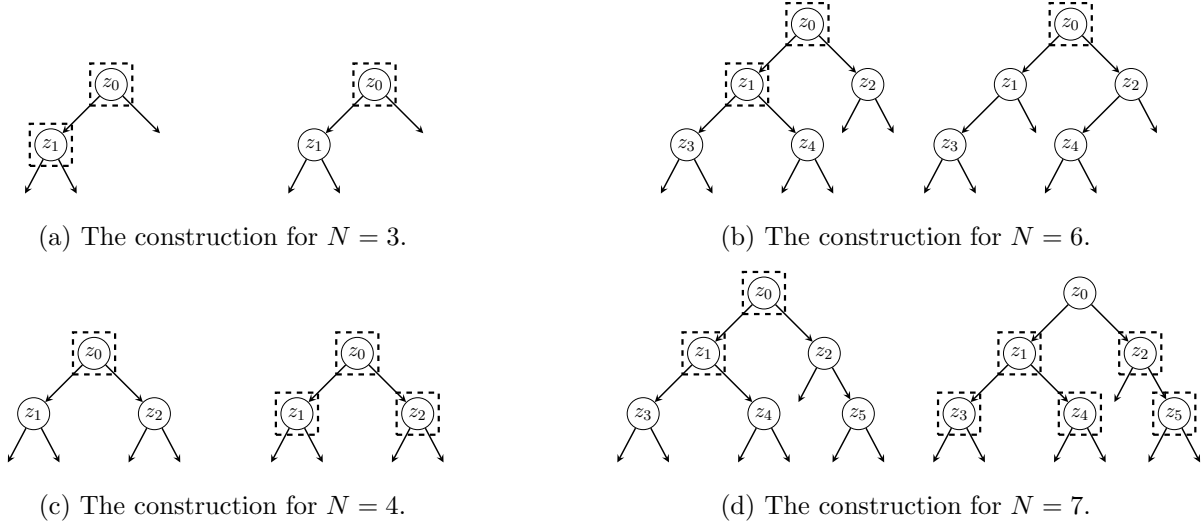
\begin{figure}[ht]
\centering
\begin{subfigure}[b]{0.4\textwidth}
\centering
\scalebox{0.7}{
\begin{tikzpicture}[emptyC/.style={draw,circle,inner sep=2pt}, outE/.style={->,>=stealth,thick},node distance=1cm]
\node[emptyC] (v0) at (0,0) {$z_0$};
\node[draw,rectangle,dashed,minimum size=8mm,very thick] (vv0) at (v0) {};
\node[emptyC] (v1) [below left=of v0] {$z_1$};
\node[draw,rectangle,dashed,minimum size=8mm,very thick] (vv1) at (v1) {};

\node (q1) [below right=of v0] {};
\node (q2) [below left=of v1, xshift=5mm] {};
\node (q3) [below right=of v1, xshift=-5mm] {};

\draw[outE] (v0) -- (v1);
\draw[outE] (v0) -- (q1);
\draw[outE] (v1) -- (q2);
\draw[outE] (v1) -- (q3);

\begin{scope}[xshift=5cm]
\node[emptyC] (u0) at (0,0) {$z_0$};
\node[draw,rectangle,dashed,minimum size=8mm,very thick] (uu0) at (u0) {};
\node[emptyC] (u1) [below left=of u0] {$z_1$};

\node (r1) [below right=of u0] {};
\node (r2) [below left=of u1, xshift=5mm] {};
\node (r3) [below right=of u1, xshift=-5mm] {};
\draw[outE] (u0) -- (u1);
\draw[outE] (u0) -- (r1);
\draw[outE] (u1) -- (r2);
\draw[outE] (u1) -- (r3);
\end{scope}
\end{tikzpicture}}
\caption{The construction for $N=3$.}\label{fig:N=3}
\end{subfigure}
\hfill
\begin{subfigure}[b]{0.5\textwidth}
\centering
\scalebox{0.7}{
\begin{tikzpicture}[emptyC/.style={draw,circle,inner sep=2pt}, outE/.style={->,>=stealth,thick},node distance=1cm]
\node[emptyC] (v0) at (0,0) {$z_0$};
\node[draw,rectangle,dashed,minimum size=8mm,very thick] (vv0) at (v0) {};
\node[emptyC] (v1) [below left=of v0] {$z_1$};
\node[draw,rectangle,dashed,minimum size=8mm,very thick] (vv1) at (v1) {};
\node[emptyC] (v2) [below right=of v0] {$z_2$};
\node[emptyC] (v3) [below left=of v1] {$z_3$};
\node[emptyC] (v4) [below right=of v1] {$z_4$};

\node (q1) [below left=of v2, xshift=5mm] {};
\node (q2) [below right=of v2, xshift=-5mm] {};
\node (q3) [below left=of v3, xshift=5mm] {};
\node (q4) [below right=of v3, xshift=-5mm] {};
\node (q5) [below left=of v4, xshift=5mm] {};
\node (q6) [below right=of v4, xshift=-5mm] {};

\draw[outE] (v0) -- (v1);
\draw[outE] (v0) -- (v2);
\draw[outE] (v1) -- (v3);
\draw[outE] (v1) -- (v4);
\draw[outE] (v2) -- (q1);
\draw[outE] (v2) -- (q2);
\draw[outE] (v3) -- (q3);
\draw[outE] (v3) -- (q4);
\draw[outE] (v4) -- (q5);
\draw[outE] (v4) -- (q6);

\begin{scope}[xshift=5cm]
\node[emptyC] (u0) at (0,0) {$z_0$};
\node[draw,rectangle,dashed,minimum size=8mm,very thick] (uu0) at (u0) {};
\node[emptyC] (u1) [below left=of u0] {$z_1$};
\node[emptyC] (u2) [below right=of u0] {$z_2$};
\node[emptyC] (u3) [below left=of u1] {$z_3$};
\node[emptyC] (u4) [below left=of u2] {$z_4$};

\node (r1) [below right=of u1, xshift=-5mm] {};
\node (r2) [below right=of u2, xshift=-5mm] {};
\node (r3) [below left=of u3, xshift=5mm] {};
\node (r4) [below right=of u3, xshift=-5mm] {};
\node (r5) [below left=of u4, xshift=5mm] {};
\node (r6) [below right=of u4, xshift=-5mm] {};

\draw[outE] (u0) -- (u1);
\draw[outE] (u0) -- (u2);
\draw[outE] (u1) -- (u3);
\draw[outE] (u1) -- (r1);
\draw[outE] (u2) -- (u4);
\draw[outE] (u2) -- (r2);
\draw[outE] (u3) -- (r3);
\draw[outE] (u3) -- (r4);
\draw[outE] (u4) -- (r5);
\draw[outE] (u4) -- (r6);
\end{scope}
\end{tikzpicture}}
\caption{The construction for $N=6$.}\label{fig:N=6}
\end{subfigure}\\\vspace{5pt}
\begin{subfigure}[b]{0.4\textwidth}
\centering
\scalebox{0.7}{
\begin{tikzpicture}[emptyC/.style={draw,circle,inner sep=2pt}, outE/.style={->,>=stealth,thick},node distance=1cm]
\node[emptyC] (v0) at (0,0) {$z_0$};
\node[draw,rectangle,dashed,minimum size=8mm,very thick] (vv0) at (v0) {};
\node[emptyC] (v1) [below left=of v0] {$z_1$};
\node[emptyC] (v2) [below right=of v0] {$z_2$};

\node (q1) [below left=of v1, xshift=5mm] {};
\node (q2) [below right=of v1, xshift=-5mm] {};
\node (q3) [below left=of v2, xshift=5mm] {};
\node (q4) [below right=of v2, xshift=-5mm] {};

\draw[outE] (v0) -- (v1);
\draw[outE] (v0) -- (v2);
\draw[outE] (v1) -- (q1);
\draw[outE] (v1) -- (q2);
\draw[outE] (v2) -- (q3);
\draw[outE] (v2) -- (q4);

\begin{scope}[xshift=5cm]
\node[emptyC] (u0) at (0,0) {$z_0$};
\node[draw,rectangle,dashed,minimum size=8mm,very thick] (uu0) at (u0) {};
\node[emptyC] (u1) [below left=of u0] {$z_1$};
\node[draw,rectangle,dashed,minimum size=8mm,very thick] (uu1) at (u1) {};
\node[emptyC] (u2) [below right=of u0] {$z_2$};
\node[draw,rectangle,dashed,minimum size=8mm,very thick] (uu2) at (u2) {};

\node (r1) [below left=of u1, xshift=5mm] {};
\node (r2) [below right=of u1, xshift=-5mm] {};
\node (r3) [below left=of u2, xshift=5mm] {};
\node (r4) [below right=of u2, xshift=-5mm] {};

\draw[outE] (u0) -- (u1);
\draw[outE] (u0) -- (u2);
\draw[outE] (u1) -- (r1);
\draw[outE] (u1) -- (r2);
\draw[outE] (u2) -- (r3);
\draw[outE] (u2) -- (r4);
\end{scope}
\end{tikzpicture}}
\caption{The construction for $N=4$.}\label{fig:N=4}
\end{subfigure}
\hfill
\begin{subfigure}[b]{0.5\textwidth}
\centering
\scalebox{0.7}{
\begin{tikzpicture}[emptyC/.style={draw,circle,inner sep=2pt}, outE/.style={->,>=stealth,thick},node distance=1cm]
\node[emptyC] (v0) at (0,0) {$z_0$};
\node[draw,rectangle,dashed,minimum size=8mm,very thick] (vv0) at (v0) {};
\node[emptyC] (v1) [below left=of v0] {$z_1$};
\node[draw,rectangle,dashed,minimum size=8mm,very thick] (vv1) at (v1) {};
\node[emptyC] (v2) [below right=of v0] {$z_2$};
\node[emptyC] (v3) [below left=of v1] {$z_3$};
\node[emptyC] (v4) [below right=of v1] {$z_4$};
\node[emptyC] (v5) [below right=of v2, xshift=-5mm] {$z_5$};

\node (q1) [below left=of v2, xshift=5mm] {};
\node (q2) [below left=of v3, xshift=5mm] {};
\node (q3) [below right=of v3, xshift=-5mm] {};
\node (q4) [below left=of v4, xshift=5mm] {};
\node (q5) [below right=of v4, xshift=-5mm] {};
\node (q6) [below left=of v5, xshift=5mm] {};
\node (q7) [below right=of v5, xshift=-5mm] {};

\draw[outE] (v0) -- (v1);
\draw[outE] (v0) -- (v2);
\draw[outE] (v1) -- (v3);
\draw[outE] (v1) -- (v4);
\draw[outE] (v2) -- (v5);
\draw[outE] (v2) -- (q1);
\draw[outE] (v3) -- (q2);
\draw[outE] (v3) -- (q3);
\draw[outE] (v4) -- (q4);
\draw[outE] (v4) -- (q5);
\draw[outE] (v5) -- (q6);
\draw[outE] (v5) -- (q7);

\begin{scope}[xshift=6cm]
\node[emptyC] (u0) at (0,0) {$z_0$};
\node[emptyC] (u1) [below left=of u0] {$z_1$};
\node[draw,rectangle,dashed,minimum size=8mm,very thick] (uu1) at (u1) {};
\node[emptyC] (u2) [below right=of u0] {$z_2$};
\node[draw,rectangle,dashed,minimum size=8mm,very thick] (uu2) at (u2) {};
\node[emptyC] (u3) [below left=of u1] {$z_3$};
\node[draw,rectangle,dashed,minimum size=8mm,very thick] (uu3) at (u3) {};
\node[emptyC] (u4) [below right=of u1] {$z_4$};
\node[draw,rectangle,dashed,minimum size=8mm,very thick] (uu4) at (u4) {};
\node[emptyC] (u5) [below right=of u2, xshift=-5mm] {$z_5$};
\node[draw,rectangle,dashed,minimum size=8mm,very thick] (uu5) at (u5) {};

\node (r1) [below left=of u2, xshift=5mm] {};
\node (r2) [below left=of u3, xshift=5mm] {};
\node (r3) [below right=of u3, xshift=-5mm] {};
\node (r4) [below left=of u4, xshift=5mm] {};
\node (r5) [below right=of u4, xshift=-5mm] {};
\node (r6) [below left=of u5, xshift=5mm] {};
\node (r7) [below right=of u5, xshift=-5mm] {};

\draw[outE] (u0) -- (u1);
\draw[outE] (u0) -- (u2);
\draw[outE] (u1) -- (u3);
\draw[outE] (u1) -- (u4);
\draw[outE] (u2) -- (u5);
\draw[outE] (u2) -- (r1);
\draw[outE] (u3) -- (r2);
\draw[outE] (u3) -- (r3);
\draw[outE] (u4) -- (r4);
\draw[outE] (u4) -- (r5);
\draw[outE] (u5) -- (r6);
\draw[outE] (u5) -- (r7);
\end{scope}
\end{tikzpicture}}
\caption{The construction for $N=7$.}\label{fig:N=7}
\end{subfigure}
\caption{The construction for $N\in\cbra{3,4,6,7}$.}\label{fig:N=3467}
\end{figure}

For $N=3$ we use \Cref{fig:N=3}. Then $X(v,q)\in\cbra{\log(1/3),\log(2/3)}$ which proves Item (C). Also
$$
\E\sbra{X(v,q)\mid\text{the left tree}}=\log(1/3)<\eta\log(1/3)
$$
and
$$
\E\sbra{X(v,q)\mid\text{the right tree}}=\pbra{\log(1/3)+2\log(2/3)}/3>\eta\log(1/3).
$$

For $N=4$ we use \Cref{fig:N=4}. Then $X(v,q)\in\cbra{\log(1/2),\log(1/4)}$ which proves Item (C). Also
$$
\E\sbra{X(v,q)\mid\text{the left tree}}=\log(1/2)>\eta\log(1/4)
$$
and
$$
\E\sbra{X(v,q)\mid\text{the right tree}}=\log(1/4)<\eta\log(1/4).
$$

For $N=6$ we use \Cref{fig:N=6}. Then $X(v,q)\in\cbra{\log(1/3),\log(1/2)}$ which proves Item (C). Also
$$
\E\sbra{X(v,q)\mid\text{the left tree}}=\log(1/3)<\eta\log(1/6)
$$
and
$$
\E\sbra{X(v,q)\mid\text{the right tree}}=\log(1/2)>\eta\log(1/6).
$$

For $N=7$ we use \Cref{fig:N=7}. Then $X(v,q)\in\cbra{\log(2/7),\log(3/7),\log(1/4),\log(1/3)}$ which proves Item (C). Also
$$
\E\sbra{X(v,q)\mid\text{the left tree}}=\pbra{4\log(2/7)+3\log(3/7)}/7>\eta\log(1/7)
$$
and
$$
\E\sbra{X(v,q)\mid\text{the right tree}}=\pbra{4\log(1/4)+3\log(1/3)}/7<\eta\log(1/7).
$$

\paragraph*{\fbox{$N\ge8$.}}
For each $x\in[N]$ define $A(N,x)=\lfloor N/x\rfloor\cdot(x+1)-N$ and $B(N,x)=N-\lfloor N/x\rfloor\cdot x$.
Let $R=\lfloor N^{1-\eta}\rfloor$. The proof of the following technical result is deferred to \Cref{sec:proof_of_fct:balanced_projection}.

\begin{fact}\label{fct:balanced_projection}
If $x\in\cbra{R-1,R,R+1}$, then $1\le x\le N$ and $A(N,x),B(N,x)\ge0$.
\end{fact}

For each $x\in\cbra{R-1,R,R+1}$ we have the following construction (See \Cref{fig:N>=8} for an intuition), the correctness of which is guaranteed by \Cref{fct:balanced_projection}.
\begin{itemize}
\item Let $\Tcal_1,\ldots,\Tcal_{A(N,x)}$ (resp., $\Tcal_{A(N,x)+1},\ldots,\Tcal_{A(N,x)+B(N,x)}$) be balanced binary trees that each has $x$ (resp., $x+1$) leaf nodes.\footnote{Here ``balanced'' simply means the sub-trees of sibling nodes have size difference at most $1$. In particular, this guarantees Item (2) of \Cref{prop:binary_is_the_worst}.}
\item Define distribution $\Dcal_x$ supported on $\cbra{\Tcal_1,\ldots,\Tcal_{A(N,x)+B(N,x)}}$ by setting
$$
\Dcal_x(\Tcal_i)=\begin{cases}
x/N & i\le A(N,x),\\
(x+1)/N & i>A(N,x).
\end{cases}
$$
Then construct a binary tree on top of $\cbra{\Tcal_1,\ldots,\Tcal_{A(N,x)+B(N,x)}}$ using \Cref{lem:construct_state_tensorization} and $\Dcal_x$.\footnote{Note that $\kappa(\Dcal_x)\le(x+1)/x\le2$. By Item (2) of \Cref{lem:construct_state_tensorization}, Item (2) of \Cref{prop:binary_is_the_worst} is still preserved. Meanwhile, since $A(N,x)+B(N,x)\le N/x=O(N^\eta)$, this step takes $O(N^\eta\log(N))=O(N)$ time.}
\item Assign elements in $\Omega_v$ uniformly to $A(N,x)\cdot x+B(N,x)\cdot(x+1)=N$ leaf nodes to get $\Tcal_v$.
\item Finally we put all the nodes on top of $\cbra{\Tcal_1,\ldots,\Tcal_{A(N,x)+B(N,x)}}$ into the marking $\Mcal$.
\end{itemize}

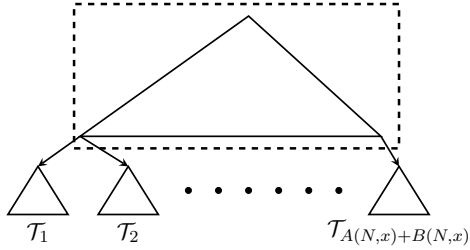
\begin{figure}[ht]
\centering
\scalebox{0.8}{
\begin{tikzpicture}[blackC/.style={draw,circle,inner sep=1pt,fill=black}, outE/.style={->,>=stealth,thick},node distance=1cm]
\draw[thick] (0,0) -- (1,0) -- (0.5,0.8) -- cycle;
\node (v0) at (0.5,-0.3) {$\Tcal_1$};
\draw[thick] (1.5,0) -- (2.5,0) -- (2,0.8) -- cycle;
\node (v1) at (2,-0.3) {$\Tcal_2$};
\node[blackC] (v2) at (3,0.4) {};
\node[blackC] (v3) at (3.5,0.4) {};
\node[blackC] (v4) at (4,0.4) {};
\node[blackC] (v5) at (4.5,0.4) {};
\node[blackC] (v6) at (5,0.4) {};
\node[blackC] (v7) at (5.5,0.4) {};
\draw[thick] (6,0) -- (7,0) -- (6.5,0.8) -- cycle;
\node (v8) at (6.5,-0.3) {$\Tcal_{A(N,x)+B(N,x)}$};

\draw[thick] (1.2,1.3) -- (6.2,1.3) -- (4,3.3) -- cycle;
\draw[very thick, dashed] (1.1,1.1) -- (6.5,1.1) -- (6.5,3.5) -- (1.1,3.5) -- cycle;
\draw[outE] (1.2,1.3) -- (0.5,0.8);
\draw[outE] (1.2,1.3) -- (2,0.8);
\draw[outE] (6.2,1.3) -- (6.5,0.8);
\end{tikzpicture}}
\caption{The construction for $N\ge8$ and $x\in\cbra{R-1,R,R+1}$. Nodes inside the dashed square are put in $\Mcal$.}\label{fig:N>=8}
\end{figure}

Observe that for any fixed $x\in\cbra{R-1,R,R+1}$ and any $q\in\Omega_v$, the construction gives $X(v,q)\in\cbra{\log(x/N),\log((x+1)/N)}$ and 
$$
\E\sbra{X(v,q)\mid x}=\frac{A(N,x)\cdot x}N\log\pbra{\frac xN}+\frac{B(N,x)\cdot(x+1)}N\log\pbra{\frac{x+1}N}\in\sbra{\log\pbra{\frac xN},\log\pbra{\frac {x+1}N}}.
$$
Therefore 
$$
\E\sbra{X(v,q)\mid x=R+1}\ge\log((R+1)/N)=\log\pbra{\pbra{\lfloor N^{1-\eta}\rfloor+1}/N}\ge\eta\log(1/N)
$$
and
$$
\E\sbra{X(v,q)\mid x=R-1}\le\log(R/N)=\log\pbra{\lfloor N^{1-\eta}\rfloor/N}\le\eta\log(1/N).
$$

Now we have two cases:
\begin{itemize}
\item If $\E\sbra{X(v,q)\mid x=R}\ge\eta\log(1/N)$\footnote{We emphasize that $\E\sbra{X(v,q)\mid x}$ has the same value for all $q\in\Omega_v$.}, we mix the construction of $x=R-1$ and $x=R$ to make sure $\E\sbra{X(v,q)}=\eta\log(1/N)$ for Item (A). Then $X(v,q)\in\cbra{\log\pbra{\frac{R-1}N},\log\pbra{\frac RN},\log\pbra{\frac{R+1}N}}$ which verifies Item (D) since $R\ge2$.
\item Otherwise, we mix the construction of $x=R$ and $x=R+1$ in the similar way. Then $X(v,q)\in\cbra{\log\pbra{\frac RN},\log\pbra{\frac{R+1}N},\log\pbra{\frac{R+2}N}}$ which also verifies Item (D).
\qedhere
\end{itemize}
\end{proof}

\section*{Acknowledgement}
KH wants to thank Weiming Feng for helpful discussion. 
KW wants to thank Chao Liao, Pinyan Lu, Jiaheng Wang, Kuan Yang, Yitong Yin, Chihao Zhang for helpful discussion on related topics in summer 2020. KW also wants to thank Kuan Yang for the help with \emph{Mathematica}. Finally we thank Weiming Feng, Chunyang Wang, and Kuan Yang for helpful comments on an earlier version of the paper.

\bibliographystyle{alphaurl} 
\bibliography{ref}

\appendix

\section[Proof of Fact 5.21]{Proof of \Cref{fct:balanced_projection}}\label{sec:proof_of_fct:balanced_projection}

Recall our setting: $N\ge8$ is an integer, $R=\lfloor N^{1-\eta}\rfloor$ where $\eta=0.595$, $A(N,x)=\lfloor N/x\rfloor\cdot(x+1)-N$, and $B(N,x)=N-\lfloor N/x\rfloor\cdot x$.
\begin{fact*}[\Cref{fct:balanced_projection} restated]
If $x\in\cbra{R-1,R,R+1}$, then $1\le x\le N$ and $A(N,x),B(N,x)\ge0$.
\end{fact*}
\begin{proof}
Since $1\le x\le N$ is equivalent to $2\le R\le N-1$, we only need to check
$$
2=\lfloor 8^{1-\eta}\rfloor\le\lfloor N^{1-\eta}\rfloor=R\le N^{1-\eta}\le\sqrt N\le N-1.
$$
Also $B(N,x)$ is always non-negative as $\lfloor N/x\rfloor\le N/x$. Hence we focus on the $A(N,x)\ge0$ part.

Since $\lfloor t\rfloor>t-1$, we have 
$$
A(N,x)>\pbra{\frac Nx-1}\cdot(x+1)-N=\frac Nx-x-1.
$$
Since $n>t$ is equivalent to $n\ge\lfloor t+1\rfloor$ for all integer $n$, we have
$$
A(N,x)\ge\left\lfloor\frac Nx-x\right\rfloor.
$$
Therefore if $x\le\sqrt N$, then $A(N,x)\ge0$. 
Since $R=\lfloor N^{1-\eta}\rfloor\le\sqrt N$, this shows $A(N,R-1)$ and $A(N,R)$ are all non-negative.
Now we deal with $A(N,R+1)$.
Note that $\lfloor N^{1-\eta}\rfloor\le N^{1-\eta}\le\sqrt N-1$ holds for all $N\ge18$, thus $R+1\le\sqrt N$ and $A(N,R+1)\ge0$ for all $N\ge18$. Finally for $8\le N\le17$, $A(N,R+1)\ge0$ can be verified numerically.
\end{proof}

\end{document}